\documentclass[12pt]{article}
\usepackage{amsmath,amssymb,amsthm,eucal,bm,cite,color}
\usepackage[top=2.5cm,right=2.5cm,left=2.5cm,bottom=3.1cm]{geometry}

\usepackage{mathptmx,eucal} %% Times-like typecafes + Euler script 

\usepackage{tikz}
\usetikzlibrary{calc}

\title{Improved Epstein--Glaser renormalization in $x$-space \\
versus differential renormalization}

\author{Jos\'e M. Gracia-Bond\'ia,$^{\!a,b,c}$
Heidy Guti\'errez$^c$
and Joseph C.~V\'arilly$^d$
\\[12pt]
{\footnotesize $^a$ Department of Theoretical Physics,
Universidad de Zaragoza, 50009 Zaragoza, Spain}\\[3pt]
{\footnotesize $^b$ BIFI Research Center,
Universidad de Zaragoza, 50018 Zaragoza, Spain}\\[3pt]
{\footnotesize $^c$ Department of Physics,
Universidad de Costa Rica, San Jos\'e 11501, Costa Rica}\\[3pt]
{\footnotesize $^d$ Department of Mathematics,
Universidad de Costa Rica, San Jos\'e 11501, Costa Rica}
}

\date{20 February 2017}

\newcommand{\al}{\alpha}           %% short for \alpha
\newcommand{\bt}{\beta}            %% short for \beta
\newcommand{\Dl}{\Delta}           %% propagator
\newcommand{\dl}{\delta}           %% Dirac's delta
\newcommand{\Ga}{\Gamma}           %% short for \Gamma
\newcommand{\ga}{\gamma}           %% short for \gamma
\newcommand{\La}{\Lambda}          %% short for \Lambda
\newcommand{\la}{\lambda}          %% short for \lambda
\newcommand{\Om}{\Omega}           %% short for \Omega
\newcommand{\om}{\omega}           %% short for \omega
\newcommand{\Sg}{\Sigma}           %% short for \Sigma
\newcommand{\sg}{\sigma}           %% short for \sigma
\renewcommand{\th}{\theta}         %% Heaviside function
\newcommand{\vf}{\varphi}          %% short for \varphi

\newcommand{\sB}{\mathcal{B}}      %% space of smooth functions
\newcommand{\sD}{\mathcal{D}}      %% space of test functions
\newcommand{\sF}{\mathcal{F}}      %% space of Feynman amplitudes
\newcommand{\sG}{\mathcal{G}}      %% vertex function
\newcommand{\sL}{\mathcal{L}}      %% internal lines
\newcommand{\sO}{\mathcal{O}}      %% slowly growing functions
\newcommand{\sS}{\mathcal{S}}      %% Schwartz functions
\newcommand{\sV}{\mathcal{V}}      %% set of (internal) vertices

\newcommand{\bL}{\mathbb{L}}       %% polylogarithm
\newcommand{\bR}{\mathbb{R}}       %% real numbers
\newcommand{\bS}{\mathbb{S}}       %% sphere

\newcommand{\FJL}{\mathrm{FJL}}    %% Freedman-Johnson-Latorre
\newcommand{\free}{\mathrm{free}}  %% free 2-point function
\newcommand{\id}{\mathrm{id}}      %% identity operator

\DeclareMathOperator*{\Res}{Res}   %% residue
\DeclareMathOperator*{\Vol}{Vol}   %% volume

\newcommand{\bull}{\mathord{\scriptstyle\bullet}} %% graph vertex
\newcommand{\del}{\partial}        %% short for \partial
\newcommand{\downto}{\downarrow}   %% right hand limit
\newcommand{\g}{\bar g}            %% redefined coupling
\newcommand{\Gbar}{{\mkern2mu \overline{\mkern-2mu G}}} %% `\bar G'
\newcommand{\less}{\setminus}      %% set difference
\newcommand{\ox}{\otimes}          %% tensor product
\newcommand{\Rbar}{{\mkern2mu \overline{\mkern-2mu R}}} %% `\bar R'
\newcommand{\upto}{\uparrow}       %% left hand limit
\newcommand{\wh}{\widehat}         %% short for \widehat
\newcommand{\x}{\times}            %% cartesian product or cross
\renewcommand{\.}{\cdot}           %% scalar product
\def\duo<#1,#2>{\langle#1,#2\rangle} %% pairing of distributions

\newcommand{\half}{{\mathchoice{\thalf}{\thalf}{\shalf}{\shalf}}}
\newcommand{\shalf}{{\scriptstyle\frac{1}{2}}} %% tiny fraction 1/2
\newcommand{\thalf}{\tfrac{1}{2}}  %% small fraction  1/2

\newcommand{\lddl}{l\frac{\partial}{\partial l}\,} %% l \del/\del l
\newcommand{\renormto}{\;\mathop{\longmapsto}\limits^{R}\;}

\newcommand{\pd}[2]{\frac{\partial#1}{\partial#2}} %% \del f/\del x
\newcommand{\set}[1]{\{\,#1\,\}}   %% set notation
\newcommand{\word}[1]{\quad\mbox{#1}\quad} %% well-spaced word(s)

\theoremstyle{plain}
\newtheorem{thm}{Theorem}           %% Theorem 1
\newtheorem{prop}[thm]{Proposition} %% Proposition 2
\newtheorem{lema}[thm]{Lemma}       %% Lemma 3

\theoremstyle{definition}
\newtheorem*{defn}{Definition}      %% Definition

\numberwithin{equation}{section}

\makeatletter
\renewcommand{\section}{\@startsection{section}{1}{\z@}%
                        {-3.5ex \@plus -1ex \@minus -.2ex}%
                        {2.3ex \@plus.2ex}%
                        {\normalfont\large\bfseries}}
\renewcommand{\subsection}{\@startsection{subsection}{2}{\z@}%
                        {-3.25ex \@plus -1ex \@minus -.2ex}%
                        {1.5ex \@plus .2ex}%
                        {\normalfont\normalsize\bfseries}}
\renewcommand{\@dotsep}{200} %% suppress dots in Contents
\makeatother

\newcommand{\cross}{\raisebox{-1mm}{%
\begin{tikzpicture}[scale=0.8]
\coordinate (A) at (0,0) ;
\draw (A) ++(-0.2,0.2) -- (A) -- ++(0.2,-0.2) ;
\draw (A) ++(-0.2,-0.2) -- (A) -- ++(0.2,0.2) ;
\draw (A) node{$\bull$} ;
\end{tikzpicture}}} %% cross diagram

\newcommand{\cruz}{{\!{\raisebox{-1mm}{%
\begin{tikzpicture}[scale=0.4]
\coordinate (A) at (0,0) ;
\draw (A) ++(-0.2,0.2) -- (A) -- ++(0.2,-0.2) ;
\draw (A) ++(-0.2,-0.2) -- (A) -- ++(0.2,0.2) ;
\draw (A) node{$.$} ;
\end{tikzpicture}}}\!}} %% little cross diagram, as subscript

\newcommand{\fish}{\raisebox{-1mm}{%
\begin{tikzpicture}[scale=0.8]
\coordinate (A) at (0,0) ; \coordinate (B) at (1,0) ; 
\draw (A) ++(-0.2,0.2) -- (A) -- ++(-0.2,-0.2) ;
\draw (A) parabola[bend pos=0.5] bend +(0,0.25) (B) ;
\draw (A) parabola[bend pos=0.5] bend +(0,-0.25) (B) ;
\draw (B) ++(0.2,0.2) -- (B) -- ++(0.2,-0.2) ;
\foreach \pt in {A,B} \draw (\pt) node{$\bull$} ;
\end{tikzpicture}}} %% fish diagram

\newcommand{\markedfish}{\raisebox{-2.6mm}{%
\begin{tikzpicture}[scale=1.6]
\coordinate (A) at (0,0) ; \coordinate (B) at (1,0) ; 
\draw (A) parabola[bend pos=0.5] bend +(0,0.25) (B) ;
\draw (A) parabola[bend pos=0.5] bend +(0,-0.25) (B) ;
\draw (A) ++(-0.2,0.2) -- (A) -- ++(-0.2,-0.2) ;
\draw (B) ++(0.2,0.2) -- (B) -- ++(0.2,-0.2) ;
\draw (A) node[above=1pt, left=5pt] {$0=x$} ;
\draw (B) node[below=1pt, right=5pt] {$z=y$} ;
\foreach \pt in {A,B} \draw (\pt) node{$\bull$} ;
\end{tikzpicture}}} %% marked fish diagram

\newcommand{\pez}{{\!{\raisebox{-1mm}{%
\begin{tikzpicture}[scale=0.4]
\coordinate (A) at (0,0) ; \coordinate (B) at (1,0) ; 
\draw (A) ++(-0.2,0.2) -- (A) -- ++(-0.2,-0.2) ;
\draw (A) parabola[bend pos=0.5] bend +(0,0.2) (B) ;
\draw (A) parabola[bend pos=0.5] bend +(0,-0.2) (B) ;
\draw (B) ++(0.2,0.2) -- (B) -- ++(0.2,-0.2) ;
\foreach \pt in {A,B} \draw (\pt) node{$.$} ;
\end{tikzpicture}}}\!}} %% little fish diagram, as subscript

\newcommand{\bikini}{\raisebox{-1mm}{%
\begin{tikzpicture}[scale=0.8]
\coordinate (A) at (0,0) ; \coordinate (B) at (1,0) ; 
\coordinate (C) at (2,0) ;
\draw (A) ++(-0.2,0.2) -- (A) -- ++(-0.2,-0.2) ;
\draw (A) parabola[bend pos=0.5] bend +(0,0.25) (B) ;
\draw (A) parabola[bend pos=0.5] bend +(0,-0.25) (B) ;
\draw (B) parabola[bend pos=0.5] bend +(0,0.25) (C) ;
\draw (B) parabola[bend pos=0.5] bend +(0,-0.25) (C) ;
\draw (C) ++(0.2,0.2) -- (C) -- ++(0.2,-0.2) ;
\foreach \pt in {A,B,C} \draw (\pt) node{$\bull$} ;
\end{tikzpicture}}} %% bikini diagram

\newcommand{\sosten}{{\!{\raisebox{-1mm}{%
\begin{tikzpicture}[scale=0.4]
\coordinate (A) at (0,0) ; \coordinate (B) at (1,0) ; 
\coordinate (C) at (2,0) ;
\draw (A) ++(-0.2,0.2) -- (A) -- ++(-0.2,-0.2) ;
\draw (A) parabola[bend pos=0.5] bend +(0,0.2) (B) ;
\draw (A) parabola[bend pos=0.5] bend +(0,-0.2) (B) ;
\draw (B) parabola[bend pos=0.5] bend +(0,0.2) (C) ;
\draw (B) parabola[bend pos=0.5] bend +(0,-0.2) (C) ;
\draw (C) ++(0.2,0.2) -- (C) -- ++(0.2,-0.2) ;
\foreach \pt in {A,B,C} \draw (\pt) node{$.$} ;
\end{tikzpicture}}}\!}} %% little bikini diagram, as subscript

\newcommand{\winecup}{\raisebox{-4mm}{%
\begin{tikzpicture}[scale=0.8]
\coordinate (A) at (0,0) ; \coordinate (B) at (1,0) ; 
\coordinate (C) at (0.5,-0.8) ;
\draw ($ (A)!-0.3!(C) $) -- ($ (A)!1.3!(C) $) ;
\draw ($ (B)!-0.3!(C) $) -- ($ (B)!1.3!(C) $) ;
\draw (A) parabola[bend pos=0.5] bend +(0,0.25) (B) ;
\draw (A) parabola[bend pos=0.5] bend +(0,-0.25) (B) ;
\foreach \pt in {A,B,C} \draw (\pt) node{$\bull$} ;
\end{tikzpicture}}} %% winecup diagram

\newcommand{\copadevino}{{\!{\raisebox{-2.2mm}{%
\begin{tikzpicture}[scale=0.4]
\coordinate (A) at (0,0) ; \coordinate (B) at (1,0) ; 
\coordinate (C) at (0.5,-0.8) ;
\draw ($ (A)!-0.3!(C) $) -- ($ (A)!1.3!(C) $) ;
\draw ($ (B)!-0.3!(C) $) -- ($ (B)!1.3!(C) $) ;
\draw (A) parabola[bend pos=0.5] bend +(0,0.2) (B) ;
\draw (A) parabola[bend pos=0.5] bend +(0,-0.2) (B) ;
\foreach \pt in {A,B,C} \draw (\pt) node{$.$} ;
\end{tikzpicture}}}\!}} %% little winecup diagram, as subscript

\newcommand{\trikini}{\raisebox{-1mm}{%
\begin{tikzpicture}[scale=0.8]
\coordinate (A) at (0,0) ; \coordinate (B) at (1,0) ; 
\coordinate (C) at (2,0) ; \coordinate (D) at (3,0) ; 
\draw (A) ++(-0.2,0.2) -- (A) -- ++(-0.2,-0.2) ;
\draw (A) parabola[bend pos=0.5] bend +(0,0.25) (B) ;
\draw (A) parabola[bend pos=0.5] bend +(0,-0.25) (B) ;
\draw (B) parabola[bend pos=0.5] bend +(0,0.25) (C) ;
\draw (B) parabola[bend pos=0.5] bend +(0,-0.25) (C) ;
\draw (C) parabola[bend pos=0.5] bend +(0,0.25) (D) ;
\draw (C) parabola[bend pos=0.5] bend +(0,-0.25) (D) ;
\draw (D) ++(0.2,0.2) -- (D) -- ++(0.2,-0.2) ;
\foreach \pt in {A,B,C,D} \draw (\pt) node{$\bull$} ;
\end{tikzpicture}}} %% trikini diagram

\newcommand{\trenza}{{\!{\raisebox{-1mm}{%
\begin{tikzpicture}[scale=0.4]
\coordinate (A) at (0,0) ; \coordinate (B) at (1,0) ; 
\coordinate (C) at (2,0) ; \coordinate (D) at (3,0) ; 
\draw (A) ++(-0.2,0.2) -- (A) -- ++(-0.2,-0.2) ;
\draw (A) parabola[bend pos=0.5] bend +(0,0.2) (B) ;
\draw (A) parabola[bend pos=0.5] bend +(0,-0.2) (B) ;
\draw (B) parabola[bend pos=0.5] bend +(0,0.2) (C) ;
\draw (B) parabola[bend pos=0.5] bend +(0,-0.2) (C) ;
\draw (C) parabola[bend pos=0.5] bend +(0,0.2) (D) ;
\draw (C) parabola[bend pos=0.5] bend +(0,-0.2) (D) ;
\draw (D) ++(0.2,0.2) -- (D) -- ++(0.2,-0.2) ;
\foreach \pt in {A,B,C,D} \draw (\pt) node{$.$} ;
\end{tikzpicture}}}\!}} %% little trikini diagram, as subscript

\newcommand{\stye}{\raisebox{-1.5mm}{%
\begin{tikzpicture}[scale=0.6]
\coordinate (A) at (0,0) ; \coordinate (B) at (2,0) ; 
\coordinate (C) at (0.7,0.35) ; \coordinate (D) at (1.3,0.35) ; 
\coordinate (E) at ($ (C)!0.5!(D) $) ; 
\draw (A) ++(-0.2,0.2) -- (A) -- ++(-0.2,-0.2) ;
\draw (A) parabola[bend pos=0.5] bend +(0,0.4) (B) ;
\draw (A) parabola[bend pos=0.5] bend +(0,-0.4) (B) ;
\draw (B) ++(0.2,0.2) -- (B) -- ++(0.2,-0.2) ;
\draw (E) circle(0.3cm) ;
\foreach \pt in {A,B,C,D} \draw (\pt) node{$\bull$} ;
\end{tikzpicture}}} %% stye diagram

\newcommand{\orzuelo}{{\!{\raisebox{-0.8mm}{%
\begin{tikzpicture}[scale=0.3]
\coordinate (A) at (0,0) ; \coordinate (B) at (2,0) ; 
\coordinate (C) at (0.7,0.35) ; \coordinate (D) at (1.3,0.35) ; 
\coordinate (E) at ($ (C)!0.5!(D) $) ; 
\draw (A) ++(-0.2,0.2) -- (A) -- ++(-0.2,-0.2) ;
\draw (A) parabola[bend pos=0.5] bend +(0,0.4) (B) ;
\draw (A) parabola[bend pos=0.5] bend +(0,-0.4) (B) ;
\draw (B) ++(0.2,0.2) -- (B) -- ++(0.2,-0.2) ;
\draw (E) circle(0.3cm) ;
\foreach \pt in {A,B,C,D} \draw (\pt) node{$.$} ;
\end{tikzpicture}}}\!}} %% little stye diagram, as subscript

\newcommand{\catseye}{\raisebox{-3.5mm}{%
\begin{tikzpicture}[scale=0.6]
\coordinate (A) at (0,0) ; \coordinate (B) at (2,0) ; 
\coordinate (C) at (1,0.4) ; \coordinate (D) at (1,-0.4) ; 
\draw (A) ++(-0.2,0.2) -- (A) -- ++(-0.2,-0.2) ;
\draw (A) parabola[bend pos=0.5] bend +(0,0.4) (B) ;
\draw (A) parabola[bend pos=0.5] bend +(0,-0.4) (B) ;
\begin{scope}[rotate=90]
\draw (C) parabola[bend pos=0.5] bend +(0,0.2) (D) ;
\draw (C) parabola[bend pos=0.5] bend +(0,-0.2) (D) ;
\end{scope}
\draw (B) ++(0.2,0.2) -- (B) -- ++(0.2,-0.2) ;
\foreach \pt in {A,B,C,D} \draw (\pt) node{$\bull$} ;
\end{tikzpicture}}} %% catseye diagram

\newcommand{\ojodegato}{{\!{\raisebox{-2mm}{%
\begin{tikzpicture}[scale=0.3]
\coordinate (A) at (0,0) ; \coordinate (B) at (2,0) ; 
\coordinate (C) at (1,0.4) ; \coordinate (D) at (1,-0.4) ; 
\draw (A) ++(-0.2,0.2) -- (A) -- ++(-0.2,-0.2) ;
\draw (A) parabola[bend pos=0.5] bend +(0,0.4) (B) ;
\draw (A) parabola[bend pos=0.5] bend +(0,-0.4) (B) ;
\begin{scope}[rotate=90]
\draw (C) parabola[bend pos=0.5] bend +(0,0.2) (D) ;
\draw (C) parabola[bend pos=0.5] bend +(0,-0.2) (D) ;
\end{scope}
\draw (B) ++(0.2,0.2) -- (B) -- ++(0.2,-0.2) ;
\foreach \pt in {A,B,C,D} \draw (\pt) node{$.$} ;
\end{tikzpicture}}}\!}} %% little catseye diagram, as subscript

\newcommand{\duncecap}{\raisebox{-3mm}{%
\begin{tikzpicture}[scale=0.6]
\coordinate (A) at (0,0) ; \coordinate (B) at (1,0) ; 
\coordinate (C) at (2,0) ; \coordinate (D) at (1,1.2) ; 
\draw ($ (A)!-0.25!(D) $) -- ($ (A)!1.25!(D) $) ;
\draw ($ (C)!-0.25!(D) $) -- ($ (C)!1.25!(D) $) ;
\draw (A) parabola[bend pos=0.5] bend +(0,0.25) (B) ;
\draw (A) parabola[bend pos=0.5] bend +(0,-0.25) (B) ;
\draw (B) parabola[bend pos=0.5] bend +(0,0.25) (C) ;
\draw (B) parabola[bend pos=0.5] bend +(0,-0.25) (C) ;
\foreach \pt in {A,B,C,D} \draw (\pt) node{$\bull$} ;
\end{tikzpicture}}} %% duncecap diagram

\newcommand{\casquete}{{\!\!{\raisebox{-1mm}{%
\begin{tikzpicture}[scale=0.3]
\coordinate (A) at (0,0) ; \coordinate (B) at (1,0) ; 
\coordinate (C) at (2,0) ; \coordinate (D) at (1,1.2) ; 
\draw ($ (A)!-0.25!(D) $) -- ($ (A)!1.25!(D) $) ;
\draw ($ (C)!-0.25!(D) $) -- ($ (C)!1.25!(D) $) ;
\draw (A) parabola[bend pos=0.5] bend +(0,0.2) (B) ;
\draw (A) parabola[bend pos=0.5] bend +(0,-0.2) (B) ;
\draw (B) parabola[bend pos=0.5] bend +(0,0.2) (C) ;
\draw (B) parabola[bend pos=0.5] bend +(0,-0.2) (C) ;
\foreach \pt in {A,B,C,D} \draw (\pt) node{$.$} ;
\end{tikzpicture}}}\!}} %% little duncecap diagram, as subscript

\newcommand{\kite}{\raisebox{-4mm}{%
\begin{tikzpicture}[scale=0.6]
\coordinate (A) at (1,1) ; \coordinate (B) at (2,1) ; 
\coordinate (C) at (1.5,0) ; \coordinate (D) at (0,0.5) ; 
\draw ($ (C)!1.3!(B) $) -- (C) -- (A) -- ($ (A)!1.3!(D) $) ;
\draw ($ (C)!1.2!(D) $) -- ($ (D)!1.2!(C) $) ;
\draw (A) parabola[bend pos=0.5] bend +(0,0.2) (B) ;
\draw (A) parabola[bend pos=0.5] bend +(0,-0.2) (B) ;
\foreach \pt in {A,B,C,D} \draw (\pt) node{$\bull$} ;
\end{tikzpicture}}} %% kite diagram

\newcommand{\cometa}{{\!\!\raisebox{-2mm}{%
\begin{tikzpicture}[scale=0.3]
\coordinate (A) at (1,1) ; \coordinate (B) at (2,1) ; 
\coordinate (C) at (1.5,0) ; \coordinate (D) at (0,0.5) ; 
\draw ($ (C)!1.3!(B) $) -- (C) -- (A) -- ($ (A)!1.3!(D) $) ;
\draw ($ (C)!1.2!(D) $) -- ($ (D)!1.2!(C) $) ;
\draw (A) parabola[bend pos=0.5] bend +(0,0.2) (B) ;
\draw (A) parabola[bend pos=0.5] bend +(0,-0.2) (B) ;
\foreach \pt in {A,B,C,D} \draw (\pt) node{$.$} ;
\end{tikzpicture}}\!}} %% little kite diagram, as subscript

\newcommand{\shark}{\raisebox{-4mm}{%
\begin{tikzpicture}[scale=0.8]
\coordinate (A) at (0,0) ; \coordinate (B) at (1,0) ; 
\coordinate (C) at (2,0.4) ; \coordinate (D) at (2,-0.4) ; 
\draw (A) ++(-0.2,0.2) -- (A) -- ++(-0.2,-0.2) ;
\draw (A) parabola[bend pos=0.5] bend +(0,0.25) (B) ;
\draw (A) parabola[bend pos=0.5] bend +(0,-0.25) (B) ;
\draw ($ (B)!1.3!(C) $) -- (B) -- ($ (B)!1.3!(D) $) ;
\begin{scope}[rotate=90]
\draw (C) parabola[bend pos=0.5] bend +(0,0.2) (D) ;
\draw (C) parabola[bend pos=0.5] bend +(0,-0.2) (D) ;
\end{scope}
\foreach \pt in {A,B,C,D} \draw (\pt) node{$\bull$} ;
\end{tikzpicture}}} %% shark diagram

\newcommand{\tiburon}{{\!\!\raisebox{-2mm}{%
\begin{tikzpicture}[scale=0.4]
\coordinate (A) at (0,0) ; \coordinate (B) at (1,0) ; 
\coordinate (C) at (2,0.4) ; \coordinate (D) at (2,-0.4) ; 
\draw (A) ++(-0.2,0.2) -- (A) -- ++(-0.2,-0.2) ;
\draw (A) parabola[bend pos=0.5] bend +(0,0.25) (B) ;
\draw (A) parabola[bend pos=0.5] bend +(0,-0.25) (B) ;
\draw ($ (B)!1.3!(C) $) -- (B) -- ($ (B)!1.3!(D) $) ;
\begin{scope}[rotate=90]
\draw (C) parabola[bend pos=0.5] bend +(0,0.2) (D) ;
\draw (C) parabola[bend pos=0.5] bend +(0,-0.2) (D) ;
\end{scope}
\foreach \pt in {A,B,C,D} \draw (\pt) node{$.$} ;
\end{tikzpicture}}\!}} %% shark diagram

\newcommand{\tetrahedron}{\raisebox{-6mm}{%
\begin{tikzpicture}[scale=0.4]
\coordinate (A) at (0,0) ; \coordinate (B) at (-45:1cm) ;
\coordinate (C) at (-5:2.4cm) ; \coordinate (D) at (50:2.2cm) ;
\draw ($ (A)!-0.25!(C) $) -- ($ (A)!1.25!(C) $) ;
\draw[line width=6pt, white] (B) -- (D) ;
\draw ($ (B)!-0.25!(D) $) -- ($ (B)!1.25!(D) $) ;
\draw (A) -- (B) -- (C) -- (D) -- cycle ;
\foreach \pt in {A,B,C,D} \draw (\pt) node{$\bull$} ;
\end{tikzpicture}}} %% tetrahedron diagram

\newcommand{\tetraedro}{{\mskip-9mu {\raisebox{1mm}{%
\begin{tikzpicture}[scale=0.2]
\coordinate (A) at (0,0) ; \coordinate (B) at (-45:1cm) ;
\coordinate (C) at (-5:2.4cm) ; \coordinate (D) at (50:2.2cm) ;
\draw ($ (A)!-0.25!(C) $) -- ($ (A)!1.25!(C) $) ;
\draw[line width=3pt, white] (B) -- (D) ;
\draw ($ (B)!-0.25!(D) $) -- ($ (B)!1.25!(D) $) ;
\draw (A) -- (B) -- (C) -- (D) -- cycle ;
\foreach \pt in {A,B,C,D} \draw (\pt) node{$.$} ;
\end{tikzpicture}}}\!\!}} %% little tetrahedron diagram, as subscript

\newcommand{\roll}{\raisebox{-3.5mm}{%
\begin{tikzpicture}[scale=0.5]
\coordinate (A) at (0,1) ; \coordinate (B) at (0,0) ; 
\coordinate (C) at (1.6,0) ; \coordinate (D) at (1.6,1) ; 
\draw ($ (A)!-0.5!(D) $) -- ($ (A)!1.5!(D) $) ; 
\draw ($ (B)!-0.5!(C) $) -- ($ (B)!1.5!(C) $) ; 
\draw ($ (A)!0.5!(B) $) circle(0.5cm) ;
\draw ($ (C)!0.5!(D) $) circle(0.5cm) ;
\foreach \pt in {A,B,C,D} \draw (\pt) node{$\bull$} ;
\end{tikzpicture}}} %% roll diagram

\newcommand{\canelon}{{\raisebox{-2mm}{%
\begin{tikzpicture}[scale=0.3]
\coordinate (A) at (0,1) ; \coordinate (B) at (0,0) ; 
\coordinate (C) at (1.6,0) ; \coordinate (D) at (1.6,1) ; 
\draw ($ (A)!-0.5!(D) $) -- ($ (A)!1.5!(D) $) ; 
\draw ($ (B)!-0.5!(C) $) -- ($ (B)!1.5!(C) $) ; 
\draw ($ (A)!0.5!(B) $) circle(0.5cm) ;
\draw ($ (C)!0.5!(D) $) circle(0.5cm) ;
\foreach \pt in {A,B,C,D} \draw (\pt) node{$.$} ;
\end{tikzpicture}}}} %% little roll diagram, as subscript

\newcommand{\justapoint}{\raisebox{-2mm}{%
\begin{tikzpicture}[scale=0.8]
\draw (-0.5,0) -- (0.5,0) ;
\draw (0,0) node{$\bull$} ;
\end{tikzpicture}}} %% point diagram

\newcommand{\punto}{{\!{\raisebox{-2mm}{%
\begin{tikzpicture}[scale=0.4]
\draw (-0.5,0) -- (0.5,0) ;
\draw (0,0) node{$.$} ;
\end{tikzpicture}}}\!}} %% little point diagram, as subscript

\newcommand{\sunset}{\raisebox{-0.4\height}{%
\begin{tikzpicture}[scale=0.8]
\coordinate (A) at (0,0) ; \coordinate (B) at (0.5,0) ;
\coordinate (C) at (1,0) ;
\draw ($ (A)!-0.3!(C) $) -- ($ (A)!1.3!(C) $) ;
\draw (B) circle(0.5cm) ;
\foreach \pt in {A,C} \draw (\pt) node{$\bull$} ;
\end{tikzpicture}}} %% sunset diagram

\newcommand{\ocaso}{{\!{\raisebox{-2.2mm}{%
\begin{tikzpicture}[scale=0.4]
\coordinate (A) at (0,0) ; \coordinate (B) at (0.5,0) ;
\coordinate (C) at (1,0) ;
\draw ($ (A)!-0.3!(C) $) -- ($ (A)!1.3!(C) $) ;
\draw (B) circle(0.5cm) ;
\foreach \pt in {A,C} \draw (\pt) node{$.$} ;
\end{tikzpicture}}}\!}} %% little sunset diagram, as subscript

\newcommand{\goggles}{\raisebox{-2mm}{%
\begin{tikzpicture}[scale=0.8]
\coordinate (A) at (0,0) ; \coordinate (B) at (1,0) ; 
\coordinate (C) at (2,0) ;
\draw (A) -- ++(-0.35,0) ;
\draw (A) parabola[bend pos=0.5] bend +(0,0.25) (B) ;
\draw (A) parabola[bend pos=0.5] bend +(0,-0.25) (B) ;
\draw (B) parabola[bend pos=0.5] bend +(0,0.25) (C) ;
\draw (B) parabola[bend pos=0.5] bend +(0,-0.25) (C) ;
\draw (A) parabola[bend pos=0.5] bend +(0,0.65) (C) ;
\draw (C) -- ++(0.35,0) ;
\foreach \pt in {A,B,C} \draw (\pt) node{$\bull$} ;
\end{tikzpicture}}} %% goggles diagram

\newcommand{\gafas}{{\!{\raisebox{-1mm}{%
\begin{tikzpicture}[scale=0.4]
\coordinate (A) at (0,0) ; \coordinate (B) at (1,0) ; 
\coordinate (C) at (2,0) ;
\draw (A) -- ++(-0.35,0) ;
\draw (A) parabola[bend pos=0.5] bend +(0,0.25) (B) ;
\draw (A) parabola[bend pos=0.5] bend +(0,-0.25) (B) ;
\draw (B) parabola[bend pos=0.5] bend +(0,0.25) (C) ;
\draw (B) parabola[bend pos=0.5] bend +(0,-0.25) (C) ;
\draw (A) parabola[bend pos=0.5] bend +(0,0.65) (C) ;
\draw (C) -- ++(0.35,0) ;
\foreach \pt in {A,B,C} \draw (\pt) node{$.$} ;
\end{tikzpicture}}}\!}} %% little goggles diagram, as subscript

\newcommand{\onions}{\raisebox{-2mm}{%
\begin{tikzpicture}[scale=0.75]
\coordinate (A) at (0,0) ; \coordinate (B) at (1,0) ; 
\coordinate (C) at (2,0) ; \coordinate (D) at (3,0) ; 
\draw (A) -- ++(-0.35,0) ;
\draw (A) parabola[bend pos=0.5] bend +(0,0.25) (B) ;
\draw (A) parabola[bend pos=0.5] bend +(0,-0.25) (B) ;
\draw (B) parabola[bend pos=0.5] bend +(0,0.25) (C) ;
\draw (B) parabola[bend pos=0.5] bend +(0,-0.25) (C) ;
\draw (C) parabola[bend pos=0.5] bend +(0,0.25) (D) ;
\draw (C) parabola[bend pos=0.5] bend +(0,-0.25) (D) ;
\draw (A) parabola[bend pos=0.5] bend +(0,0.65) (D) ;
\draw (D) -- ++(0.35,0) ;
\foreach \pt in {A,B,C,D} \draw (\pt) node{$\bull$} ;
\end{tikzpicture}}} %% onions diagram

\newcommand{\cebollas}{{\!{\raisebox{-1mm}{%
\begin{tikzpicture}[scale=0.4]
\coordinate (A) at (0,0) ; \coordinate (B) at (1,0) ; 
\coordinate (C) at (2,0) ; \coordinate (D) at (3,0) ; 
\draw (A) -- ++(-0.35,0) ;
\draw (A) parabola[bend pos=0.5] bend +(0,0.2) (B) ;
\draw (A) parabola[bend pos=0.5] bend +(0,-0.2) (B) ;
\draw (B) parabola[bend pos=0.5] bend +(0,0.2) (C) ;
\draw (B) parabola[bend pos=0.5] bend +(0,-0.2) (C) ;
\draw (C) parabola[bend pos=0.5] bend +(0,0.2) (D) ;
\draw (C) parabola[bend pos=0.5] bend +(0,-0.2) (D) ;
\draw (A) parabola[bend pos=0.5] bend +(0,0.65) (D) ;
\draw (D) -- ++(0.35,0) ;
\foreach \pt in {A,B,C,D} \draw (\pt) node{$.$} ;
\end{tikzpicture}}}\!}} %% little onions diagram, as subscript

\newcommand{\saturn}{\raisebox{-0.4\height}{%
\begin{tikzpicture}[scale=0.6]
\coordinate (A) at (0,0) ; \coordinate (B) at (0.5,0) ;
\coordinate (C) at (1,0) ; \coordinate (D) at (1.5,0) ;
\coordinate (E) at (2,0) ; 
\draw ($ (A)!-0.25!(E) $) -- ($ (A)!1.25!(E) $) ;
\draw (C) circle(1cm) ; \draw (C) circle(0.5cm) ;
\foreach \pt in {A,B,D,E} \draw (\pt) node{$\bull$} ;
\end{tikzpicture}}} %% saturn diagram

\newcommand{\saturno}{{\!{\raisebox{-0.4\height}{%
\begin{tikzpicture}[scale=0.3]
\coordinate (A) at (0,0) ; \coordinate (B) at (0.5,0) ;
\coordinate (C) at (1,0) ; \coordinate (D) at (1.5,0) ;
\coordinate (E) at (2,0) ; 
\draw ($ (A)!-0.25!(E) $) -- ($ (A)!1.25!(E) $) ;
\draw (C) circle(1cm) ; \draw (C) circle(0.5cm) ;
\foreach \pt in {A,B,D,E} \draw (\pt) node{$.$} ;
\end{tikzpicture}}}\!}} %% little saturn diagram, as subscript

\newcommand{\roach}{\raisebox{-6.5mm}{%
\begin{tikzpicture}[scale=0.7]
\coordinate (A) at (0,0) ; \coordinate (B) at (2,0) ; 
\coordinate (C) at (1,0) ; \coordinate (D) at (1,-0.8) ; 
\draw (A) -- (D) -- (B) ;
\draw (A) parabola[bend pos=0.5] bend +(0,0.6) (B) ;
\draw ($ (A)!-0.25!(B) $) -- ($ (A)!1.25!(B) $) ;
\begin{scope}[rotate=90]
\draw (C) parabola[bend pos=0.5] bend +(0,0.2) (D) ;
\draw (C) parabola[bend pos=0.5] bend +(0,-0.2) (D) ;
\end{scope}
\foreach \pt in {A,B,C,D} \draw (\pt) node{$\bull$} ;
\end{tikzpicture}}} %% roach diagram

\newcommand{\cuca}{{\!{\raisebox{-2mm}{%
\begin{tikzpicture}[scale=0.35]
\coordinate (A) at (0,0) ; \coordinate (B) at (2,0) ; 
\coordinate (C) at (1,0) ; \coordinate (D) at (1,-0.8) ; 
\draw (A) -- (D) -- (B) ;
\draw (A) parabola[bend pos=0.5] bend +(0,0.6) (B) ;
\draw ($ (A)!-0.25!(B) $) -- ($ (A)!1.25!(B) $) ;
\begin{scope}[rotate=90]
\draw (C) parabola[bend pos=0.5] bend +(0,0.2) (D) ;
\draw (C) parabola[bend pos=0.5] bend +(0,-0.2) (D) ;
\end{scope}
\foreach \pt in {A,B,C,D} \draw (\pt) node{$.$} ;
\end{tikzpicture}}}\!}} %% little roach diagram, as subscript

\newcommand{\snail}{\raisebox{-3.3mm}{%
\begin{tikzpicture}[scale=0.7]
\coordinate (A) at (0,0) ; \coordinate (B) at (1,0) ; 
\coordinate (C) at (2,0) ; \coordinate (D) at (3,0) ; 
\draw (A) -- ++(-0.5,0) ;
\draw (A) parabola[bend pos=0.5] bend +(0,0.25) (B) ;
\draw (A) parabola[bend pos=0.5] bend +(0,-0.25) (B) ;
\draw (C) parabola[bend pos=0.5] bend +(0,0.25) (D) ;
\draw (C) parabola[bend pos=0.5] bend +(0,-0.25) (D) ;
\draw (B) -- (C) ;
\draw (A) parabola[bend pos=0.5] bend +(0,0.6) (C) ;
\draw (B) parabola[bend pos=0.5] bend +(0,-0.6) (D) ;
\draw (D) -- ++(0.5,0) ;
\foreach \pt in {A,B,C,D} \draw (\pt) node{$\bull$} ;
\end{tikzpicture}}} %% snail diagram

\newcommand{\caracol}{{\!{\raisebox{-1mm}{%
\begin{tikzpicture}[scale=0.35]
\coordinate (A) at (0,0) ; \coordinate (B) at (1,0) ; 
\coordinate (C) at (2,0) ; \coordinate (D) at (3,0) ; 
\draw (A) -- ++(-0.5,0) ;
\draw (A) parabola[bend pos=0.5] bend +(0,0.25) (B) ;
\draw (A) parabola[bend pos=0.5] bend +(0,-0.25) (B) ;
\draw (C) parabola[bend pos=0.5] bend +(0,0.25) (D) ;
\draw (C) parabola[bend pos=0.5] bend +(0,-0.25) (D) ;
\draw (B) -- (C) ;
\draw (A) parabola[bend pos=0.5] bend +(0,0.6) (C) ;
\draw (B) parabola[bend pos=0.5] bend +(0,-0.6) (D) ;
\draw (D) -- ++(0.5,0) ;
\foreach \pt in {A,B,C,D} \draw (\pt) node{$.$} ;
\end{tikzpicture}}}}} %% little snail diagram, as subscript

\newcommand{\barbells}{\raisebox{-0.4\height}{%
\begin{tikzpicture}[scale=0.8]
\coordinate (A) at (-1.25,0) ; \coordinate (B) at (-0.75,0) ;
\coordinate (C) at (-0.25,0) ; \coordinate (D) at (0.25,0) ;
\coordinate (E) at (0.75,0) ; \coordinate (F) at (1.25,0) ; 
\draw ($ (A)!-0.12!(F) $) -- ($ (A)!1.12!(F) $) ;
\draw (B) circle(0.5cm) ; \draw (E) circle(0.5cm) ;
\foreach \pt in {A,C,D,F} \draw (\pt) node{$\bull$} ;
\end{tikzpicture}}} %% barbells diagram

\newcommand{\pesas}{{\!{\raisebox{-0.4\height}{%
\begin{tikzpicture}[scale=0.4]
\coordinate (A) at (-1.25,0) ; \coordinate (B) at (-0.75,0) ;
\coordinate (C) at (-0.25,0) ; \coordinate (D) at (0.25,0) ;
\coordinate (E) at (0.75,0) ; \coordinate (F) at (1.25,0) ; 
\draw ($ (A)!-0.12!(F) $) -- ($ (A)!1.12!(F) $) ;
\draw (B) circle(0.5cm) ; \draw (E) circle(0.5cm) ;
\foreach \pt in {A,C,D,F} \draw (\pt) node{$.$} ;
\end{tikzpicture}}}\!}} %% little barbells diagram, as subscript

\begin{document}

\maketitle

\begin{abstract}
Renormalization of massless Feynman amplitudes in $x$-space is
reexamined here, using almost exclusively real-variable methods. We
compute a wealth of concrete exam\-ples by means of recursive
extension of distributions. This allows us to show perturbative
expansions for the four-point and two-point functions at several loop
order. To deal with internal vertices, we expound and expand on
convolution theory for log-homogeneous distributions. The approach has
much in common with differential renormalization as given by Freedman,
Johnson and Latorre; but differs in important details.
\end{abstract}

% \paragraph{\mdseries\textit{Keywords}:}
% Differential renormalization; extension of distributions; vertex 
% functions; scaling. arXiv:1403.1785

% \S 1
\section{Introduction}
\label{sec:ad-altare}

The long-awaited publication of~\cite{NikolovST13} has again brought
to the fore renormalization of Feynman amplitudes in $x$-space. The
method in that reference is distribution-theoretical, in the spirit of
Epstein and Glaser~\cite{EpsteinG73}. That means cutoff- and
counterterm-free. That infinities never be met is something devoutly
to be wished, as regards the logical and mathematical status of
quantum field theory~\cite{Helling12}.

Twenty-odd years ago, an equally impressive paper~\cite{FreedmanJL92}
with the same general aim introduced to physicists a version of
differential renormalization in $x$-space. From the beginning, it must
have been obvious to the cognoscenti that this version and
Epstein--Glaser renormalization were two sides of the same coin. The
main aim of this article is to formalize this relation, to the
advantage of $x$-space renormalization in general.

Both~\cite{NikolovST13} and~\cite{FreedmanJL92} grant pride of place
to the massless Euclidean $\phi^4_4$ model, and it suits us to follow
them in that. Namely, we show how to compute amputated diagrams which
are proper (without cutlines), contributing to the four-point
functions for this model, up to the fourth order in the coupling
constant.

In \cite{NikolovST13} a recursive process to deal with
``subdivergences'' seeks to demonstrate the renormalization process as
a sequence of extensions of distributions. Since in the present paper
we are concerned with the two-point and four-point functions needed to
obtain the effective action, we furthermore need to integrate over
internal vertices of the graphs. Indeed, in \cite{FreedmanJL92},
internal vertices are integrated over, yielding convolution-like
integrals. However, some bogus justifications for this essentially
sound procedure are put forward there. Also, computations
in~\cite{FreedmanJL92} lack the natural algebraic rules set forth by
one of us in~\cite{Carme,Bettina} -- that \cite{NikolovST13} also
adopts and generalizes as the ``causal factorization property''
\cite[Thm.~2.1]{NikolovST13}.

For the Euclidean quantum amplitudes with which we deal in this paper,
we borrow the language of (divergent) subgraphs and
cographs~\cite{LangL83}. Let $\Ga(\sV,\sL)$ denote a graph one is
working with: $\sV$ is the set of its vertices and $\sL$ the set of
its internal lines. A subgraph $\ga \subseteq \Ga$ is a set of
vertices $\sV(\ga) \subseteq \sV$ and the set of \emph{all} lines
in~$\sL$ joining any two elements of~$\sV(\ga)$ -- which is to say, a
full subgraph in the usual mathematical parlance. Let $\ga$ be any of
the subgraphs. A lodestone is the rule contained in the (rigorous as
well as illuminating) treatment in \cite[Sect.~11.2]{KuznetsovTV96}.
It is written:
\begin{equation}
\duo< R[\Ga], \vf> = \duo< R[\ga], (\Ga/\ga)\,\vf>,
\label{eq:in-annum} % (1.1)
\end{equation}
where $R[\Ga]$, $R[\ga]$ and $\Ga/\ga$ denote corresponding
amplitudes, and $\vf$ is supported outside the singular points
of~$\Ga/\ga$.

As long as we need not integrate over internal vertices in~$\Ga$, this
is all we require for the recursive treatment of the hierarchy of
cographs in the diagrams. This rule implies the Euclidean version of
the causal factorization property of~\cite{NikolovST13}, as will be
thoroughly checked in the upcoming examples.

The treatment of diagrams with internal vertices calls for a
convolution-like machine. Thus we streamline the framework of
\cite{FreedmanJL92}, and in passing we correct some minor mistakes
there. While proceeding largely by way of example, along the way we
tune up that machine by (invoking and/or) proving a few rigorous
results.

The reader will notice how easy these computations in $x$-space are,
once the right methods are found. The outcome, we expect, is a
convincing case study for gathering the ways of \cite{NikolovST13}
and~\cite{FreedmanJL92} together.

\medskip

The plan of the paper is as follows. In Section~\ref{sec:get-ready} we
lay out the basic renormalization procedure, for divergent graphs
which are primitive (that is, without divergent subgraphs). The
mathematical task is to extend a function or distribution defined away
from the origin of $\bR^d$ to the whole of $\bR^d$. This is a
simplification of a more general extension problem of distributions
defined off a diagonal, calling on the translation invariance of the
distributions involved.

In Section~\ref{sec:convoluted}, adapting work by Horv\'ath, and later
by Ortner and Wagner, we adjust our convolu\-tion-like engine.
Horv\'ath's results have not found their way into textbooks, and seem
little known to physicists: it falls to us here to report on, and
complete them, in some detail.

The long Sections \ref{sec:tough-going} and~\ref{sec:graphic-arts}
deliver detailed and fully explicit calculations of concrete
(tadpole-part-free) graphs. The ``integration by parts'' method
of~\cite{FreedmanJL92} is put on a surer theoretical footing here by
linking it to the successive extensions in~\cite{NikolovST13} and the
``locality'' rule in~\cite{KuznetsovTV96}. In every case, differential
renormalization yields the leading term; here we do find the
correcting terms necessary to fulfil the algebraic strictures
in~\cite{Carme, NikolovST13}. The dilation behaviour of the
renormalized graphs is examined. We trust that these two sections give
a clear picture of the perturbative expansions for the four-point
function $\sG^4(x_1,x_2,x_3,x_4)$.

In Section~\ref{sec:feel-the-beat} we turn towards conceptual matters:
the renormalization group (RG) and the $\bt$-function, leading to
the ``main theorem of renormalization''~\cite{PopineauS82} and
Bogoliubov's recurrence relation at the level of the coupling constant
in this context. These are briefly discussed in the concluding
Section~\ref{sec:sic-transit}.

In Appendix~\ref{app:exten-formulas}, we collect for easy reference
explicit formulas for the distributional extensions in $x$-space
employed throughout. Graphs contributing to the two-point function
$\sG^2(x_1,x_2)$ are solved in Appendix~\ref{app:colon}. Calculations
of $p$-amplitudes are dealt with briefly in
Appendix~\ref{app:pspace-extens}.

\medskip

The relation between our scheme and dimensional regularization in
$x$-space was investigated in~\cite{Carme,Bettina} -- as ``analytical
prolongation'' -- and has recently been exhaustively
researched~\cite{Keller10, DuetschFKR13}. Reasons of spacetime prevent
us from going into that, for now; nor do we take up attending issues
of the Hopf algebra approach to the combinatorics of
renormalization~\cite{Kreimer98}.

% \S 2
\section{Primitive extensions of distributions} 
\label{sec:get-ready}

The reader is supposed familiar with the basics of distribution
theory; especially homogeneous distributions. Apart from this, the
article is self-contained, in that extension of homogeneous
distributions is performed from scratch. In that respect, outstanding
work in the eighties by Estrada and Kanwal~\cite{EstradaK85,
EstradaK89} has been very helpful to us.

\medskip

Ref.~\cite{NikolovST13} uses a complex-variable method for extension
of homogeneous distributions, following in the
main~\cite{Hormander90}. Basically, this exploits that Riesz's
normalized radial powers on~$\bR^d$, defined by
\begin{align*}
R_\la(r) &:= A_d(\la)\, r^\la, \word{where}
\\
A_d(\la) &:= \frac{2}{\Om_d\,\Ga\bigl(\frac{\la + d}{2}\bigr)}
\word{and} \Om_d = \Vol\bigl(\bS^{d-1}\bigr)
= \frac{2\pi^{d/2}}{\Ga(d/2)}\,,
\end{align*}
constitute an entire function of~$\la$. There is much to be said in
favour of such methods; the interested reader should consult also
\cite{GelfandS64} and~\cite{HorvathCol74}. Nevertheless, we choose to
recruit and popularize here real-variable methods. They are more in
the spirit both of the original Epstein--Glaser procedure
\cite{EpsteinG73} and differential renormalization
itself~\cite{FreedmanJL92}.

Let us call a homogeneous distribution $T$ on $\bR^d$ \emph{regular}
if it is smooth away from the origin; the smooth function on
$\bR^d \less \{0\}$ associated to it is homogeneous of the same
degree. Homogeneous distributions of all kinds are tempered (see the
discussion in Section~\ref{ssc:log-homog}), and thus possess Fourier
transforms.

Consider first spaces of homogeneous distributions on the real
half-line. The function $r^{-1}$ defines a distribution on the space%
\footnote{It is not satisfactory for us to consider extensions merely
from $\sS(\bR^+ \less \{0\})$ to~$\sS(\bR^+)$.}
of Schwartz functions vanishing at the origin, $r\,\sS(\bR^+)$. It
seems entirely natural to extend it to a functional on the whole space
$\sS(\bR^+)$ by defining
\begin{equation}
r_1[r^{-1}] := \Bigl( \log \frac{r}{l} \Bigr)'
\label{eq:the-power-of-logarithm} % (2.1)
\end{equation}
where $l$ is a convenient scale. Note that $\log(r/l)$ is a well
defined distribution, and so is its distributional derivative. The
difference between two versions of this recipe, with different scales,
lies in the kernel of the map $f \mapsto rf$ on distributions, i.e.,
it is a multiple of the delta function. Of course $r_1[r^{-1}]$ is no
longer homogeneous, since
\begin{equation}
r_1[(\la r)^{-1}] 
= \la^{-1} \,r_1[r^{-1}] + \la^{-1} \log\la \,\dl(r).
\label{eq:first-log} % (2.2)
\end{equation}

Now, for $z$ not a negative integer, the property
\begin{align}
r\,r^z &= r^{z+1}
\label{eq:mult-property} % (2.3)
\word{holds; and also}
\\
\frac{d}{dr}(r^z) &= z\,r^{z-1}.
\label{eq:diff-property} % (2.4)
\end{align}
For the homogeneous functions $r^{-n}$ with $n = 2,3,\dots$ which are
not locally integrable, one might adopt the recipe:
$$
r_1[r^{-n}]
:= \frac{(-)^{n-1}}{(n - 1)!} \Bigl( \log \frac{r}{l} \Bigr)^{(n)},
$$
generalizing \eqref{eq:diff-property} by definition. This is
\textit{differential renormalization on} $\bR^+$ in a nutshell. That,
however, loses property~\eqref{eq:mult-property}.

Thus we look for a recipe respecting~\eqref{eq:mult-property} instead.
Let $f$ denote a smooth function on $\bR^+ \less \{0\}$ with
$f(r) = O(r^{-k-1})$ as $r \downto 0$. Epstein and Glaser introduce a
general subtraction projection $W_w$ from $\sS(\bR^+)$ to the space of
test functions vanishing at order~$k$ at the origin, whereby the whole
$k$-jet of a test function $\vf$ on~$\bR^d$ at the origin
$$
j^k_0(\vf)(x) := \vf(0) + \sum_{|\al|=1} \frac{x^\al}{\al!}
\,\vf^{(\al)}(0) +\cdots+ \sum_{|\al|=k} \frac{x^\al}{\al!}
\,\vf^{(\al)} (0)
$$
is weighted by an infrared regulator~$w$, satisfying $w(0) = 1$ and
$w^{(\al)}(0) = 0$ for $1 \leq |\al| \leq k$:
$$
W_w\vf(x) := \vf(x) - w(x) j^k_0(\vf)(x).
$$
One may use instead \cite{Carme,Bettina} the simpler subtraction
projection:
\begin{equation}
P_w\vf(x) :=  \vf(x) - j^{k-1}_0(\vf)(x)
- w(x) \sum_{|\al|= k} \frac{x^{\al}}{\al!} \,\vf^{(\al)}(0).
\label{eq:take-away} % (2.5)
\end{equation}
Just $w(0) = 1$ is now required from~$w$ for the projection
property $P_w(P_w\vf) = P_w\vf$ to hold.

Define now the operations $W_w$ and $P_w$ on $\sS'(\bR^n)$ by duality:
$$
\duo< W_w f,\vf> := \duo< f, W_w\vf>  \word{and likewise}
\duo< P_w f,\vf> := \duo< f, P_w\vf>.
$$
On the half-line, by use of Lagrange's expression for MacLaurin
remainders, for $k = 0$ (logarithmic divergence) we arrive in a short
step~\cite{Carme} at the dual integral formula:
$$
W_w f(r) = P_w f(r) 
= - \frac{d}{dr} \biggl[ r \int_0^1 \frac{dt}{t^2}\, 
f\Bigl( \frac{r}{t} \Bigr) w\Bigl( \frac{r}{t} \Bigr) \biggr].
$$
We choose (and always use henceforth) the simple regulator
\begin{equation}
w(r) := \th(l - r), \quad \text{for some fixed } l > 0,
\label{eq:simple-scale} % (2.6)
\end{equation}
with $\th$ being the Heaviside function. Actually, for $k = 0$ this
yields the general result, since the difference between two extensions
with acceptable dilation behaviour is a multiple of the delta
function. In the homogeneous case, one immediately
recovers~\eqref{eq:the-power-of-logarithm}:
\begin{align*}
R_1[r^{-1}] := P_{\th(l-\cdot)}[r^{-1}]
= -\biggl[ \int_{r/l}^1 \frac{dt}{t} \biggr]' 
= \biggl( \log\frac{r}{l} \biggr)' =: r_1[r^{-1}].
\end{align*}

For any positive integer $k$,
\begin{align}
P_w f(r) &= (-)^k\,k \biggl[ \frac{r^k}{k!} \int_0^1 
\frac{(1 - t)^{k-1}}{t^{k+1}} f\Bigl( \frac{r}{t} \Bigr)
\biggl(1 - w\Bigl( \frac{r}{t} \Bigr) \biggr) \, dt \biggr]^{(k)}
\nonumber \\
&\quad + (-)^{k+1} (k + 1) \biggl[ \frac{r^{k+1}}{(k+1)!} \int_0^1
\frac{(1 - t)^k}{t^{k+2}} f\Bigl( \frac{r}{t} \Bigr) w\Bigl(
\frac{r}{t} \Bigr) \,dt \biggr]^{(k+1)}
\label{eq:proj-one-dim} % (2.7)
\end{align}
which yields
\begin{equation}
R_1[r^{-k-1}] := P_{\th(l-\cdot)}[r^{-k-1}]
= \frac{(-)^k}{k!} \biggl[ \Bigl( \log\frac{r}{l} \Bigr)^{(k+1)}
+ H_k \,\dl^{(k)}(r) \biggr].
\label{eq:powers-tamed} % (2.8)
\end{equation}
Here $H_k$ is the $k$-th harmonic number:
$$
H_k := \sum_{j=1}^k \frac{(-)^{j+1}}j \binom{k}{j} = \sum_{j=1}^k
\frac{1}{j} \,; \word{and} H_0 := 0.
$$ 
(See \cite[p.~267]{GrahamKP89} for the equality of the two sums.) Note
that $R_1[r^{-k-1}] \neq W_{\th(l-\cdot)}[r^{-k-1}]$ for $k > 0$, as
well as $R_1[r^{-k-1}] \neq r_1[r^{-k-1}]$.

\medskip

One of us in \cite{Carme} proved that:
\begin{itemize}
\item
$R_1$ coincides with (a straightforward generalization of) Hadamard's
finite part extension and the meromorphic continuation extension of
\cite{NikolovST13,Hormander90,GelfandS64,HorvathCol74}.
\item
$R_1$ (but not in general the $W_w$ subtraction) fulfils the
\textit{algebra property}
$$
r^m\,R_1[r^{-k-1}] = R_1[r^{-k+m-1}],
$$
extending~\eqref{eq:mult-property} to the realm of renormalized 
distributions.%
\footnote{The algebra property can be in contradiction with arbitrary
``renormalization prescriptions''~\cite{Falk05}; but this does not
detract from its utility.}
\end{itemize}

% \S 2.1
\subsection{Dimensional reduction}
\label{ssc:diml-red}

The task is now to extend radial-power distributions, that is, to
compute $\duo< f(r), \vf(x)>$, where~$f(r)$ denotes a scalar, radially
symmetric distribution defined on~$\bR^d$. We keep borrowing from
Estrada and Kanwal \cite{EstradaK85,EstradaK89}. First sum over all
angles, by defining
$$
\Pi\vf(r) := \int_{|\om|=1} \vf(r\om) \,d\sg(\om).
$$
The resulting function $\Pi\vf$ is to be regarded either as defined
on~$\bR^+$, or as an even function on the whole real line. Its
derivatives of odd order with respect to~$r$ at~$0$ vanish, and those
of even order satisfy:
$$
(\Pi\vf)^{(2l)}(0) = \Om_{d,l}\,\Dl^l\vf(0)
:= \biggl( \int_{|\om|=1} x_i^{2l} \biggr) \Dl^l\vf(0)
= \frac{2\,\Ga(l + \half)\,\pi^{(d-1)/2}}
{\Ga(l + \frac{d}{2})}\,\Dl^l\vf(0).
$$
With that in mind, one can write the \textit{dimensional reduction}
formula:
\begin{equation}
\duo< R_d[f(r)], \vf(x) >_{\bR^d}
= \bigl< R_1[f(r)r^{d-1}], \Pi\vf(r)\bigr>_{\bR^+}.
\label{eq:reductio-ad-unum} % (2.9)
\end{equation}
The notation $R_d[f(r)]$ for the renormalized object handily keeps
track of the space dimension. The formula can be taken as a
\textit{definition} of $R_d[f(r)]$, and so for radially symmetric
distributions the simple $R_1$ method, as well as Hadamard's and
meromorphic continuation on the real line, are lifted in tandem to
higher dimensions by~\eqref{eq:reductio-ad-unum}. As a bonus, the
multiplicativity condition for radial functions is automatically
preserved.

Keep also in mind, however, that Epstein--Glaser-type subtraction
works in any number of dimensions. In particular, our modified
Epstein--Glaser method for $f(r) = O(r^{-k-d})$ leads to the integral
form, generalizing~\eqref{eq:proj-one-dim},
\begin{align}
P_w f(x) &= (-)^k k
\sum_{|\al|=k} \del^\al \biggl[ \frac{x^\al}{\al!}
\int_0^1 dt\, \frac{(1-t)^{k-1}}{t^{k+d}} f\Bigl( \frac{x}{t} \Bigr)
\Bigl( 1 - w\Bigl( \frac{x}{t} \Bigr) \Bigr) \biggr]
\nonumber \\
&\qquad + (-)^{k+1} (k+1) \sum_{|\bt|=k+1}
\del^\bt \biggl[ \frac{x^\bt}{\bt!} \int_0^1 dt\,
\frac{(1-t)^k}{t^{k+d+1}} f\Bigl( \frac{x}{t} \Bigr)
w\Bigl( \frac{x}{t} \Bigr) \biggr];
\label{eq:yo-yo} % (2.10)
\end{align}
and, as it turns out, $P_{\th(l-\cdot)} f(r) = R_d[f(r)]$, when using
the regulator~\eqref{eq:simple-scale}. All this was clarified
in~\cite{Carme}. The operator $\del_\al x^\al = E + d$, with
$E := x^\al\,\del_\al$ denoting the Euler operator, figures
prominently there.

Note, moreover, when both the distribution $f$ and the regulator $w$
enjoy rotational symmetry, employing the MacLaurin--Lagrange expansion
for~$\vf$ and summation over the angles, the last displayed formula
amounts to the computation:
\begin{equation}
\duo< P_wf(r), \vf(x)> \equiv \biggl< f(r),
\vf(x) - \vf(0) - \frac{\Dl\vf(0)}{2!\,d}\,r^2 -\cdots- 
w(r) \frac{\Om_{d,l}\,\Dl^l\vf(0)}{(2l)!\,\Om_d}\, r^{2l} \biggr>,
\label{eq:pizza-crust} % (2.11)
\end{equation}
up to the highest~$l$ such that $2l \leq k$. Rotational symmetry of 
extensions in general can be studied like Lorentz covariance was 
in~\cite[Sect.~4]{Bettina} and in~\cite[Sect.~3.3]{Hollands08}.

We remark finally that the MacLaurin expansion for $\Pi\vf$ is written
$$
\Pi\vf(r) = 2^{d/2-1} \Ga(d/2) \bigl( r \sqrt{-\Dl}\,\bigr)^{1-d/2}
J_{d/2-1}\bigl( r \sqrt{-\Dl}\,\bigr) \vf(0),
$$
for $J_\al$ the Bessel function of the first kind and order~$\al$.
This makes sense for complex~$\al$. That is the nub of dimensional
regularization in position space, as found by Bollini and Giambiagi
themselves~\cite{BolliniG96} -- with Euclidean signature, in the
present case.

% \S 2.2
\subsection{Log-homogeneous distributions}
\label{ssc:log-homog}

We are interested in the amplitude $R_4[r^{-4}]$, corresponding to 
the ``fish'' graph \fish\ in the $\phi^4_4$ model. For clarity, it is 
useful to work in any dimension $d \geq 3$. Note the following:
\begin{align}
R_d[r^{-d}] &= r^{-d+1} \Bigl( \log\frac{r}{l} \Bigr)'
= r^{-d} E \Bigl( \log\frac{r}{l} \Bigr)
= \del_\al \Bigl( x^\al r^{-d} \log\frac{r}{l} \Bigr)
\notag \\
&= - \del_\al \del^\al \biggl( \frac{r^{-d+2}}{d - 2} \log\frac{r}{l}
+ \frac{r^{-d+2}}{(d - 2)^2} \biggr)
= - \frac{1}{d - 2}\, \biggl[
\Dl\biggl( r^{-d+2} \log\frac{r}{l} \biggr) - \Om_d\,\dl(r) \biggr].
\label{eq:basic-exten} % (2.12)
\end{align}
The last term appears because $r^{-d+2}/(-d + 2)\Om_d$ is the
fundamental solution for the Laplacian on~$\bR^d$. An advantage of
this form is that the corresponding momentum space amplitudes are
easily computed -- as will be exploited later: see 
Appendix~\ref{app:pspace-extens}.

In calculation of graphs on $\bR^4$ with subdivergences, extensions of
$r^{-4}\log^m(r/l)$, with growing powers of logarithms, crop up again
and again. It is best to grasp them all together. One can introduce
different scales, but for simplicity we stick with just one scale.
Dimensional reduction gives
\begin{equation}
R_d \Bigl[r^{-d} \log^m \frac{r}{l}\Bigr]
= \frac{1}{m+1}\, r^{-d+1}\,
\frac{d}{dr} \biggl( \log^{m+1} \frac{r}{l} \biggr)
\label{eq:genl-exten} % (2.13)
\end{equation}
for any $m = 0,1,2,\dots$

\medskip

A distribution $f$ on $\bR^d$ is called \textit{log-homogeneous} of
\textit{bidegree} $(a,m)$ if
\begin{equation}
(E - a)^{m+1} f = 0,  \word{but}  (E - a)^m f \neq 0.
\label{eq:log-homog} % (2.14)
\end{equation}
Here $m$ is a nonnegative integer but $a$ can be any complex number;
the case $m = 0$ obviously corresponds to homogeneous distributions.
For example, the distribution $\log r \in \sS'(R^d)$ is
log-homogeneous of bidegree~$(0,1)$. Essentially the same definition
is found in \cite[Sect.~4.1.6]{Scott10}. See also
\cite[Sect.~I.4]{GelfandS64}, \cite[Sect.~2.4]{Carme}
and~\cite[Prop.~4.4]{NikolovST13}, where the nomenclature used is
``associate homogeneous of degree~$a$ and order~$m$''.

Log-homogeneous distributions are tempered. Indeed, if $f$ is
homogeneous of bidegree $(a,0)$, then~\cite{EstradaFunfacts14} one can
find $g_0 \in \sD'(\bS^{n-1})$ so that $f(x) = r^a\,g_0(\om)$ for
$x = r\om \in \bR^d \less \{0\}$. More generally, for $f$ of bidegree
$(a,m)$, one can inductively construct
$g_0,\dots,g_m \in \sD'(\bS^{n-1})$ such that 
$f(r\om) = \sum_{k=0}^m r^a \log^{m-k} r \,g_k(\om)$ for $r > 0$. It
follows that $f \in \sS'(\bR^d)$.

A related issue is whether a log-homogeneous distribution on
$\bR^d \less \{0\}$ can be extended to one on~$\bR^d$. As we shall 
immediately exemplify, this can always be achieved although the 
bidegree may change: if the bidegree off the origin is $(a,m)$, that 
of the extension may be $(a,n)$ with $n \geq m$. For a general proof of 
that, showing also that rotational (or Lorentz) invariance may be 
kept in the extension, see Lemma~6 of~\cite{Hollands08}.

\medskip

The dilation behaviour of $R_d[r^{-d}]$ is immediate from
formula \eqref{eq:basic-exten}:
$$
R_d[(\la r)^{-d}]
= \del_\al \biggl( x^\al (\la r)^{-d} \Bigl(
\log\frac{r}{l} + \log\la \Bigr) \biggr)
= \la^{-d} R_d[r^{-d}] + \la^{-d} \log\la \,\Om_d \,\dl(r),
$$
generalizing~\eqref{eq:first-log}. Note that $\Om_d$ is simply minus
the coefficient of~$\log l$. In infinitesimal terms,
$$
E R_d[r^{-d}] = -d\,R_d[r^{-d}] + \Om_d\,\dl(r),  \word{so that}
\Res[r^{-d}] := [E, R_d](r^{-d}) = \Om_d\,\dl(r).
$$
Hence $R_d[r^{-d}]$ is log-homogeneous of bidegree~$(-d,1)$. The
functional $\Res$ coincides with the Wodzicki residue
\cite[Chap.~7.3]{Polaris}. It coincides as well with the ``analytic
residue'' in~\cite{NikolovST13}.

For our own purposes (RG calculations), we prefer to invoke the
logarithmic derivative of the amplitudes with respect to the length
scale~$l$:
\begin{equation}
\frac{\del}{\del\log l}\, R_d[r^{-d}] = \lddl R_d[r^{-d}]
= - \Om_d\,\dl(r);
\label{eq:spiny-scale} % (2.15)
\end{equation}
which for primitive diagrams like the fish yields the residue yet
again. As was shown in~\cite{Carme}, this is actually a functional
derivative with respect to the regulator~$w$; and so it can be widely
generalized.

\begin{lema} % 1
\label{lm:logdiv-dilation}
For $d \geq 3$, $m = 0,1,2,\dots$ and $\la > 0$, the following 
relation holds:
\begin{equation}
R_d \Bigl[ (\la r)^{-d} \log^m \frac{\la r}{l} \Bigr]
= \sum_{k=0}^m \la^{-d} \log^k\la \binom{m}{k}
R_d \Bigl[ r^{-d} \log^{m-k} \frac{r}{l} \Bigr]
+ \la^{-d} \log^{m+1}\la \,\frac{\Om_d}{m + 1}\, \dl(r).
\label{eq:logdiv-dilation} % (2.16)
\end{equation}
Therefore, $R_d\bigl[ r^{-d} \log^m(r/l) \bigr]$ is log-homogeneous of
bidegree $(-d, m+1)$.
\end{lema}

\begin{proof}
This is a direct verification:
\begin{align*}
R_d \Bigl[ (\la r)^{-d} \log^m \frac{\la r}{l} \Bigr]
&= \frac{\la^{-d}}{m + 1}\, \del_\al \biggl( x^\al r^{-d} \Bigl(
\log\frac{r}{l} + \log \la \Bigr)^{m+1} \biggr) 
\\
&= \frac{\la^{-d}}{m + 1} \sum_{k=0}^{m+1} \binom{m+1}{k} \log^k\la
\,\del_\al \Bigl( x^\al r^{-d} \log^{m-k+1}\frac{r}{l} \Bigr)
\\
&= \sum_{k=0}^{m+1} \la^{-d} \log^k\la \,\frac{m!}{k!(m-k+1)!} 
\,\del_\al \Bigl( x^\al r^{-d} \log^{m-k+1}\frac{r}{l} \Bigr)
\\
&= \sum_{k=0}^m \la^{-d} \log^k\la \binom{m}{k} 
R_d \Bigl[ r^{-d} \log^{m-k} \frac{r}{l} \Bigr]
+ \la^{-d} \log^{m+1}\la \,\frac{1}{m + 1}\, \del_\al(x^\al r^{-d}),
\end{align*}
and the result follows from the relation 
$\del_\al(x^\al r^{-d}) = - \Dl((d - 2) r^{-d+2}) = \Om_d\,\dl(r)$.
\end{proof}

As an immediately corollary, we get the effect of the Euler operator,
when $m \geq 1$:
$$
E R_d \Bigl[ r^{-d} \log^m \frac{r}{l} \Bigr]
= -d\,R_d \Bigl[ r^{-d} \log^m \frac{r}{l} \Bigr]
+ m\,R_d \Bigl[ r^{-d} \log^{m-1} \frac{r}{l} \Bigr].
$$
On the other hand, an elementary calculation for $r > 0$ gives
$$
E \Bigl[ r^{-d} \log^m \frac{r}{l} \Bigr]
= r\,\frac{d}{dr} \Bigl[ r^{-d} \log^m \frac{r}{l} \Bigr]
= -d\,r^{-d} \log^m \frac{r}{l} + m\,r^{-d} \log^{m-1} \frac{r}{l} \,,
$$
so that $R_d E\bigl[ r^{-d} \log^m(r/l) \bigr]
= E R_d\bigl[ r^{-d} \log^m(r/l) \bigr]$; consequently, higher
residues all vanish:
$$
\Res \Bigl[ r^{-d} \log^m \frac{r}{l} \Bigr]
:= [E, R_d] \Bigl( r^{-d} \log^m \frac{r}{l} \Bigr) = 0
\quad\text{for } m = 1,2,3,\dots
$$

We summarize. First, by the same trick of \eqref{eq:basic-exten},
$$
R_d \Bigl[ r^{-d} \log^m \frac{r}{l} \Bigr]
= \frac{E + d}{m + 1} \Bigl( r^{-d} \log^{m+1} \frac{r}{l} \Bigr),
$$
which makes obvious much of the above. This formula also shows that 
the aforementioned algebra property applies to logarithms as well as 
polynomials:
$$
\log \frac{r}{l}\, R_d \Bigl[ r^{-d} \log^m \frac{r}{l} \Bigr]
= R_d \Bigl[ r^{-d} \log^{m+1} \frac{r}{l} \Bigr].
$$

Second, we can use this operator to amplify a well-known property of
homogeneous distributions: the \textit{Fourier transform} $\sF f$ of a
log-homogeneous distribution $f \in \sS'(\bR^d)$ of bidegree $(a,m)$
is itself log-homogeneous of bidegree $(-d-a, m)$. Indeed, since
$\sF(x^\al\,\del_\al) = -(\del_\al x^\al)\sF$, i.e.,
$\sF E = -(E + d)\sF$ as operators on $\sS'(\bR^d)$, the
relations~\eqref{eq:log-homog} are equivalent to
$$
(E + d + a)^{m+1} \sF f = 0,  \word{but}  (E + d + a)^m \sF f \neq 0.
$$
The Fourier transforms of the considered regular distributions are
also regular~\cite[Chap.~7.3]{Polaris}. Moreover, $\sF$ is an
isomorphism of the indicated spaces~\cite{OrtnerW13}.

Third, one can rewrite the result of Lemma~\ref{lm:logdiv-dilation} to
show that it exhibits a representation of the dilation group. Indeed,
on multiplying both sides of~\eqref{eq:logdiv-dilation} by $\la^d/m!$,
we obtain
$$
\frac{\la^d}{m!} R_d \Bigl[ (\la r)^{-d} \log^m \frac{\la r}{l} \Bigr]
= \sum_{k=0}^m  \frac{\log^k\la}{k!}\, \frac{1}{(m - k)!}\,
R_d \Bigl[ r^{-d} \log^{m-k} \frac{r}{l} \Bigr]
+ \,\frac{\log^{m+1}\la}{(m + 1)!}\, \Om_d\,\dl(r).
$$
This shows that the distributions
$\frac{1}{k!}\, R_d[r^{-d} \log^k(r/l)]$, for $k = 0,1,\dots$, plus
the special case $\Om_d\,\dl(r)$ for $k = -1$, form an eigenvector
(with eigenvalue~$\la^d$) for a certain unipotent matrix
$\exp(A \log\la)$, yielding an action of the dilation group -- this is
just Proposition~3.2 of~\cite{NikolovST13}.

Fourth, the obvious relation
\begin{equation}
\lddl  R_d \Bigl[ r^{-d} \log^m \frac{r}{l} \Bigr]
= -m R_d \Bigl[ r^{-d} \log^{m-1} \frac{r}{l} \Bigr],
\quad \text{for } m \geq 1,
\label{eq:huevo-de-Colon} % (2.17)
\end{equation}
will be most useful in the sequel.

Fifth, formulas involving the Laplacian, like~\eqref{eq:basic-exten},
do exist for all the log-homogeneous distributions
\eqref{eq:genl-exten}, and thus for the graphs. We develop these
formulas in Appendix~\ref{app:exten-formulas}.

% \S 2.3
\subsection{Here comes the sun}
\label{ssc:quad-exten}

One of us introduced in \cite[Sect.~4.2]{Carme}, on the basis of
related expressions by Estrada and Kanwal
\cite{EstradaK85,EstradaK89}, the powerful formula
\begin{align}
\Dl^n R_d\bigl[ r^{-d-2m} \bigr]
&= \frac{(d + 2m + 2n - 2)!!}{(d + 2m - 2)!!}\,
\frac{(2m + 2n)!!}{(2m)!!}\, R_d\bigl[ r^{-d-2m-2n} \bigr]
\nonumber \\
&\qquad - \frac{\Om_{d,m}}{(2m)!} \sum_{l=1}^n
\frac{(4m+4l+d-2)}{2(m+l)(2m+2l+d-2)}\, \Dl^{n+m}\dl(r).
\label{eq:extra-laps} % (2.18)
\end{align}
The first term on the right hand side just corresponds to the na\"ive
derivation formula. Once the case $n = 1$ is established, the general
formula follows by a straightforward iteration, using the relation
$\Om_{d,m+1}/\Om_{d,m} = (2m + 1)/(2k + d)$. This provides explicit
expressions for divergences higher than logarithmic.

\medskip

Consider thus the case: $d = 4$, $m = 0$, $n = 1$, which yields
$$
\Dl R_4[r^{-4}] = 8 R_4[r^{-6}] - \frac{3\pi^2}{2} \,\Dl\dl(r).
$$
Without further ado, we get the renormalization of the quadratically
divergent ``sunset'' graph of the $\phi^4_4$ model:
\begin{gather}
\sunset \quad
\text{which in $x$-space is \emph{primitive} (subdivergence-free)}:
\notag \\
R_4[r^{-6}]
= \frac{1}{8}\,\Dl R_4[r^{-4}] + \frac{3\pi^2}{16}\,\Dl\dl(r)
= -\frac{1}{16}\,\Dl^2 \Bigl( r^{-2} \log\frac{r}{l} \Bigr)
+ \frac{5\pi^2}{16}\,\Dl\dl(r).
\label{eq:ren-hex} % (2.19)
\end{gather}
The same result can be retrieved directly from
formula~\eqref{eq:yo-yo}, see~\cite{Carme}. Its first term is
log-homogeneous of bidegree $(-6,1)$. It is worth noting here that in
the paper by Freedman, Johnson and Latorre two different extensions
\cite[Eq.~(A.1)]{FreedmanJL92} and \cite[Eq.~(2.8)]{FreedmanJL92} are
given for this graph,
\begin{align*}
r_\FJL[r^{-6}]
= -\frac{1}{16}\, \Dl^2 \Bigl( r^{-2} \log \frac{r}{l} \Bigr),
\!\word{respectively}\! 
r_\FJL[r^{-6}] 
= -\frac{1}{16}\, \Dl^2\Bigl( r^{-2} \log \frac{r}{l} \Bigr)
- \frac{16\pi^2\dl(r)}{l^2} \,.
\end{align*}
The first one does not fulfil the algebra property, the second one
moreover brings in an unwelcome type of dependence on~$l$. Note as
well that rotational symmetry allows \textit{two} arbitrary constants
in the renormalization of this graph; the algebra property reduces the
ambiguity to one.

The scale derivative gives
\begin{equation}
\lddl R_4[r^{-6}] = - \frac{\Om_4}{8}\, \Dl\dl(r).
\label{eq:ren-hex-scale} % (2.20)
\end{equation}

The reader may renormalize straightforwardly
from~\eqref{eq:extra-laps} the simplest vacuum graph of the model.

We compute the following commutation relations:
\begin{equation}
[\Dl, E + d] = [\Dl, E] = 2\Dl; \quad 
[\Dl, r^2] = 2d + 4E;  \quad
[E, r^2] = 2r^2,
\label{eq:manes-de-sl2r} % (2.21)
\end{equation}
valid for radial functions or distributions. This allows us to run an
indirect check of~\eqref{eq:ren-hex}, highlighting the algebra
property:
\begin{align*}
r^2 R_4[r^{-6}] &= \frac{1}{8}\,\Dl(r^{-2}) - R_4[r^{-4}]
- \frac{1}{2}\,ER_4[r^{-4}] + \frac{3\pi^2}{2}\,\dl(r)
\\
&= - \frac{\pi^2}{2}\,\dl(r) - R_4[r^{-4}] + 2 R_4[r^{-4}]
- \pi^2\,\dl(r) + \frac{3\pi^2}{2}\,\dl(r) = R_4[r^{-4}];
\end{align*}
where we have used $r^2 \Dl\dl(r) = 2d\,\dl(r)$ and the third identity
in~\eqref{eq:manes-de-sl2r}.

\medskip

More generally, the distribution $R_d[r^{-d-2k}]$ is log-homogeneous
of bidegree $(-d-2k, 1)$, since one finds~\cite{Carme} that
\begin{equation}
R_d\bigl[ (\la r)^{-d-2k} \bigr] 
= \la^{-d-2k}\, R_d[r^{-d-2k}]
+ \la^{-d-2k} \log\la \,\frac{\Om_{d,k}}{k!} \,\Dl^k\dl.
\label{eq:powers-homog} % (2.22)
\end{equation}
A few explicit expressions for $R_4\bigl[ r^{-6} \log^m(r/l) \bigr]$
terms, which we shall need later on, are given in 
Appendix~\ref{app:exten-formulas}.

% \S 2.4
\subsection{Trouble with the formulas for derivatives}
\label{ssc:la-puntilla}

We remind the reader that there are no extensions of $r^{-k-1}$ for
which the generalization of both requirements \eqref{eq:mult-property}
and \eqref{eq:diff-property} hold simultaneously. One finds
\cite{EstradaK89} that
$$
r^m\,r_1[r^{-k-1}] = r_1[r^{-k+m-1}] + [H_{k-m} - H_k]\,\dl^{(k-m)}(r).
$$
This also means that differential renormalization in the sense of
Freedman, Johnson and Latorre is inconsistent with dimensional
reduction; which early on drew justified criticism \cite{Schnetz97}
towards heuristic prescriptions on~$\bR^4$ in~\cite{FreedmanJL92},
such as
$$
r_\FJL[r^{-4}]
= - \frac{1}{2}\, \Dl\Bigl( r^{-2} \log\frac{r}{l} \Bigr).
$$
For instance, with a glance at \eqref{eq:reductio-ad-unum}
and~\eqref{eq:basic-exten}, we immediately see that the implied
renormalization of $r^{-1}$ on the half-line would be
$\log'(r/l) - \half\dl(r)$, instead of $\log'(r/l)$. This makes too
much of a break with the rules of calculus.

In general, the distributional derivative of a natural extension of a
singular function will not coincide with the natural extension of its
derivative. An instructive discussion of this point is given
in~\cite{EstradaF02}.

% \S 3
\section{Convolution-like composition of distributions}
\label{sec:convoluted}

The convolution of two integrable functions defined on a Euclidean
space $\bR^d$ is given by the well-known formula
$$
f * g(x) = \int f(y) g(x - y) \,dy,
$$
the integral being taken over $\bR^d$. To convolve two distributions,
one starts with the equivalent formula
\begin{equation}
\duo< f * g, \vf> := \iint f(x) g(y) \vf(x + y) \,dx\,dy
\label{eq:convol-defn-a} % (3.1)
\end{equation}
where $\vf \in \sD(\bR^d)$ here and always denotes a test function.
The right hand side of~\eqref{eq:convol-defn-a} may be regarded as a
duality formula:
\begin{equation}
\duo< f * g, \vf> := \duo< f \ox g, \vf^\Dl>,
\label{eq:convol-defn-b} % (3.2)
\end{equation}
where $\vf^\Dl(x,y) := \vf (x + y)$ and the pairing on the right hand 
side takes place over~$\bR^{2d}$. 

Notice that $\vf^\Dl \in C^\infty(\bR^{2d})$ is smooth but no longer
has compact support, so that \eqref{eq:convol-defn-b} only makes sense
for certain pairs of distributions $f$ and~$g$. If, say, one of the
distributions $f$ or~$g$ has compact support, then $f * g$ is well
defined as a distribution by~\eqref{eq:convol-defn-b}, and
associativity formulas like $(f * g) * h = f * (g * h)$ are meaningful
and valid if at least two of the three factors have compact support.
Also if, for instance, a distribution is tempered, $f \in \sS'(\bR^d)$,
then one can take $g \in \sO'_c(\bR^d)$, the space of distributions
``with rapid decrease at infinity''. This variant is dealt with in the
standard references, see \cite[p.~246]{Schwartz66} or
\cite[p.~423]{Horvath66}.

However, these decay conditions are not met in quantum field practice,
so we must amplify the definition of convolution.

\medskip

One can alternatively interpret the integral
in~\eqref{eq:convol-defn-a} as pairing by duality the distribution\\
$f(x) g(y) \vf(x + y)$ with the constant function~$1$:
\begin{equation}
\duo< f * g, \vf> := \duo< \vf^\Dl(f \ox g), 1> .
\label{eq:convol-defn-c} % (3.3)
\end{equation}
For that, one must determine conditions on $f$ and~$g$ so that the
pairing on the right hand side -- again over~$\bR^{2d}$ -- makes
sense.

Consider the space $\sB_0(\bR^d)$ of smooth functions on~$\bR^d$ that
vanish at infinity together with all their derivatives. Its dual space
$\sB_0'(\bR^d)$ is the space of \textit{integrable distributions}.
(The notation follows~\cite{Horvath66}; the space of integrable
distributions is called $\sD'_{L^1}(\bR^d)$ by
Schwartz~\cite{Schwartz66}.) The dual space of $\sB_0'(\bR^d)$ itself
is larger than $\sB_0(\bR^d)$: it is
$\sB_0''(\bR^d) \equiv \sB(\bR^d)$, the space of smooth functions all
of whose derivatives are merely bounded on~$\bR^d$.

It is known \cite[Sect.~4.5]{Horvath66} that a distribution $f$ is
integrable if and only if it can be written as
$f = \sum_\al \del^\al \mu_\al$, a finite sum of derivatives of
finite measures $\mu_\al$. A particularly useful class of integrable
distributions are those of the form $f = h + k$ where $h$ has compact
support and $k$ is a function which is integrable (in the usual sense)
and vanishes on the support of~$h$.

\begin{defn}
\label{df:convolvability}
Two distributions $f,g \in \sD'(\bR^d)$ are called
\textit{convolvable} if $\vf^\Dl(f \ox g) \in \sB_0'(\bR^{2d})$ for
any $\vf \in \sD(\bR^d)$.

Since $1 \in \sB(\bR^{2d})$, the right hand side of
\eqref{eq:convol-defn-c} then makes sense as the evaluation of a
(separately continuous) bilinear form; and hence it defines
$f * g \in \sD'(\bR^d)$.
\end{defn}

This definition of convolvability was introduced in~\cite{Horvath74}
by Horv\'ath, under the name ``condition~$(\Ga)$''; and he showed
there that it subsumes previous convolvability conditions, such as the
aforementioned one between $\sS'(\bR^d)$ and~$\sO'_c(\bR^d)$. It is
known that $\sB_0'(\bR^d)$ is an (associative) convolution algebra:
such a result already appears in~\cite{Schwartz66} and the proof has
been adapted to the above notion of convolvability by Ortner and
Wagner~\cite[Prop.~9]{OrtnerW89}.

\medskip

Now, how can one tell when two given (say, log-homogeneous)
distributions are convolvable or not? Consider, for instance, the
log-homogeneous distribution $R_d[r^{-4}]$ on~$\bR^4$, defined in the
previous section. It yields the renormalization of the ``fish'' graph
in the massless~$\phi^4_4$ model; and its convolution with itself
amounts to the correct definition of a chain (articulated, one-vertex
reducible) diagram, the ``spectacles'' or ``bikini'' graph \bikini\,.
The following result allows us to attack the calculation of several
graphs in the next section.

\begin{prop} % 2
\label{pr:convolvables}
The convolution of log-homogeneous distributions
\begin{equation}
R_d \Bigl[ r^\la \log^m \frac{r}{l} \Bigr]
* R_d \Bigl[ r^\mu \log^k \frac{r}{l} \Bigr]
\label{eq:nonum} % (3.4)
\end{equation}
is well defined whenever $\Re(\la + \mu) < -d$, for any 
$m,k = 0,1,2,\dots$
\end{prop}

\begin{proof}
The convolution algebra property takes care of the case where
$\la < -d$ and $\mu < -d$. The weaker condition $\Re(\la + \mu) < -d$
allows us to incorporate also the borderline cases $\la = -d$, which
we shall need.

Consider first the case where $m = k = 0$. Theorem~3
in~\cite{Horvath78} shows that a distribution $f$ on~$\bR^d$ is
convolvable with $R_d[r^\mu]$ if $(1 + r^2)^{\Re\mu/2} f$ lies in
$\sB_0'(\bR^d)$; this uses our earlier remark that $R_d[r^\mu]$
coincides with the meromorphic continuation extension of the
function~$r^\mu$. This sufficient condition on~$f$ is guaranteed in
turn if $f = f_0 + f_1$ where $f_0$ has compact support and $f_1$ is
locally integrable with $|f_1(x)| \leq C\,|x|^{\Re\mu}$ for
large~$|x|$.

The last statement is not obvious; the crucial lemma
of~\cite{Horvath78} shows that integrability follows from the
boundedness of the following functions:
$$
h_{s,c,p}(y) := \int_{A_c}
|x|^s \,\del_y^p \bigl( (1 + |y|^2)^{-s/2} \bigr) \,dx,
\word{where} A_c = \set{x : |x| \geq 1,\ |x + y| \leq c},
$$
which holds for any real~$s$, any $c > 0$ and derivatives $\del_y^p$
of all orders. Consequently, in the previous argument, $R_d[r^\mu]$
itself may be replaced by any distribution $g$ of the form
$g = g_0 + g_1$ where $g_0$ has compact support and $g_1$ is locally
integrable with $|g_1(x)| \leq C\,|x|^{\Re\mu}$ for large~$|x|$.

In particular, taking $f = R_d[r^\la]$ and letting $f_0$ be its
restriction to a ball centred at the origin, it follows that
$R_d[r^\la]$ and $R_d[r^\mu]$ are convolvable whenever
$r^{\Re\la}(1 + r^2)^{\Re\mu/2}$ is integrable for $r \geq 1$, which 
is true if $\Re(\la + \mu) < -d$.

For the general case, a similar decomposition may be applied to both
convolution factors. Subtracting off their restrictions to balls 
centred at the origin, we can bound the remainders thus:
$$
r^{\Re\la} \log^m \frac{r}{l} \leq r^\al,  \quad
r^{\Re\mu} \log^k \frac{r}{l} \leq r^\bt,  \word{for all} r \geq r_0.
$$
Since $\Re(\la + \mu) < -d$, we can still choose $\al,\bt$ so that
$\al + \bt < -d$ if $r_0$ is large enough. The convolvability then 
follows from integrability of $(1 + r^2)^{(\al+\bt)/2}$ over
$r \geq r_0$.
\end{proof}

The actual calculation of $R_d[r^{-d}] * R_d[r^{-d}]$ was performed by
Wagner in \cite{Wagner87}, by a simple meromorphic continuation
argument. In dimensions $d \geq 3$, and for scale $l = 1$, the result
is:
\begin{equation}
R_d[r^{-d}] * R_d[r^{-d}] = 2\,\Om_d\, R_d[r^{-d}\log r]
+ \frac{\Om_d^2}{4}\Bigl( \psi'(d/2) - \frac{\pi^2}{6} \Bigr)\,\dl(r),
\label{eq:wagnerian-opera} % (3.5)
\end{equation}
with $\psi$ denoting the digamma function. By the same token,
convolutions of (renormalized) logarithmically divergent graphs, and
in particular chain graphs \textit{of any length}, are rigorously
defined and computable in massless $\phi^4_4$~theory in $x$-space. We 
shall perform a few of these convolutions later on.

The calculation may be transferred to $p$-space. Now -- contrary to an
implied assertion in~\cite{FreedmanJL92} -- the product of two
tempered distributions is not defined in general. Even so, in the
present case, the product of the Fourier transforms may be defined as
the Fourier transform of the convolution of their preimages in
$x$-space, under the same condition $\Re(\la + \mu) < -d$. In
Appendix~\ref{app:pspace-extens}, we calculate these regular $p$-space
representatives for $d = 4$.

\medskip

To conclude: maybe because the relevant information
\cite{HorvathCol74, Horvath74, OrtnerW89, Horvath78, Wagner87,
HorvathOW87} is scattered in several different languages, the powerful
mathematical framework for this, by Horv\'ath, Ortner and Wagner,
appears to be little known. So we felt justified in giving a detailed
treatment here.

\medskip

In the computation of graphs of the massless $\phi^4_4$ model up to
fourth order, one moreover finds slightly more complicated
convolution-like integrals. They will be tackled here by easy
generalizations of Horv\'ath's theory of convolution. We do not claim
that every infrared problem lurking in higher-order graphs can be
solved by these methods.

% \S 4
\section{Graphs}
\label{sec:tough-going}

The renormalization of multiloop graphs may be accomplished in
position space with the tools developed in Sections
\ref{sec:get-ready} and~\ref{sec:convoluted}. For a model $g\phi^4/4!$
scalar field theory on~$\bR^4$ -- widely used, e.g., in the theory of
critical exponents~\cite{KleinertSF01,ZinnJustin02} -- we perform here
the detailed comparison of Epstein--Glaser renormalization with the
differential renormalization approach, which was the subject of
extensive calculation in~\cite{FreedmanJL92}.

\medskip

We go about this as follows: first we compute the graphs of the second
and third order in the coupling constant for the (one-particle
irreducible) four-point function (respectively corresponding to one
and two loops), seemingly by brute force. Along the way we find the
scale derivatives for these graphs. Next, we exhibit the perturbation
expansion up to that order.

After that, we solve the more involved graphs of the fourth order in
the coupling constant, corresponding to three loops, for the
four-point function. We trust that the procedure to construct the
perturbation expansion up to fourth order is by then clear.

In Appendix~\ref{app:colon}, we perform similar calculations for the
two-point function, up to the same order in the coupling.

There are three groups of Feynman diagrams involved in the four-point
function:
\begin{align}
& \Bigl\{ \fish, \bikini, \trikini, \stye, \catseye \Bigr\};
\notag \\
& \biggl\{ \winecup, \duncecap, \kite, \shark \biggr\}; \quad
\biggl\{ \tetrahedron, \roll \biggr\};
\label{eq:diagram-list} % (4.1)
\end{align}
depending on the external leg configurations. 

\medskip

We begin with the one-loop fish graph. Since the Euclidean
``propagator'' is given by\\
$(-4\pi^2)^{-1} r^{-2}$, its bare amplitude is of the form:
\begin{align*}
& \frac{g^2}{(4\pi^2)^2} \bigl[
\dl(x_1 - x_2)\, (x_2 - x_3)^{-4} \,\dl(x_3 - x_4)
+ \dl(x_1 - x_3)\, (x_3 - x_4)^{-4} \,\dl(x_4 - x_2)
\\
&\hspace*{3em}
+ \dl(x_1 - x_4)\, (x_4 - x_2)^{-4} \,\dl(x_2 - x_3) \bigr].
\end{align*}
We write three terms because there are three topologically distinct
configurations of the vertices. Moreover, one must divide their
contribution by the ``symmetry factor'', which counts the order of the
permutation group of the lines, with the vertices fixed. Here this
number is equal to~$2$. This gives a total \textit{weight} of~$3/2$
for the fish graph in the perturbation expansion. In this paper we do
not use these variations, so we simply compute the weights of all the
graphs we deal with by the direct method in
\cite[Chap.~14]{KleinertSF01}.

Let us moreover leave aside weights and $(4\pi^2)^{-1}$ factors until
the moment when we sum the perturbation expansions. Taking advantage
of translation invariance to label the vertices as:
$$
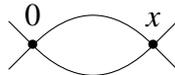
\begin{tikzpicture}[scale=1.6]
\coordinate (A) at (0,0) ; \coordinate (B) at (1,0) ; 
\draw (A) parabola[bend pos=0.5] bend +(0,0.25) (B) ;
\draw (A) parabola[bend pos=0.5] bend +(0,-0.25) (B) ;
\draw (A) ++(-0.2,0.2) -- (A) -- ++(-0.2,-0.2) ;
\draw (B) ++(0.2,0.2) -- (B) -- ++(0.2,-0.2) ;
\draw (A) node[above=3pt] {$0$} ;
\draw (B) node[above=3pt] {$x$} ;
\foreach \pt in {A,B} \draw (\pt) node{$\bull$} ;
\end{tikzpicture}
$$
we are left with just our old acquaintance $R_4[r^{-4}]$, with
$r = |x|$.

In what follows, we write $x^2 = x_\al x^\al$ for $x \in \bR^4$,
$x^4 = (x^2)^2$ and so on; all integrals are taken over~$\bR^4$.

% \S 4.1
\subsection{A third-order graph by convolution}
\label{ssc:bikini}

The next simplest case is the ``bikini'' graph, which can be labelled
thus, with $u$ denoting the internal vertex:
$$
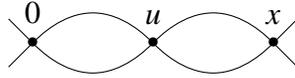
\begin{tikzpicture}[scale=1.6]
\coordinate (A) at (0,0) ; \coordinate (B) at (1,0) ; 
\coordinate (C) at (2,0) ;
\draw (A) ++(-0.2,0.2) -- (A) -- ++(-0.2,-0.2) ;
\draw (A) parabola[bend pos=0.5] bend +(0,0.25) (B) ;
\draw (A) parabola[bend pos=0.5] bend +(0,-0.25) (B) ;
\draw (B) parabola[bend pos=0.5] bend +(0,0.25) (C) ;
\draw (B) parabola[bend pos=0.5] bend +(0,-0.25) (C) ;
\draw (C) ++(0.2,0.2) -- (C) -- ++(0.2,-0.2) ;
\draw (A) node[above=3pt] {$0$} ;
\draw (B) node[above=3pt] {$u$} ;
\draw (C) node[above=3pt] {$x$} ;
\foreach \pt in {A,B,C} \draw (\pt) node{$\bull$} ;
\end{tikzpicture}
$$
The rules of quantum mechanics prescribe integration over the internal
vertices in $x$-space. From the unrenormalized amplitude
$$
\int \frac{1}{u^4}\,\frac{1}{(x - u)^4} \,du
$$
(which looks formally like a convolution, though the factors are not
actually convolvable) we obtain, on replacing these factors by their
extensions, the renormalized version
\begin{equation}
\int R_4[u^{-4}]\, R_4[(x - u)^{-4}] \,du,
\label{eq:bikini-convl} % (4.2)
\end{equation}
a \textit{bona fide} convolution, since, as shown in the previous
section, $R_4[r^{-4}]$ is convolvable with itself. Specializing
\eqref{eq:wagnerian-opera} to $d = 4$ and using  
$\psi'(2) - \pi^2/6 = -1$, one may conclude that
\begin{equation}
\bikini
= 4\pi^2 R_4\Bigl[ x^{-4} \log\frac{|x|}{l} \Bigr] - \pi^4\,\dl(x).
\label{eq:bikini-ren} % (4.3)
\end{equation}
This looks like a log-homogeneous amplitude of bidegree $(-4,2)$. More
precisely, in the three difference variables, say $x_1 - x_2$,
$x_2 - x_3$, $x_3 - x_4$, with the obvious relabeling of the indices,
it is quasi-log-homogeneous of bidegree $(-4,0;-4,2;-4,0)$.

A peek at Eq.~\eqref{eq:rminus4-one-log} now shows that the
coefficient of $\log^2 l$ in the result \eqref{eq:bikini-ren} is equal
to~$4\pi^4$. Here we observe that ``predicting'' the coefficients of
$\log^k l$ for $k > 1$ is fairly easy. With the help of
\cite{ChryssomalakosQRV02}, which determines the primitive elements
in bialgebras of graphs,%
\footnote{Actually that reference deals with the bialgebra of rooted
trees, but \textit{cela fait rien \`a l'affaire}.}
a method recommended by Kreimer \cite{Kreimer01} was applied
in~\cite{HeidyMSc06}. To wit, primitive elements should have
\textit{vanishing} coefficients of $\log^k l$ for $k > 1$. In the
present case, the bikini graph minus the square of the fish is
primitive, and so for the coefficient of $\log^2 l$ the value
$\Om_4^2$, equal to the obtained $4\pi^4$, was predicted.

\medskip

For later use, we obtain the scale derivative:
\begin{equation}
\lddl  \Bigl( \bikini \Bigr) = -4\pi^2 \fish,
\label{eq:fewer-fish} % (4.4)
\end{equation}
directly from~\eqref{eq:huevo-de-Colon} and~\eqref{eq:bikini-ren}.

% \S 4.2
\subsection{A third-order ladder graph: the winecup}
\label{ssc:wine-cup}

Next comes the winecup or ice-cream ladder graph, with vertices 
labelled as follows:
$$
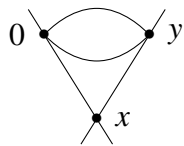
\begin{tikzpicture}[scale=1.4]
\coordinate (A) at (0,0) ; \coordinate (B) at (1,0) ; 
\coordinate (C) at (0.5,-0.8) ;
\draw ($ (A)!-0.3!(C) $) -- ($ (A)!1.3!(C) $) ;
\draw ($ (B)!-0.3!(C) $) -- ($ (B)!1.3!(C) $) ;
\draw (A) parabola[bend pos=0.5] bend +(0,0.25) (B) ;
\draw (A) parabola[bend pos=0.5] bend +(0,-0.25) (B) ;
\draw (A) node[left=3pt] {$0$} ;
\draw (B) node[right=3pt] {$y$} ;
\draw (C) node[right=3pt] {$x$} ;
\foreach \pt in {A,B,C} \draw (\pt) node{$\bull$} ;
\end{tikzpicture}
$$
We denote it $\copadevino(x,y)$ for future use. The corresponding
bare amplitude is given by
$$
f(x,y) = \frac{1}{x^2 (x - y)^2 y^4} \,.
$$
Consider a ``partially regularized'' version of it, for which the
known formulas yield:
\begin{equation}
\Rbar_8 \bigl[ x^{-2} (x - y)^{-2} y^{-4} \bigr]
= - \frac{1}{2}\, x^{-2} (x - y)^{-2} 
\,\Dl\Bigl( y^{-2} \log \frac{|y|}{l} \Bigr) + \pi^2 x^{-4}\,\dl(y).
\label{eq:prematur} % (4.5)
\end{equation}
The last expression indeed makes sense for all $(x,y) \neq (0,0)$.
To proceed, we largely follow~\cite{FreedmanJL92},%
\footnote{Since $\Rbar_4[y^{-4}] x^{-2} (x - y)^{-2}$ is undefined
only at the origin, the procedure~\eqref{eq:yo-yo} assuredly works,
giving rise to alternative expressions, very much like the ones
proposed by Smirnov and Zav'yalov some time ago~\cite{ZavyalovS93}.
This was the path taken in~\cite{HeidyMSc06}. We find those, however,
somewhat unwieldy.}
which invokes Green's integration-by-parts formula,
\begin{equation}
(\Dl B)\,A = (\Dl A)\,B + \del^\bt(A \del_\bt B - B \del_\bt A),
\label{eq:green-backs} % (4.6)
\end{equation}
that will be rigorously justified soon, in the present context. Thus
\begin{align*}
- \frac{1}{2}\, x^{-2}(x - y)^{-2}
\,\Dl_y \Bigl( y^{-2} \log\frac{|y|}{l} \Bigr)
&= - \frac{1}{2}\, x^{-2} y^{-2} \log\frac{|y|}{l}
\,\Dl_y \bigl( (x - y)^{-2} \bigr)
+ \frac{1}{2}\, x^{-2} \,\del_y^\bt L_\bt(y; x - y) 
\\
&= 2\pi^2x^{-4} \log\frac{|x|}{l}\,\dl(x - y) 
+ \frac{1}{2} x^{-2} \,\del_y^\bt  L_\bt(y; x - y),
\end{align*}
where
\begin{equation}
L_\bt(y; x - y) := y^{-2} \log\frac{|y|}{l}\, \del^y_\bt((x - y)^{-2})
- (x - y)^{-2} \,\del^y_\bt \Bigl( y^{-2} \log\frac{|y|}{l} \Bigr)
\label{eq:total-deriv-term1} % (4.7)
\end{equation}
deserves a name, since it is going to reappear often. The presence of
the $\dl(x - y)$ factor is rewarding. Now it is evident that
renormalized forms of $x^{-4}$ and $x^{-4} \log|x|$ should be used.
The only treatment required by the last term is that the derivative be
understood in the distributional sense. Thus, in the end, we have
computed:
\begin{align}
\copadevino(x,y)
&= R_8\bigl[ x^{-2} (x - y)^{-2} y^{-4} \bigr]
\notag \\
&= 2\pi^2 R_4\Bigl[ x^{-4} \log\frac{|x|}{l} \Bigr] \,\dl(x - y)
+ \pi^2 R_4[x^{-4}] \,\dl(y)
+ \frac{1}{2} x^{-2} \,\del_y^\bt L_\bt(y; x - y).
\label{eq:my-cup-overfloweth} % (4.8)
\end{align}

We now carefully justify Eq.~\eqref{eq:green-backs} for this case.
Under the hypothesis $\vf(0,0) = 0$, substitute
$A(y) = (x - y)^{-2}$ and $B(y) = y^{-2}\log(|y|/l)$ there, and
compute:
\begin{align*}
& \Bigl< x^{-2} \Dl\Bigl( y^{-2} \log\frac{|y|}{l} \Bigr),
(x - y)^{-2} \vf(x,y) \Bigr> 
= \Bigl< x^{-2} y^{-2} \log\frac{|y|}{l},
\Dl_y\bigl( (x - y)^{-2} \vf(x,y) \bigr) \Bigr>
\\
&= \Bigl< x^{-2} y^{-2} \log\frac{|y|}{l} \,\Dl_y((x - y)^{-2}),
\vf(x,y) \Bigr>
+ 2 \Bigl< x^{-2} y^{-2} \log\frac{|y|}{l}\, \del^y_\bt((x - y)^{-2}),
\del_y^\bt \vf(x,y) \Bigr>
\\
&\qquad - \Bigl< x^{-2} \,\del^y_\bt \Bigl(
(x - y)^{-2} y^{-2} \log\frac{|y|}{l}\Bigr), \del_y^\bt\vf(x,y) \Bigr>
\\
&= \Bigl< x^{-2} y^{-2} \log\frac{|y|}{l} \,\Dl_y((x - y)^{-2}),
\vf(x,y) \Bigr>
+ \Bigl< x^{-2} y^{-2} \log\frac{|y|}{l}
\del^y_\bt((x - y)^{-2}), \del_y^\bt\vf(x,y) \Bigr>
\\
&\qquad - \Bigl< x^{-2} (x - y)^{-2} \,\del^y_\bt \Bigl(
y^{-2} \log\frac{|y|}{l} \Bigr), \del_y^\bt\vf(x,y) \Bigr>
\\
&= \Bigl< x^{-2}y^{-2} \log\frac{|y|}{l} \Dl_y((x - y)^{-2}),
\vf(x,y) \Bigr>
\\
&\quad - \Bigl< x^{-2}\,\del^y_\bt \Bigl( y^{-2} \log\frac{|y|}{l}
\,\del_y^\bt((x - y)^{-2}) - (x - y)^{-2} \,\del_y^\bt \Bigl(
y^{-2} \log\frac{|y|}{l} \Bigr) \Bigr), \vf(x,y) \Bigr> .
\end{align*}

Observe again that the coefficient of $\log^2 l$ was foreordained: the
two-vertex tree minus half of the square of the one-vertex tree (here
the fish) is primitive, and so for the numerical coefficient of
$\log^2 l$ in~\eqref{eq:my-cup-overfloweth} we were bound to obtain
$\half\Om_4^2 = 2\pi^4$, which is correct: 
see~\eqref{eq:rminus4-one-log}.

We turn to the scale derivative for the winecup. \textit{Prima facie}
it yields:
\begin{equation}
\lddl \! \copadevino(x,y)
= - 2\pi^2 R_4[x^{-4}] \,\dl(x - y) - 2\pi^4 \,\dl(x)\,\dl(y)
+ 2\pi^2 x^{-4} \,\dl(x - y) - 2\pi^2 x^{-4} \,\dl(y).
\label{eq:bass-scale} % (4.9)
\end{equation}
The difference between the third and fourth terms above \textit{is} a
well-defined distribution, since
$$
\int \frac{\vf(x,x) - \vf(x,0)}{x^4} \,d^4x
$$
converges for any test function~$\vf$. We may reinterpret the $x^{-4}$
in~\eqref{eq:bass-scale} as $R_4[x^{-4}]$, since the corresponding
difference is still the same unique extension. So the scale derivative
becomes
\begin{align}
&-  2\pi^2 R_4[x^{-4}] \,\dl(x - y) - 2\pi^4 \,\dl(x) \,\dl(y)
+ 2\pi^2 R_4[x^{-4}] \,\dl(x - y) - 2\pi^2 R_4[x^{-4}] \,\dl(y)
\notag \\
&\quad = - 2\pi^2 \fish(x) \,\dl(y) - 2\pi^4 \,\dl(x)\,\dl(y).
\label{eq:now-you-see-it} % (4.10)
\end{align}

\medskip

Let us now reflect on what we have just done, in order to set forth
our methods. Understanding and fulfilment of the fundamental equation
\eqref{eq:in-annum} is essential. The winecup graph exemplifies it
well. There formula~\eqref{eq:prematur} expresses $R[\ga]$ for
$\sV(\ga) = \{0,y\}$. The cograph~$\Ga/\ga$ is a fish:
$$
\begin{tikzpicture}[scale=1.6]
\coordinate (A) at (0,0) ; \coordinate (B) at (1,0) ; 
\draw (A) parabola[bend pos=0.5] bend +(0,0.25) (B) ;
\draw (A) parabola[bend pos=0.5] bend +(0,-0.25) (B) ;
\draw (A) ++(-0.2,0.2) -- (A) -- ++(-0.2,-0.2) ;
\draw (B) ++(0.2,0.2) -- (B) -- ++(0.2,-0.2) ;
\draw (A) node[left=5pt] {$0=y$} ;
\draw (B) node[right=5pt] {$x$} ;
\foreach \pt in {A,B} \draw (\pt) node{$\bull$} ;
\end{tikzpicture}
$$
Now, the test function $\vf$ in~\eqref{eq:in-annum} is assumed to
\textit{vanish} on the thin diagonal $x = y = 0$. Then
\eqref{eq:in-annum} simply means
\begin{align*}
\duo< R_8\bigl[ x^{-2} (x - y)^{-2} y^{-4} \bigr], \vf(x,y)>
&= \duo< \Rbar_8\bigl[ x^{-2} (x - y)^{-2} y^{-4} \bigr], \vf(x,y)>
\\
&= \duo< \Rbar_4\bigl[ y^{-4} \bigr], x^{-2} (x - y)^{-2} \vf(x,y)>
\end{align*}
when $\vf(0,0) = 0$; and this is all we need to ask.%
\footnote{As also noted in~\cite[Remark~6.1]{Duetsch14}, our use of
renormalized expressions for divergent subgraphs avoids appeal to the
forest formula entirely.}

As we shall see next, similar procedures obeying the fundamental
formula~\eqref{eq:in-annum}, canvassing help from
Section~\ref{sec:convoluted} when necessary, allow one to compute all
the fourth-order contributions to the four-point function. It will
become clear that the scale derivative is related to the hierarchy of
cographs.

% \S 4.3
\subsection{Empirical remarks on the main theorem of renormalization}
\label{ssc:stora-story}

For the four-point function $\sG^4$, we are able to consider already
the contributions at orders $g$, $g^2$, $g^3$. From now on we adopt a
standard redefinition of the coupling constant:
$$
\g = \frac{g}{16\pi^2},
$$
which eliminates many $\pi$~factors. Thus, for the first-order
contribution, introducing a global minus sign as a matter of
convention:
$$
\sG^4_{(1)}(x_1,x_2,x_3,x_4;\g) \equiv \sG^4_\cruz(x_1,x_2,x_3,x_4)
= 16\pi^2 \g \,\dl(x_1 - x_2) \,\dl(x_2 - x_3) \,\dl(x_3 - x_4).
$$

Contributions of graphs of different orders to the four-point
function come with alternating signs \cite[Chap.~13]{KleinertSF01}.
Thus the fish diagram contribution $\sG^4_\pez$ is given by
\begin{align*}
& \sG^4_\pez(x_1,x_2,x_3,x_4;\g;l) 
= - \frac{g^2}{32\pi^4} \bigl( \dl(x_1 - x_2)\, 
R_4[(x_2 - x_3)^{-4};l]\,\dl(x_3 - x_4) + 2\text{ permutations} \bigr)
\\
&\quad = - 8\g^2 \bigl(
\dl(x_1 - x_2)\, R_4[(x_2 - x_3)^{-4};l] \,\dl(x_3 - x_4)
+ \dl(x_1 - x_3)\, R_4[(x_3 - x_4)^{-4};l] \,\dl(x_4 - x_2)
\\
&\hspace*{3.8em}
+ \dl(x_1 - x_4)\, R_4[(x_4 - x_2)^{-4};l] \,\dl(x_2 - x_3) \bigr).
\end{align*}
The practical rule to go from the scale derivatives of the graphs as
we have calculated them to their actual contributions to the
four-point function is simple: multiply the coefficient of the scale
derivative by $-\bar g/\pi^2$ raised to a power equal to the
difference in the number of vertices, and also by the relative weight,
for the diagrams in question.

A key point here, harking back to~\eqref{eq:spiny-scale} and recalling
that the fish has weight~$\frac{3}{2}$, is that
\begin{equation}
\lddl \sG^4_\pez = 3\g\,\sG^4_\cruz .
\label{eq:fish-scale} % (4.11)
\end{equation}

If we now define the second-order approximation
$\sG^4_{(2)} := \sG^4_\cruz + \sG^4_\pez$, we find that
\begin{align*}
& \sG^4_{(2)}(x_1,x_2,x_3,x_4;\g;l)
- \sG^4_{(2)}(x_1,x_2,x_3,x_4;\g;l')
\\
&\qquad = 48\pi^2 \g^2\,\log\frac{l}{l'}\,
\dl(x_1 - x_2) \,\dl(x_2 - x_3) \,\dl(x_3 - x_4).
\end{align*}
This of course means that
\begin{align*}
\sG^4_{(2)}(x_1,x_2,x_3,x_4;\g;l)
&= \sG^4_{(2)}\bigl(x_1,x_2,x_3,x_4;\Gbar_{(2)}^{\,ll'}(\g);l'\bigr),
\\
\word{where}
\Gbar^{\,ll'}_{(2)}(\g) &= \g + 3 \log\frac{l}{l'}\,\g^2 + O(\g^3).
\end{align*}

The bikini graph has a weight of~$3/4$ in the perturbation expansion.
We obtain
\begin{align*}
\sG^4_\sosten(x_1,x_2,x_3,x_4;\g;l) 
&= 16\,\g^3 \,\dl(x_1 - x_2) \biggl( R_4\biggl[ (x_2 - x_3)^{-4} 
\log\frac{|x_2 - x_3|}{l} \biggr]
\\
&\qquad - \frac{\pi^2}4 \,\dl(x_2 - x_3) \biggr) \,\dl(x_3 - x_4)
+ 2 \text{ permutations}.
\end{align*}
Here the two permutations have the same structure as those of the 
fish graph. Therefore,
\begin{equation}
\lddl \sG^4_\sosten = 2\g\,\sG^4_\pez ,
\label{eq:bikini-scale} % (4.12)
\end{equation}
coming from~\eqref{eq:fewer-fish} when all factors have been taken
into account, according to the rule explained above.

In the third-order approximation,
$$
\sG^4_{(3)}
:= \sG^4_\cruz + \sG^4_\pez + \sG^4_\sosten + \sG^4_\copadevino,
$$
we examine first the difference
$$
\sG^4_{\sosten}(x_1,x_2,x_3,x_4;\g;l) 
- \sG^4_{\sosten} \bigl( x_1,x_2,x_3,x_4;\g;l' \bigr).
$$
{}From~\eqref{eq:rminus4-one-log} one sees that
$$
R_4 \Bigl[ r^{-4} \log \frac{r}{l} \Bigr]
- R_4 \Bigl[ r^{-4} \log \frac{r}{l'} \Bigr]
= \pi^2(\log^2 l - \log^2 l')\,\dl(r) 
- \log \frac{l}{l'}\, R_4[r^{-4};1],
$$
so the difference may be rewritten as
\begin{align}
& 48\,\g^3\pi^2 \bigl( \log^2 l - \log^2 l' \bigr) 
\,\dl(x_1 - x_2) \,\dl(x_2 - x_3) \,\dl(x_3 - x_4)
\nonumber \\
&\quad
-16\,\g^3 \log\frac{l}{l'} \bigl(
\dl(x_1 - x_2)\, R_4[(x_2 - x_3)^{-4};1] \,\dl(x_3 - x_4) 
+ 2 \text{ permutations} \bigr).
\label{eq:cup-sizes} % (4.13)
\end{align}

The winecup graph enters with weight~$3$ in the perturbation
expansion. In detail:
\begin{align}
& \sG^{(4)}_\copadevino(x_1,x_2,x_3,x_4;\g;l)
= 16\g^3 \,\dl(x_1 - x_2) \biggl( R_4\biggl[ (x_2 - x_4)^{-4}
\log\frac{|x_2 - x_4|}{l} \biggr] \,\dl(x_2 - x_3)
\nonumber \\
&\ + \frac{1}{2} R_4\bigl[ (x_2 - x_4)^{-4};l \bigr] \dl(x_3 - x_4)
+ \frac{(x_2 - x_4)^{-2}}{4\pi^2} \,\del_{x_3}^\bt \biggl(
(x_3 - x_4)^{-2} \log\frac{|x_3 - x_4|}{l}
\,\del_{x_3}^\bt(x_2 - x_3)^{-2}
\nonumber \\
&\ - (x_2 - x_3)^{-2} \,\del_{x_3}^\bt \Bigl(
(x_3 - x_4)^{-2} \log \frac{|x_3 - x_4|}{l} \Bigr) \biggr)
+ R_4\biggl[ (x_2 - x_3)^{-4} \log\frac{|x_2 - x_3|}{l} \biggr]
\,\dl(x_2 - x_4)
\nonumber \\
&\ + \frac{1}{2} R_4\bigl[ (x_2 - x_3)^{-4};l \bigr] \dl(x_3 - x_4)
+ \frac{(x_2 - x_3)^{-2}}{4\pi^2} \,\del_{x_4}^\bt \biggl(
(x_3 - x_4)^{-2} \log\frac{|x_3 - x_4|}{l} 
\,\del_{x_4}^\bt(x_2 - x_4)^{-2}
\nonumber \\
&\ - (x_2 - x_4)^{-2} \,\del_{x_4}^\bt \Bigl(
(x_3 - x_4)^{-2} \log \frac{|x_3 - x_4|}{l} \Bigr) \biggr) \biggr)
+ 2 \text{ permutations of each}.
\label{eq:wine-tasting} % (4.14)
\end{align}
The scale derivative can be obtained either
from~\eqref{eq:now-you-see-it} by applying the conversion rule, or now
directly from \eqref{eq:wine-tasting}:
\begin{equation}
\lddl \sG^4_\copadevino = 4\g\,\sG^4_\pez - 6\g^2\,\sG^4_\cruz .
\label{eq:winecup-scale} % (4.15)
\end{equation}

The difference between the winecup expansion~\eqref{eq:wine-tasting} 
at two scales has several contributions. From the
$R_4\bigl[ r^{-4} \log(r/l) \bigr]$ terms we collect
$$
-16\,\g^3 \log\frac{l}{l'} \bigl( \dl(x_1 - x_2)\,
R_4[(x_2 - x_4)^{-4};1] \,\dl(x_2 - x_3)
+ \dl(x_1 - x_2)\, R_4[(x_2 - x_3)^{-4};1] \,\dl(x_2 - x_4)\bigr),
$$
plus its two permutations. None of those is of the same type as those
coming from~\eqref{eq:cup-sizes}. Nevertheless, from the divergence
part of~\eqref{eq:wine-tasting} we recover
$$
16\,\g^3 \log\frac{l}{l'} \biggl( \dl(x_1 - x_2)
R_4\bigl[ (x_2 - x_4)^{-4};1 \bigr] \dl(x_2 - x_3) \biggr),
$$
and similar terms, \textit{cancelling all} of the previous terms. The 
divergence part also produces terms like the second line
of~\eqref{eq:cup-sizes}, with a doubled coefficient. The remaining
summand $8\g^3 \,\dl(x_1 - x_2) R_4\bigl[ (x_2 - x_4)^{-4};l \bigr] 
\,\dl(x_3 - x_4)$, and its permutations, generate
$$
-96\pi^2\g^3 \log\frac{l}{l'}\,
\dl(x_1 - x_2) \,\dl(x_2 - x_3) \,\dl(x_3 - x_4).
$$

Putting together all the contributions up to third order, we arrive at
\begin{align*}
& \sG^4_{(3)}(x_1,x_2,x_3,x_4;\g;l) 
- \sG^4_{(3)}(x_1,x_2,x_3,x_4;\bar g;l'\,)
\\
&\quad = 48\pi^2 \g^2\, \log\frac{l}{l'} \,\dl(x_1 - x_2)
\,\dl(x_2 - x_3) \,\dl(x_3 - x_4)
\\
&\qquad - 144\bar g^3 \log\frac{l}{l'} \biggl( \dl(x_1 - x_2)
R_4\bigl[ (x_2 - x_4)^{-4};1 \bigr] \,\dl(x_2 - x_3) \biggr)
\\
&\qquad -96\pi^2\g^3 \log\frac{l}{l'}\, \dl(x_1 - x_2)
\,\dl(x_2 - x_3) \,\dl(x_3 - x_4)
\\
&\qquad + 144\,\g^3\pi^2 \bigl( \log^2 l - \log^2 l' \bigr)
\,\dl(x_1 - x_2) \,\dl(x_2 - x_3) \,\dl(x_3 - x_4).
\end{align*}
Now, we look for a coefficient $\al$ in
$$
\Gbar^{\,ll'}_{(3)}(\g)
= \g + 3 \log\frac{l}{l'}\,\g^2 + \al\,\g^3 + O(\g^4),
$$
such that 
\begin{equation}
\sG^4_{(3)}(x_1,x_2,x_3,x_4;\g;l)
= \sG^4_{(3)}\bigl( x_1,x_2,x_3,x_4;\Gbar_{(3)}^{\,ll'}(\g);l' \bigr)
+ O(\g^4).
\label{eq:main-thm-three} % (4.16)
\end{equation}
We see that the term in $\bar g^2$ does not need revisiting. Of
course, general theory ensures that. Next,
\begin{align*}
& \sG^4_{(3)}\bigl(x_1,x_2,x_3,x_4;\Gbar_{(3)}^{\,ll'}(\g);l'\bigr) 
\\
&\quad = \sG^4_{(3)}(x_1,x_2,x_3,x_4;\g;l)
+ 16\pi^2\al\bar g^3 \,\dl(x_1 - x_2)\,\dl(x_2 - x_3)\,\dl(x_3 - x_4)
\\
&\qquad - 144\,\g^3\pi^2 (\log^2 l - \log^2 l'\,) \,\dl(x_1 - x_2)
\,\dl(x_2 - x_3) \,\dl(x_3 - x_4)
\\
&\qquad + 96\pi^2\g^3 \log\frac{l}{l'} \,\dl(x_1 - x_2)
\,\dl(x_2 - x_3) \,\dl(x_3 - x_4) + O(\g^4).
\end{align*}
Therefore, the relation~\eqref{eq:main-thm-three} is verified by
taking
$$
\Gbar^{\,ll'}_{(3)}(\g)
= \g + 3 \log\frac{l}{l'}\,\g^2 + \biggl( 9\,(\log^2 l - \log^2 l'\,)
- 6 \log\frac{l}{l'}  \biggr) \,\g^3  + O(\g^4).
$$

At this very humble level, this illustrates the Popineau--Stora ``main
theorem of renormalization'' \cite{PopineauS82} -- as applied to the
effective action. To wit, there exists a formal power series
$\Gbar(\g)$, tangent to the identity, that effects the change between
any two renormalization recipes.%
\footnote{When one allows $\g$ to become a test function $g(x)$, that
series is local in the coordinates, although it may depend on the
derivatives of~$g(x)$.}
Actually, there is more to the theorem than was allowed 
in~\cite{PopineauS82}. One can write
$$
\sG^4_{(n)}(x_1,x_2,x_3,x_4;\g;l)
= \sG^4_{(n)}\bigl(x_1,x_2,x_3,x_4;\Gbar_{(N-1)}^{\,ll'}(\g);l'\bigr)
+ O(\g^{n+1}),
$$
for any $n \leq N - 1$. Here $\Gbar_{(N-1)}^{\,ll'}$ need only be
taken up to order~$n$. Let
$$
\Gbar_{(N-1)}^{\,ll'}(\g) = \g + H_2^{\,ll'}(\g) + H_3^{\,ll'}(\g)
+\cdots+ H_{N-1}^{\,ll'}(\g) =: \g + H_{(N-1)}^{\,ll'}(\g).
$$
Here each $H_n^{\,ll'}$ comes from a distribution with support on
the corresponding thin diagonal \textit{exclusively}. Then
$$
\Gbar_{(N)}^{\,ll'} = \Gbar_{(N-1)}^{\,ll'}
+ H_{(N)}^{\,ll'}\bigl( \Gbar_{(N-1)}^{\,ll'} \bigr),
$$
where only terms up to order $N$ in the last expansion need be taken.
One may let $N \upto \infty$, obtaining the ``tautological'' identity
$$
\Gbar^{\,ll'} = \id + H\bigl( \Gbar^{\,ll'} \bigr).
$$
This is what Stora understands by the Bogoliubov recursion relation
for the coupling constant -- apparently a deeper fact than the
Bogoliubov recursion for the graphs~\cite{Stora08,RaymondCERNFeb13}.

\medskip

For future reference, we report here the weight factors of the graphs
contributing to $\sG^4$ at order~$g^4$. For the sets:
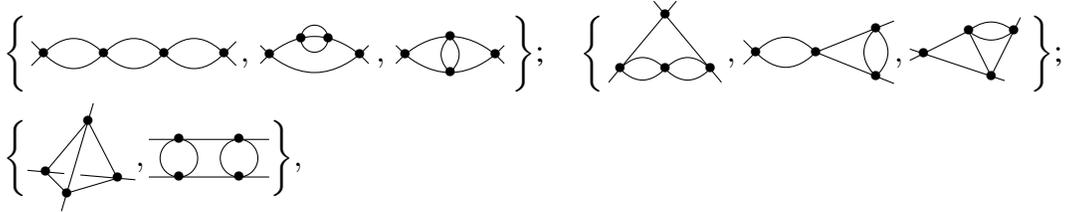
\begin{align*}
& \biggl\{ \trikini, \stye, \catseye \biggr\}; \quad
\biggl\{ \duncecap, \shark, \kite \biggr\};
\\
& \biggl\{ \tetrahedron, \roll \biggr\},
\end{align*}
these numbers are respectively given by
$\{\frac{3}{8}, \frac{1}{2}, \frac{3}{4}\}$; 
$\{\frac{3}{2}, \frac{3}{2}, 6\}$; $\{1, \frac{3}{2}\}$.

% \S 5
\section{More graphs}
\label{sec:graphic-arts}

In this section, we consider each one of the eight graphs with four
vertices required to compute the four-point function. We discuss
first the three graphs with two external vertices, then those three
with three external vertices, and finally those two with four external
vertices. The last of these do not require convolution-like
operations.

% \S 5.1
\subsection{The trikini}
\label{ssc:three-eyed}

The ``trikini'' is a fourth-order, three-loop chain graph, that is a
convolution cube. The quickest method is to pass to multiplication in
$p$-space, using~\eqref{eq:log-roll} from 
Appendix~\ref{app:pspace-extens}:
\begin{align}
\trikini &= \bigl( R_4[x^{-4}] \bigr)^{*3} 
\notag \\
&= 12\pi^4 R_4\Bigl[ x^{-4} \log^2\frac{|x|}{l} \Bigr]
- 3\pi^4 R_4[r^{-4}] + (4\zeta(3) - 2)\pi^6 \,\dl(x).
\label{eq:golden-braid} % (5.1)
\end{align}
In general, the amplitude of a chain graph with $(n + 1)$ vertices in
$p$-space is given by
\begin{equation}
\pi^{2n} \biggl( 1 - 2\log\frac{|p|}{\La} \biggr)^n
\label{eq:le-canard-enchaine} % (5.2)
\end{equation}
where $\La = 2e^{-\ga}/l$. The result is readily transferable to
$x$-space, using the formulas of Appendix~\ref{app:pspace-extens} to
invert the Fourier transforms.

We note that with the differential renormalization method
of~\cite{FreedmanJL92} only $\bigl( 2\log(|p|/\La) \bigr)^n$ is
computed. The problem is compounded by a mistake in their Fourier
transform formula, whose origin is dealt with in the Appendix.

We easily compute as usual the scale derivative:
\begin{align*}
\lddl  \trikini &= -6\pi^2 \bikini,
\word{translating into}
\\
\lddl \sG^4_\trenza &= 3\g\,\sG^4_\sosten,
\end{align*}
since the weight of the $\trikini$ graph is half the weight of the
$\bikini$ graph. In general, the scale derivative of the contribution
to $\sG^4$ of a chain graph with $n$ bubbles equals $n\g$~times
that of the graph with $(n - 1)$ bubbles. Indeed, it follows at once 
from~\eqref{eq:spiny-scale} or alternatively 
from~\eqref{eq:le-canard-enchaine} that 
$$
\lddl \bigl( R_4[x^{-4}] \bigr)^{*n}
= -2n\pi^2 \bigl( R_4[x^{-4}] \bigr)^{*(n-1)},
$$
and the aforementioned practical rule gives the result for the
contributions to~$\sG^4$.

% \S 5.2
\subsection{The stye}
\label{ssc:stye}

The ``partially renormalized'' amplitude for the stye diagram, 
labelled as follows:
$$
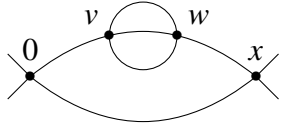
\begin{tikzpicture}[scale=1.5]
\coordinate (A) at (0,0) ; \coordinate (B) at (2,0) ; 
\coordinate (C) at (0.7,0.36) ; \coordinate (D) at (1.3,0.36) ; 
\coordinate (E) at ($ (C)!0.5!(D) $) ; 
\draw (A) ++(-0.2,0.2) -- (A) -- ++(-0.2,-0.2) ;
\draw (A) parabola[bend pos=0.5] bend +(0,0.4) (B) ;
\draw (A) parabola[bend pos=0.5] bend +(0,-0.4) (B) ;
\draw (B) ++(0.2,0.2) -- (B) -- ++(0.2,-0.2) ;
\draw (E) circle(0.3cm) ;
\draw (A) node[above=2pt] {$0$} ;
\draw (B) node[above=2pt] {$x$} ;
\draw (C) node[above left] {$v$} ;
\draw (D) node[above right] {$w$} ;
\foreach \pt in {A,B,C,D} \draw (\pt) node{$\bull$} ;
\end{tikzpicture}
$$
is of the form
$$
x^{-2} \iint \Rbar_{12} \bigl[
v^{-2} (v - w)^{-6} (w - x)^{-2} \bigr] \,dv\,dw
= x^{-2} \iint \Rbar_{12} \bigl[
v^{-2} u^{-6} (v - u - x)^{-2} \bigr] \,du\,dv.
$$
Notice that this is a nested convolution; the inner integral is of the
form $R_4[r^{-6}] * r^{-2}$, which \textit{exists} by the theory of
Section~\ref{sec:convoluted}. On account of~\eqref{eq:ren-hex}, the
integral becomes
$$
-\frac{1}{16}\, x^{-2} \iint v^{-2} (v - u - x)^{-2} \biggl(
\Dl^2 \Bigl( u^{-2} \log\frac{|u|}{l} \Bigr)  - 5\pi^2 \,\Dl\dl(u) 
\biggr)  \,dv\,du.
$$
On integrating by parts with~\eqref{eq:green-backs} and dropping total
derivatives in the integrals over internal vertices, we then obtain
\begin{align*}
& \frac{\pi^2}{4}\, x^{-2} \iint v^{-2} \biggl(
\Dl \Bigl( u^{-2} \log\frac{|u|}{l} \Bigr) - 5\pi^2 \,\dl(u) \biggr)
\,\dl(v - u - x) \,dv\,du
\\
&\quad = \frac{\pi^2}{4}\, x^{-2} \int (u + x)^{-2} \biggl(
\Dl\Bigl( u^{-2} \log\frac{|u|}{l} \Bigr)
- 5\pi^2 \,\dl(u) \biggr) \,du
\\
&\quad = - \pi^4 x^{-2} \int u^{-2} \log\frac{|u|}{l} \,\dl(u + x)\,du
- \frac{5\pi^4}{4}\, x^{-4}
= - \pi^4 x^{-4} \log\frac{|x|}{l} - \frac{5\pi^4}{4}\, x^{-4}.
\end{align*}

The fully renormalized amplitude for the stye graph is then simply
given by
$$
\stye = - \pi^4 R_4\Bigl[ x^{-4} \log\frac{|x|}{l} \Bigr]
- \frac{5\pi^4}{4}\, R_4[x^{-4}].
$$

Therefore
\begin{align*}
\lddl \stye &= \pi^4 \fish + \frac{5\pi^6}{2} \cross,
\word{translating into}
\\
\lddl \sG^4_\orzuelo
&= \frac{1}{3}\,\g^2\,\sG^4_\pez - \frac{5}{4}\, \g^3\,\sG^4_\cruz.
\end{align*}

% \S 5.3
\subsection{The cat's eye}
\label{ssc:cats-eye}

The ``cat's eye'' graph, which we label as follows:
$$
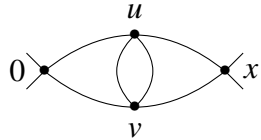
\begin{tikzpicture}[scale=1.2]
\coordinate (A) at (0,0) ; \coordinate (B) at (2,0) ; 
\coordinate (C) at (1,0.4) ; \coordinate (D) at (1,-0.4) ; 
\draw (A) ++(-0.2,0.2) -- (A) -- ++(-0.2,-0.2) ;
\draw (A) parabola[bend pos=0.5] bend +(0,0.4) (B) ;
\draw (A) parabola[bend pos=0.5] bend +(0,-0.4) (B) ;
\begin{scope}[rotate=90]
\draw (C) parabola[bend pos=0.5] bend +(0,0.2) (D) ;
\draw (C) parabola[bend pos=0.5] bend +(0,-0.2) (D) ;
\end{scope}
\draw (B) ++(0.2,0.2) -- (B) -- ++(0.2,-0.2) ;
\draw (A) node[left=3pt] {$0$} ;
\draw (B) node[right=3pt] {$x$} ;
\draw (C) node[above=2pt] {$u$} ;
\draw (D) node[below=2pt] {$v$} ;
\foreach \pt in {A,B,C,D} \draw (\pt) node{$\bull$} ;
\end{tikzpicture}
$$
is sometimes counted as an ``overlapping divergence'' in $p$-space.
But for renormalization on configuration space, this problem is more
apparent than real: ``the external points can be kept separated until
the regularization of subdivergences is accomplished''
\cite[Sect.~3.3]{FreedmanJL92}. For the same reasons, when dealing
with this graph we find it unnecessary to bring in partitions of unity
for overlapping divergences \cite[Example~4.16]{DuetschFKR13}.

There are two internal vertices; its ``bare'' amplitude is of the form
$$
f(x)
= \iint u^{-2} v^{-2} (u - x)^{-2} (v - x)^{-2} (u - v)^{-4} \,du\,dv.
$$

A first natural rewriting is
\begin{align*}
f(x) &\renormto
- \frac{1}{2} \iint u^{-2} v^{-2} (u - x)^{-2} (v - x)^{-2} 
 \,\Dl_u \Bigl( (u - v)^{-2} \log\frac{|u - v|}{l} \Bigr) \,du\,dv
\\
&\qquad\ + \pi^2 \iint u^{-2} v^{-2} (u - x)^{-2} (v - x)^{-2}
\,\dl(u - v) \,du\,dv.
\end{align*}
{}From now on, we shall use the notation $\renormto$ to denote a single
step in a sequence of one of more partial renormalizations, by
replacements
$r^{-d-2m}\log^k(r/l) \renormto R_d\bigl[r^{-d-2m}\log^k(r/l)\bigr]$.
That is \textit{not yet} $\Rbar_4[f(x)]$, since there are other
untreated subdivergences in the diagram. The second term in this
expression, however, is just the bikini convolution integral, which
becomes
$$
\pi^2 \int v^{-4} (v - x)^{-4} \,dv
\renormto 4\pi^4 R_4\Bigl[ x^{-4} \log\frac{|x|}{l} \Bigr]
- \pi^6\,\dl(x).
$$

The first term can be simplified by an integration by parts, to get
\begin{align}
& 2\pi^2 x^{-2} \int (v - x)^{-2} v^{-4} \log\frac{|v|}{l} \,dv
+ 2\pi^2 x^{-2} \int v^{-2} (v - x)^{-4} \log\frac{|v-x|}{l} \,dv
\nonumber \\
&\quad - \iint v^{-2} (v - x)^{-2} \,\del_u^\bt(u^{-2})
\,\del^u_\bt((u - x)^{-2})\, (u - v)^{-2} \log\frac{|u - v|}{l}
\,du\,dv.
\label{eq:cats-paw} % (5.3)
\end{align}
The first two of these integrals are equal. To deal with the
$v \sim 0$ region, we proceed as before:
\begin{align*}
& 4\pi^2 x^{-2} \int (v - x)^{-2} v^{-4} \log\frac{|v|}{l} \,dv
\\
&\quad \renormto - \pi^2 x^{-2} \int (v - x)^{-2} \biggl(
\Dl\Bigl( v^{-2} \log^2 \frac{|v|}{l} \Bigr) 
+ \Dl\Bigl( v^{-2} \log \frac{|v|}{l} \Bigr)
- 2\pi^2\,\dl(v) \biggr) \,dv
\\
&\quad \renormto 4\pi^4 R_4\Bigl[ x^{-4} \log^2\frac{|x|}{l} \Bigr]
+ 4\pi^4 R_4\Bigl[ x^{-4} \log \frac{|x|}{l} \Bigr]
+ 2\pi^4 R_4[x^{-4}].
\end{align*}

The third term in~\eqref{eq:cats-paw} is only divergent
\textit{overall}, at $x = 0$. After rescaling the integrand by
$u \mapsto |x|s$, $v \mapsto |x|t$, and setting $x = |x|\om$, this
term takes the form
$$
- c_1 \pi^4\, x^{-4} \log \frac{|x|}{l} - c_2 \pi^4\, x^{-4},
$$
where
\begin{align*}
c_1 &= \frac{1}{\pi^4} \iint t^{-2} (t - \om)^{-2} \,\del_s^\bt(s^{-2})
\,\del^s_\bt((s - \om)^{-2})\, (s - t)^{-2} \,ds\,dt = 4,
\\
c_2 &= \frac{1}{\pi^4} \iint t^{-2} (t - \om)^{-2} \,\del_s^\bt(s^{-2})
\,\del^s_\bt((s - \om)^{-2})\, (s - t)^{-2} \log\frac{|s - t|}{l}
\,ds\,dt = 4.
\end{align*}
These are computed straightforwardly, if tediously, by use of
Gegenbauer polynomials~\cite{FreedmanJL92}. Thus, this third term
yields
$$
- 4\pi^4 R_4\Bigl[ x^{-4} \log \frac{|x|}{l} \bigr]
- 4\pi^4 R_4[x^{-4}].
$$

Putting it all together, we arrive at
$$
\catseye
= 4\pi^4 R_4\Bigl[ x^{-4} \log^2\frac{|x|}{l} \Bigr]
+ 4\pi^4 R_4\Bigl[ x^{-4} \log\frac{|x|}{l} \Bigr]
- 2\pi^4 R_4[x^{-4}] - \pi^6\,\dl(x).
$$
With that, we obtain
\begin{align*}
\lddl  \catseye
&= -8\pi^4 R_4\Bigl[ x^{-4} \log\frac{|x|}{l} \Bigr]
- 4\pi^4 R_4[x^{-4}] + 4\pi^6\,\dl(x)
\\
&= -2\pi^2 \bikini - 4\pi^4 \fish + 2\pi^6 \,\dl(x);
\word{translating into}
\\
\lddl  \sG^4_\ojodegato
&= 2\g\,\sG^4_\sosten - 2\g^2\,\sG^4_\pez
- \frac{3}{2}\,\g^3\,\sG^4_\cruz.
\end{align*}

% \S 5.4
\subsection{The duncecap}
\label{ssc:the-dunciad}

Consider next the ``duncecap'', which contains a bikini subgraph:
$$
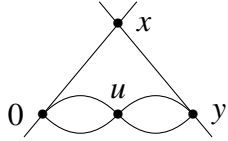
\begin{tikzpicture}
\coordinate (A) at (0,0) ; \coordinate (B) at (1,0) ; 
\coordinate (C) at (2,0) ; \coordinate (D) at (1,1.2) ; 
\draw ($ (A)!-0.25!(D) $) -- ($ (A)!1.25!(D) $) ;
\draw ($ (C)!-0.25!(D) $) -- ($ (C)!1.25!(D) $) ;
\draw (A) parabola[bend pos=0.5] bend +(0,0.25) (B) ;
\draw (A) parabola[bend pos=0.5] bend +(0,-0.25) (B) ;
\draw (B) parabola[bend pos=0.5] bend +(0,0.25) (C) ;
\draw (B) parabola[bend pos=0.5] bend +(0,-0.25) (C) ;
\draw (A) node[left=3pt] {$0$} ;
\draw (B) node[above=2pt] {$u$} ;
\draw (C) node[right=3pt] {$y$} ;
\draw (D) node[right=3pt] {$x$} ;
\foreach \pt in {A,B,C,D} \draw (\pt) node{$\bull$} ;
\end{tikzpicture}
$$
The unrenormalized amplitude is given by
$$
f(x,y) = x^{-2} (x - y)^{-2} (y^{-4})^{*2}.
$$
Once again, we may partially regularize this, 
using~\eqref{eq:bikini-ren}, to get
\begin{align*}
\Rbar_8 \bigl[ x^{-2} (x - y)^{-2} (y^{-4})^{*2} \bigr]
&= \pi^2 x^{-2} (x - y)^{-2} \biggl(
4 R_4\Bigl[ y^{-4} \log\frac{|y|}{l} \Bigr] - \pi^2\,\dl(y) \biggr),
\end{align*}
which is well defined for $(x,y) \neq (0,0)$. We want to apply the
integration by parts formula~\eqref{eq:green-backs} here, and we
invoke \eqref{eq:rminus4-one-log} for the purpose:
\begin{align*}
\Rbar_8 \bigl[ x^{-2} (x - y)^{-2} (y^{-4})^{*2} \bigr]
&= -\pi^2 x^{-2} (x - y)^{-2}
\,\Dl\Bigl( y^{-2} \log^2 \frac{|y|}{l} \Bigr)
\\
&\quad - \pi^2 x^{-2} (x - y)^{-2}
\,\Dl\Bigl( y^{-2} \log\frac{|y|}{l} \Bigr) + \pi^4 x^{-4} \,\dl(y),
\end{align*}
which leads easily to the renormalized version:
\begin{align*}
R_8 \bigl[ x^{-2} (x - y)^{-2} (y^{-4})^{*2} \bigr]
&= 4\pi^4\, R_4 \Bigl[ x^{-4} \log^2 \frac{|x|}{l}
+ x^{-4} \log \frac{|x|}{l} \Bigr] \,\dl(x - y)
+ \pi^4 R_4[x^{-4}] \,\dl(y)
\\
&\quad + \pi^2 x^{-2} \del_y^\bt \bigl(
L_\bt(y; x - y) + M_\bt(y; x - y) \bigr),
\end{align*}
where we call upon \eqref{eq:total-deriv-term1} and a companion
formula:
\begin{equation}
M_\bt(y; x - y)
:= y^{-2} \log^2 \frac{|y|}{l}\, \del^y_\bt((x - y)^{-2})
- (x - y)^{-2} \,\del^y_\bt \Bigl( y^{-2} \log^2 \frac{|y|}{l} \Bigr).
\label{eq:total-deriv-term2} % (5.4)
\end{equation}

The Green formula \eqref{eq:green-backs} easily yields
$$
\lddl \bigl[ \del_y^\bt L_\bt(y; x - y) \bigr]
= 4\pi^2 x^{-2} \bigl( \dl(x - y) - \dl(y) \bigr),
$$
and clearly $\lddl \bigl[ \del_y^\bt M_\bt(y; x - y) \bigr]
= - 2 \,\del_y^\bt L_\bt(y; x - y)$. From this, we obtain for the
scale derivative, in the same way as for the winecup:
\begin{align*}
\lddl  \duncecap
&= - 8\pi^4 R_4\Bigl[ x^{-4} \log \frac{|x|}{l} \Bigr]\,\dl(x - y)
- 4\pi^4 R_4[x^{-4}]\,\dl(x - y) - 2\pi^6 \,\dl(x)\,\dl(y)
\\
&\quad + 4\pi^4 R[x^{-4}]\,\dl(x - y) - 4\pi^4 R[x^{-4}]\,\dl(y)
- 2\pi^2\, x^{-2} \,\del_y^\bt L_\bt(y; x - y)
\\
&= - 4\pi^2 \copadevino(x,y) - 2\pi^6 \,\dl(x)\,\dl(y); 
\end{align*}
yielding, after the usual manipulation,
$$
\lddl \sG^4_\casquete = 2\g\,\sG^4_\copadevino + 3\g^3\,\sG^4_\cruz.
$$

% \S 5.5
\subsection{The kite}
\label{ssc:kite}

The \textit{kite} graph has the following structure, showing an
internal winecup:
$$
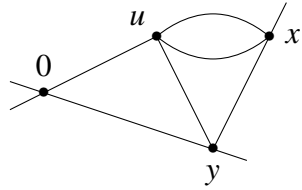
\begin{tikzpicture}[scale=1.5]
\coordinate (A) at (1,1) ; \coordinate (B) at (2,1) ; 
\coordinate (C) at (1.5,0) ; \coordinate (D) at (0,0.5) ; 
\draw ($ (C)!1.3!(B) $) -- (C) -- (A) -- ($ (A)!1.3!(D) $) ;
\draw ($ (C)!1.2!(D) $) -- ($ (D)!1.2!(C) $) ;
\draw (A) parabola[bend pos=0.5] bend +(0,0.2) (B) ;
\draw (A) parabola[bend pos=0.5] bend +(0,-0.2) (B) ;
\draw (A) node[above left] {$u$} ;
\draw (B) node[right=2pt] {$x$} ;
\draw (C) node[below=2pt] {$y$} ;
\draw (D) node[above=2pt] {$0$} ;
\foreach \pt in {A,B,C,D} \draw (\pt) node{$\bull$} ;
\end{tikzpicture}
$$
The labelling is that recommended by~\cite{FreedmanJL92}.

One may anticipate the coefficients of logarithmic degrees $3$ and~$2$
in the final result, by invoking again \cite{ChryssomalakosQRV02},
particularly its Section~4.2, and \cite{Kreimer01, HeidyMSc06}. We
focus on the third-degree coefficients. Notice that the kite is a
rooted-tree graph, actually a ``stick'', with the fish as unique
``decoration''. Subtracting from it the disconnected juxtaposition of
the winecup and the fish graph, and adding one third of the product of
three fishes, one obtains a primitive graph in the cocommutative
bialgebra of sticks. For this combination the dilation coefficient of
$\log^3 l$ must vanish. We obtain then for this graph the coefficient
$2\pi^4 \x 2\pi^2 - 8\pi^6/3 = 4\pi^6/3$.

Starting from the bare amplitude
$$
f(x,y) = y^{-2} (x - y)^{-2} \int u^{-2}(u - y)^{-2}(u - x)^{-4} \,du,
$$
partial renormalization gives
\begin{align*}
& f(x,y) \renormto y^{-2} (x - y)^{-2} \int u^{-2} (u - y)^{-2} 
R_4\bigl[ (u - x)^{-4} \bigr] \,du
\\
&\quad = - \frac{1}{2} y^{-2} (x - y)^{-2} \int u^{-2} (u - y)^{-2} 
\,\Dl_x\Bigl( (x - u)^{-2} \log\frac{|x - u|}{l} \Bigr) \,du
+ \pi^2 x^{-2} y^{-2} (x - y)^{-4}.
\end{align*}
The second summand is a winecup: namely, the cograph obtained on
contracting the $u$--$x$ fish. It contributes
\begin{align}
& \pi^2 x^{-2} y^{-2}\, R_4\bigl[ (x - y)^{-4} \bigr]
= - \frac{\pi^2}{2} x^{-2} y^{-2} 
\,\Dl_y \Bigl( (x - y)^{-2} \log \frac{|x - y|}{l} \Bigr)
+ \pi^4 x^{-4} \,\dl(x - y)
\nonumber \\
&\qquad = 2 \pi^4 x^{-4} \log \frac{|x|}{l} \,\dl(y)
+ \pi^4 x^{-4} \,\dl(x - y) 
- \frac{\pi^2}{2}\, x^{-2} \,\del_y^\bt L_\bt(x - y; y)
\nonumber \\
&\renormto 2 \pi^4 R_4\Bigl[ x^{-4} \log\frac{|x|}{l} \Bigr] \,\dl(y)
+ \pi^4 R_4[x^{-4}] \,\dl(x - y) 
- \frac{\pi^2}{2}\, x^{-2}  \,\del_y^\bt L_\bt(x - y; y).
\label{eq:kite-first} % (5.5)
\end{align}

There remains the term
$- \frac{1}{2} y^{-2} (x - y)^{-2} \,\Dl_x \int u^{-2} (u - y)^{-2} 
(x - u)^{-2} \log\frac{|x - u|}{l} \,du$. Following a suggestion
of~\cite{FreedmanJL92}, we may write
\begin{equation}
\log\frac{|x - u|}{l}
= \log\frac{|x - u|}{|y - u|} + \log\frac{|y - u|}{l} \,.
\label{eq:log-splitting} % (5.6)
\end{equation}
The second summand above contributes
\begin{align}
& - \frac{1}{2}\, y^{-2} (x - y)^{-2} \int u^{-2} (u - y)^{-2} 
\,\Dl_x \bigl( (x - u)^{-2} \bigr) \log\frac{|y - u|}{l} \,du
\nonumber \\
&\qquad = 2\pi^2 x^{-2} y^{-2} (x - y)^{-4}  \log\frac{|x - y|}{l}
\nonumber \\
&\renormto - \frac{\pi^2}{2} x^{-2} y^{-2}
\,\Dl_x \Bigl( (x - y)^{-2} \log^2 \frac{|x - y|}{l} 
+ (x - y)^{-2} \log \frac{|x - y|}{l}\Bigr) 
+ \pi^4 x^{-4} \,\dl(x - y)
\nonumber \\
&\qquad = 2\pi^4 \Bigl( y^{-4} \log^2 \frac{|y|}{l}
+ y^{-4} \log \frac{|y|}{l}\Bigr) \,\dl(x) + \pi^4 x^{-4} \,\dl(x - y)
\nonumber \\
&\hspace*{4em} - \frac{\pi^2}{2}\, y^{-2} 
\,\del_x^\bt \bigl( L_\bt(x - y; x) + M_\bt(x - y; x) \bigr)
\nonumber \\
&\renormto 2\pi^4 R_4\Bigl[ y^{-4} \log^2 \frac{|y|}{l}
+ y^{-4} \log \frac{|y|}{l} \Bigr] \,\dl(x)
+ \pi^4 R_4[x^{-4}] \,\dl(x - y)
\nonumber \\
&\hspace*{4em} - \frac{\pi^2}{2}\, y^{-2}
\,\del_x^\bt \bigl( L_\bt(x - y; x) + M_\bt(x - y; x) \bigr).
\label{eq:kite-second} % (5.7)
\end{align}

It helps to introduce the integral:
\begin{align*}
\int u^{-2} (x - u)^{-2} (y - u)^{-2} \log\frac{|x - u|}{|y - u|} \,du
&= \frac{1}{2}\, \log \frac{|x|}{|y|}
\int u^{-2} (x - u)^{-2} (y - u)^{-2} \,du
\\
&=: \frac{1}{2}\, \log \frac{|x|}{|y|}\, K(x,y).
\end{align*}
The first equality above is obtained by easy symmetry arguments. Thus
the remaining term of the partially renormalized expression for the
kite is of the form
$$
-\frac{1}{4} y^{-2} (x - y)^{-2}
\,\Dl_x\Bigl( K(x,y) \log \frac{|x|}{|y|} \Bigr).
$$
Integration by parts once more expands this to
\begin{align*}
& \pi^2 y^{-2} K(x,y) \log\frac{|x|}{|y|} \,\dl(x - y)
\\
&\quad + \frac{1}{4}\, y^{-2} \,\del_x^\bt \biggl(
\del^x_\bt((x - y)^{-2}) K(x,y) \log\frac{|x|}{|y|} 
- (x - y)^{-2} \del^x_\bt \Bigl( K(x,y) \log\frac{|x|}{|y|} \Bigr)
\biggr).
\end{align*}
The first term vanishes since $K(x,y) \log\bigl( |x|/|y| \bigr)$ is
skewsymmetric; and only the total derivative part remains.

Combining this total derivative with the other contributions
\eqref{eq:kite-first} and~\eqref{eq:kite-second}, we arrive at the
renormalized amplitude for the kite:
\begin{align*}
\kite &= 2\pi^4 R_4\Bigl[ y^{-4} \log^2 \frac{|y|}{l} \Bigr] \,\dl(x)
+ 2\pi^4 R_4\Bigl[ y^{-4} \log \frac{|y|}{l} \Bigr] \,\dl(x)
+ 2 \pi^4 R_4\Bigl[ x^{-4} \log\frac{|x|}{l} \Bigr] \,\dl(y)
\\
&\quad + 2\pi^4 R_4[x^{-4}] \,\dl(x - y)
- \frac{\pi^2}{2}\, x^{-2}  \,\del_y^\bt L_\bt(x - y; y)
\\
&\quad - \frac{\pi^2}{2}\, y^{-2}
\,\del_x^\bt \bigl( L_\bt(x - y; x) + M_\bt(x - y; x) \bigr)
\\
&\quad + \frac{1}{4} \del_x^\bt \biggl(
y^{-2} \del^x_\bt((x - y)^{-2}) K(x,y) \log\frac{|x|}{|y|} 
- y^{-2} (x - y)^{-2} \del^x_\bt \Bigl( K(x,y) \log\frac{|x|}{|y|}
\Bigr) \biggr).
\end{align*}
Note now, using~\eqref{eq:rminus4-two-logs}, that the coefficient
$-\pi^4/3$ of $y^{-2} \log^3(|y|/l)$ agrees with our expectation. To
wit, $-4\pi^2 \x (-\pi^4/3) = 4\pi^6/3$.

\medskip

The scale derivative of the kite now follows readily; note that the 
last line in the above display for the amplitude will not contribute.
We obtain
\begin{align*}
\lddl \bigl( \cometa \bigr)
&= - 4\pi^4 R_4\Bigl[ y^{-4} \log \frac{|y|}{l} \Bigr] \,\dl(x)
- 2\pi^4 R_4[y^{-4}] \,\dl(x) - 2\pi^4 R_4[x^{-4}] \,\dl(y)
\\
&\qquad - 4\pi^6 \,\dl(x) \,\dl(y)
- 2\pi^4 R_4[x^{-4}] \bigl( \dl(x - y) - \dl(y) \bigr)
\\
&\qquad - 2\pi^4 R_4[y^{-4}] \bigl( \dl(x - y) - \dl(x) \bigr)
+ \pi^2 y^{-2} \,\del_x^\bt L_\bt(x - y; x)
\\
&= - 4\pi^4 R_4\Bigl[ y^{-4} \log \frac{|y|}{l} \Bigr] \,\dl(x)
- 4\pi^4 R_4[y^{-4}] \,\dl(x - y) - 4\pi^6 \,\dl(x) \,\dl(y)
\\
&\qquad - \pi^2 y^{-2} \,\del_x^\bt L_\bt(y - x; x)
\\
&= - 2\pi^2 \copadevino(y, y - x) - 2\pi^4 \pez(x) \,\dl(x - y)
- 4\pi^6 \,\dl(x) \,\dl(y).
\end{align*}
(The first cograph on the right hand side has a different labelling of
the vertices from that of Section~\ref{ssc:wine-cup}.) We end up with
$$
\lddl \sG^4_\cometa = 4 \g \,\sG^4_\copadevino - 8 \g^2 \,\sG^2_\pez
+ 24 \g^3 \,\sG^4_\cruz .
$$

\medskip

To derive an explicit expression for $K(x,y)$, note first that
$K(tx, ty) = t^{-2}\,K(x,y)$, so we can assume for now that
$|y| = 1$. Writing $x = r\om$, $u = s\sg$, $y = \eta$ in polar form, 
we can expand the integrand in Gegenbauer polynomials. For instance,
in the region $s < r < 1$ of the $(r,s)$ positive quadrant, we get
$$
(y - u)^{-2} = \sum_{n=0}^\infty s^n C^1_n(\eta\.\sg);
\qquad
(x - u)^{-2} 
= \frac{1}{r^2} \sum_{m=0}^\infty \frac{s^m}{r^m}\, C^1_m(\sg\.\om).
$$
Following~\cite{Todorov14}, we take advantage of the underlying 
conformal symmetry to introduce a complex variable~$z$ determined by 
$$
|z| = r = |x|  \word{and}
(y - x)^2 = 1 - 2r\eta\.\om + r^2 = |1 - z|^2.
$$
%%% More generally, $|z|^2 = x^2/y^2$ and
%%% $|1 - z|^2 = (y - x)^2/y^2$.
The Gegenbauer orthogonality relations and the formula
$C^1_m(\cos\th) = \sin(m + 1)\th / \sin\th$ show that
\begin{align*}
\int_{\bS^3} C^1_n(\eta\.\sg) C^1_m(\sg\.\om) \,d^3\sg
&= \frac{\Om_4\,\dl_{mn}}{n + 1}\, C^1_n(\eta\.\om)
= \frac{2\pi^2\,\dl_{mn}}{n + 1}\, C^1_n\biggl(
\frac{z + \bar z}{2|z|} \biggr)
\\
&= \frac{2\pi^2\,\dl_{mn}}{(n + 1)|z|^n}\,
\frac{z^{n+1} - \bar z^{n+1}}{z - \bar z}\,.
\end{align*}
In the given $(r,s)$-region, each such term is multiplied by
$r^{-n-2} \int_0^r s^{2n+1} \,ds = |z|^n/2(n+1)$, yielding a 
contribution to $K(x,y)$ of
$$
\frac{\pi^2}{z - \bar z} \sum_{n=0}^\infty 
\frac{z^{n+1} - \bar z^{n+1}}{(n + 1)^2}
= \frac{\pi^2}{z - \bar z} \bigl( \bL_2(z) - \bL_2(\bar z) \bigr),
$$ 
where $\bL_2(\zeta) = \sum_{m=1}^\infty \zeta^n/n^2$ is the Euler 
dilogarithm.

Similar expressions are found for the other regions of the positive 
quadrant: see~\cite{Todorov14} for more detail. The full result, 
always assuming that $|y| = 1$, is
$$
K(x,y) = \frac{\pi^2}{z - \bar z} \biggl(
2\bL_2(z) - 2\bL_2(\bar z) + \log|z|^2 \log\frac{1 - z}{1 - \bar z}
\biggr),
$$
which turns out to be a single-valued complex function, the so-called
Bloch--Wigner dilogarithm \cite{Zagier07}. For general~$|y|$, this
yields the expression of~\cite{FreedmanJL92} for $K(x,y)$:
\begin{align*}
\frac{\pi^2}{2i\sqrt{x^2y^2 - (x\.y)^2}} 
& \biggl( \!
2 \bL_2\biggl( \frac{(x\.y)^2 + i\sqrt{x^2y^2 - (x\.y)^2}}{y^2}
\biggr) \!
- 2 \bL_2\biggl( \frac{(x\.y)^2 - i\sqrt{x^2y^2 - (x\.y)^2}}{y^2}
\biggr)
\\
&\quad + \log \frac{x^2}{y^2}\,
\log \frac{y^2 - (x\.y)^2 - i\sqrt{x^2y^2 - (x\.y)^2}}
{y^2 - (x\.y)^2 + i\sqrt{x^2y^2 - (x\.y)^2}}\, \biggr).
\end{align*}

% \S 5.6
\subsection{The shark}
\label{ssc:shark}

Last of its class, we consider the \textit{shark} graph, which has the
structure of a convolution of two renormalized diagrams:
$$
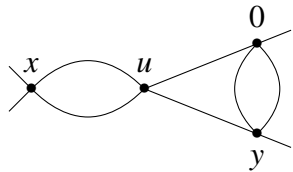
\begin{tikzpicture}[scale=1.5]
\coordinate (A) at (0,0) ; \coordinate (B) at (1,0) ; 
\coordinate (C) at (2,0.4) ; \coordinate (D) at (2,-0.4) ; 
\draw (A) ++(-0.2,0.2) -- (A) -- ++(-0.2,-0.2) ;
\draw (A) parabola[bend pos=0.5] bend +(0,0.25) (B) ;
\draw (A) parabola[bend pos=0.5] bend +(0,-0.25) (B) ;
\draw ($ (B)!1.3!(C) $) -- (B) -- ($ (B)!1.3!(D) $) ;
\begin{scope}[rotate=90]
\draw (C) parabola[bend pos=0.5] bend +(0,0.2) (D) ;
\draw (C) parabola[bend pos=0.5] bend +(0,-0.2) (D) ;
\end{scope}
\draw (A) node[above=2pt] {$x$} ;
\draw (B) node[above=2pt] {$u$} ;
\draw (C) node[above=2pt] {$0$} ;
\draw (D) node[below=2pt] {$y$} ;
\foreach \pt in {A,B,C,D} \draw (\pt) node{$\bull$} ;
\end{tikzpicture}
$$
One-vertex reducible diagrams of this type do not require overall
renormalization: once the convolution is effected, the task is over.
However, the matter is not as simple as implied in
\cite[Sect.~3.3]{FreedmanJL92}: ``Since each factor in the convolution
has a finite Fourier transform, so does the full result''. Of course
not: rather, it is in view of our Proposition~\ref{pr:convolvables}
that the convolution product makes sense.

\medskip

For the winecup subgraph, \eqref{eq:my-cup-overfloweth} gives:
$$
\Rbar_8[u^{-2} (u - y)^{-2} y^{-4}]
= 2\pi^2 R_4\Bigl[ u^{-4} \log\frac{|u|}{l} \Bigr]\,\dl(u - y)
+ \pi^2 R_4[u^{-4}] \,\dl(y)
+ \frac{1}{2}\, u^{-2} \,\del_y^\bt L_\bt(y; u - y).
$$
The fish subgraph is just $R_4[(x - u)^{-4}]$.

Several contributions are immediately computable. First, a
(well-defined) product of distributions,
$$
2\pi^2 \int R_4\Bigl[ u^{-4} \log\frac{|u|}{l} \Bigr]\,
R_4[(x - u)^{-4}] \,\dl(u - y) \,du
= 2\pi^2 R_4\Bigl[ y^{-4} \log\frac{|y|}{l}\Bigr]\, R_4[(x - y)^{-4}].
$$
Next, the straightforward convolution,
$$
\pi^2 \,\dl(y) \int R_4[u^{-4}]\, R_4[(x - u)^{-4}] \,du
= 4\pi^4 R_4\Bigl[ x^{-4} \log\frac{|x|}{l} \Bigr] \,\dl(y)
- \pi^6 \,\dl(x) \,\dl(y).
$$
Thirdly, using $R_4[(x - u)^{-4}]
= -\half\,\Dl\bigl((x - u)^{-2} \log(|x - u|/l) \bigr)
+ \pi^2 \,\dl(x - u)$, we extract
$$
\frac{\pi^2}{2} \int u^{-2} \,\del_y^\bt L_\bt(y; u - y) \,\dl(x - u)\,du
= \frac{\pi^2}{2}\, x^{-2} \,\del_y^\bt L_\bt(y; x - y).
$$
Lastly, we have to add the term
$$
- \frac{1}{4} \,\del_y^\bt \,\Dl_x \int u^{-2} L_\bt(y; u - y)
\,(x - u)^{-2} \log \frac{|x - u|}{l} \,du.
$$
Using \eqref{eq:log-splitting} once more to expand $\log(|x - u|/l)$,
this can be rewritten in terms of $K(x,y)$ as defined in the previous 
case; but this is hardly worthwhile.

\medskip

For the scale derivative, we look first at
$\lddl \int \half u^{-2} \,\del_y^\bt L_\bt(y; u - y) 
R_4[(x - u)^{-4}] \,du$, yielding:
\begin{align*}
& -\pi^2 \int u^{-2} \,\del_y^\bt L_\bt(y; u - y) \,\dl(x - u) \,du
\\
&\qquad + 2\pi^2 \int \bigl( R_4[u^{-4}] \,\dl(u - y)
- R_4[u^{-4}] \,\dl(y) \bigr) R_4[(x - u)^{-4}] \,du
\\
&\quad = -\pi^2 x^{-2} \,\del_y^\bt L_\bt(y; x - y)
+ 2\pi^2 R_4[y^{-4}]\, R_4[(x - y)^{-4}] - 2\pi^2 \sosten(x) \,\dl(y).
\end{align*}
Therefore,
\begin{align*}
\lddl \bigl( \tiburon \bigr)
&= -2\pi^2 R_4[y^{-4}]\, R_4[(x - y)^{-4}]
- 4\pi^4 R_4\Bigl[ x^{-4} \log\frac{|x|}{l} \Bigr] \,\dl(x - y)
- 4\pi^4 R_4[x^{-4}] \,\dl(y)
\\
&\qquad - \pi^2 x^{-2} \,\del_y^\bt L_\bt(y; x - y)
+ 2\pi^2 R_4[y^{-4}]\, R_4[(x - y)^{-4}] - 2\pi^2 \sosten(x) \,\dl(y)
\\
&= - 2\pi^2 \copadevino(x,y) - 2\pi^2 \sosten(x) \,\dl(y)
- 2\pi^4 \pez(x) \,\dl(y).
\end{align*}
This translates into:
$$
\lddl \sG^4_\tiburon = \g\,\sG^4_\copadevino + 4\g\,\sG^4_\sosten
- 2\g^2 \,\sG^4_\pez.
$$

% \S 5.7
\subsection{The tetrahedron diagram}
\label{ssc:four-eyed}

The tetrahedron graph is \textit{primitive} and already
understood~\cite{Carme}. We report the results. On labelling the
vertices thus:
$$
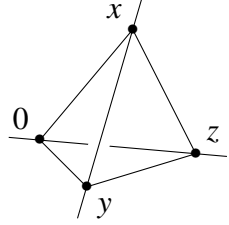
\begin{tikzpicture}[scale=0.8]
\coordinate (A) at (0,0) ; \coordinate (B) at (-45:1.1cm) ;
\coordinate (C) at (-5:2.6cm) ; \coordinate (D) at (50:2.4cm) ;
\draw ($ (A)!-0.2!(C) $) -- ($ (A)!1.2!(C) $) ;
\draw[line width=8pt, white] (B) -- (D) ;
\draw ($ (B)!-0.2!(D) $) -- ($ (B)!1.2!(D) $) ;
\draw (A) node[above left] {$0$} -- (B) node[below right] {$y$}
   -- (C) node[above right] {$z$} -- (D) node[above left] {$x$}
   -- cycle ;
\foreach \pt in {A,B,C,D}  \draw (\pt) node {$\bull$} ;
\end{tikzpicture}
$$
we arrive at
\begin{align*}
\tetrahedron &= (E + 12) \Bigl[
x^{-2} y^{-2} z^{-2} (x-y)^{-2} (y-z)^{-2} (x-z)^{-2}
\log\frac{|(x,y,z)|}{l} \Bigr];
\\
&\qquad \word{with} E + 12
= \del^x_\al x^\al + \del^y_\bt y^\bt + \del^z_\rho z^\rho;
\\
l \frac{\del}{\del l} \tetrahedron &= -12\pi^6\zeta(3) \cross,
\word{leading to} \lddl \sG^4_\tetraedro
= 12\g^3\zeta(3)\,\sG^4_\cruz.
\end{align*}
Notice that the scale derivative for this graph, coincident with the
residue for this case, is numerically large.

% \S 5.8
\subsection{The roll}
\label{ssc:roll}

The unrenormalized amplitude for the roll diagram, with vertices
labelled as follows:
$$
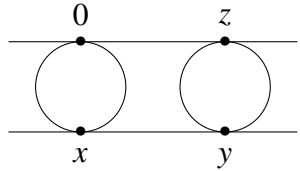
\begin{tikzpicture}[scale=1.2]
\coordinate (A) at (0,1) ; \coordinate (B) at (0,0) ; 
\coordinate (C) at (1.6,0) ; \coordinate (D) at (1.6,1) ; 
\draw ($ (A)!-0.5!(D) $) -- ($ (A)!1.5!(D) $) ; 
\draw ($ (B)!-0.5!(C) $) -- ($ (B)!1.5!(C) $) ; 
\draw ($ (A)!0.5!(B) $) circle(0.5cm) ;
\draw ($ (C)!0.5!(D) $) circle(0.5cm) ;
\draw (A) node[above=2pt] {$0$} ;
\draw (B) node[below=2pt] {$x$} ;
\draw (C) node[below=2pt] {$y$} ;
\draw (D) node[above=2pt] {$z$} ;
\foreach \pt in {A,B,C,D} \draw (\pt) node{$\bull$} ;
\end{tikzpicture}
$$
is of the form
$$
f(x,y,z) = (x - y)^{-2} x^{-4} z^{-2} (y - z)^{-4}.
$$

Before plunging into calculation, this very interesting graph without
internal vertices prompts a couple of comments. It exemplifies well
the ``causal factorization property''. Consider the relevant partition
$\{0,x\},\{z,y\}$ of its set of vertices. Partial renormalization
adapted to it will yield a valid distribution outside the diagonals
$0 = z$, $x = y$. Formula~\eqref{eq:in-annum} here means:
$$
\duo< R[\Ga], \vf>
= \duo< R[\ga_1\uplus\ga_2], \bigl( \Ga/(\ga_1\uplus\ga_2) \bigr)\vf>
= \duo< R[\ga_1] R[\ga_2], \bigl( \Ga/(\ga_1\uplus\ga_2) \bigr)\vf>,
$$
for $\vf$ vanishing on those diagonals; where the rule for
disconnected graphs, also found in~\cite[Sect.~11.2]{KuznetsovTV96}:
$$
R[\ga_1 \uplus \ga_2] = R[\ga_1]\,R[\ga_2],
$$
holds; and
$$
\Ga/(\ga_1\uplus\ga_2) = \biggl( \markedfish \biggr).
$$

Second, one may anticipate the coefficients of logarithmic degrees $3$
and~$2$ in the final result, by using again \cite{ChryssomalakosQRV02}
and~\cite{Kreimer01}. Note that the roll is a rooted-tree graph with
the fish as unique ``decoration''; subtracting from it the
juxtaposition of the winecup and the fish graph, and adding one sixth
of the product of three fishes, one obtains a primitive graph in the
bialgebra of rooted trees. For this combination the coefficient of
$\log^3 l$ must vanish. We obtain for such a graph
$2\pi^4 \x 2\pi^2 - 8\pi^6/6 = 8\pi^6/3$.

Partial renormalization leads at once to
\begin{align}
\Rbar_{12}[f(x,y,z)]
&= (x - y)^{-2} R_4[x^{-4}]\, z^{-2} R_4[(y - z)^{-4}]
\nonumber \\
&= \frac{1}{4} \biggl( 
(x - y)^{-2} \,\Dl\Bigl( x^{-2} \log\frac{|x|}{l} \Bigr)
- 2\pi^2 y^{-2}\,\dl(x) \biggr) 
\nonumber \\
&\qquad \x \biggl(
z^{-2} \,\Dl_z \Bigl( (y - z)^{-2} \log\frac{|y - z|}{l} \Bigr)
- 2\pi^2 y^{-2} \,\dl(y - z) \biggr).
\label{eq:four-uglies} % (5.8)
\end{align}

The expression~\eqref{eq:four-uglies} is a sum of four terms. Three of
these yield, respectively:
\begin{align*}
& \pi^4 y^{-4} \,\dl(x)\,\dl(y - z) 
\renormto \pi^4 R_4[y^{-4}] \,\dl(x)\,\dl(y - z);
\\
& - \frac{\pi^2}{2}\, y^{-2} (x - y)^{-2} 
\,\Dl\Bigl( x^{-2} \log\frac{|x|}{l} \Bigr) \,\dl(y - z)
\\
&\qquad \renormto 2\pi^4 R_4\bigl[ x^{-4} \log\frac{|x|}{l} \Bigr]
\,\dl(x - y) \,\dl(x - z) 
+ \frac{\pi^2}{2}\, y^{-2} \,\del_x^\al L_\al(x; y - x) \,\dl(y - z);
\\
& - \frac{\pi^2}{2}\, y^{-2} z^{-2} 
\,\Dl_z\Bigl( (y - z)^{-2} \log\frac{|y - z|}{l} \Bigr) \,\dl(x)
\\
&\qquad \renormto 2\pi^4 R_4\Bigl[ y^{-4} \log\frac{|y|}{l} \Bigr]
\,\dl(x)\,\dl(z)
+ \frac{\pi^2}{2}\, y^{-2} \,\del_z^\bt L_\bt(y - z; z) \,\dl(x).
\end{align*}

The remaining contribution from~\eqref{eq:four-uglies}, after 
integrating by parts twice with~\eqref{eq:green-backs}, yields
\begin{align*}
&4\pi^4 x^{-4} \log^2 \frac{|x|}{l} \,\dl(x - y)\,\dl(z)
+ \pi^2 y^{-2} \log\frac{|y|}{l} \,\del_x^\al L_\al(x; y - x) \,\dl(z)
\\
&\quad 
+ \pi^2 y^{-2} \log\frac{|y|}{l} \,\del_z^\bt L_\bt(y-z; z) \,\dl(x-y)
+ \frac{1}{4} \,\del_x^\al L_\al(x; y-x) \,\del_z^\bt L_\bt(y-z; z).
\end{align*}
Only the first of these terms requires renormalization:
$$
4\pi^4 x^{-4} \log^2 \frac{|x|}{l} \,\dl(x - y)\,\dl(z)
\renormto 4\pi^4 R_4\Bigl[ x^{-4} \log^2 \frac{|x|}{l} \Bigr]
\,\dl(x - y)\,\dl(z).
$$

Summing up, we arrive at
\begin{align}
\roll
&= 4\pi^4 R_4\Bigl[ y^{-4} \log^2 \frac{|y|}{l} \Bigr] 
\,\dl(x - y)\,\dl(z)
\notag \\
&\quad + 2\pi^4 R_4\Bigl[ y^{-4} \log \frac{|y|}{l} \Bigr]
\bigl( \dl(x)\,\dl(z) + \dl(x - y)\,\dl(x - z) \bigr)
\notag \\
&\quad + \pi^4 R_4[y^{-4}] \,\dl(x)\,\dl(y - z)
+ (\text{total derivative terms}).
\label{eq:ampli-rollo} % (5.9)
\end{align}

A peek at~\eqref{eq:rminus4-two-logs} confirms that the coefficient for
$\log^3 l$ in this expression is
$$
- \frac{2\pi^4}{3} \x (-4\pi^2) = \frac{8\pi^6}{3} \,,
$$
as predicted by Kreimer's argument.

\medskip

For the scale derivative of the roll amplitude, the first three terms 
in~\eqref{eq:ampli-rollo} contribute
\begin{align*}
& - 8\pi^4 R_4\Bigl[ y^{-4} \log\frac{|y|}{l} \Bigr] \dl(x-y)\,\dl(z)
- 2\pi^4 R_4[y^{-4}] \bigl( \dl(x)\,\dl(z) + \dl(x-y)\,\dl(y-z) \bigr)
\\
&\quad - 2\pi^6 \,\dl(x)\,\dl(y)\,\dl(z),
\end{align*}
while the other (total derivative) terms contribute
\begin{align*}
& - 4\pi^4 R_4[y^{-4}] \,\dl(x)\,\dl(y - z)
+ 2\pi^4 R_4[y^{-4}] \bigl( \dl(x)\,\dl(z) +\dl(x-y)\,\dl(x-z) \bigr)
\\
&\quad + 8\pi^4 R_4\Bigr[ y^{-4} \log\frac{|y|}{l} \Bigr]\,\dl(x-y)\,\dl(z)
- 4\pi^4 R_4\Bigl[ y^{-4} \log\frac{|y|}{l} \Bigr] 
\bigl( \dl(x)\,\dl(z) + \dl(x-y)\,\dl(y-z) \bigr)
\\
&\quad - \pi^2 y^{-2} \,\del_x^\al  L_\al(x; y - x) \,\dl(y - z)
- \pi^2 y^{-2} \,\del_z^\bt  L_\bt(y - z; z) \,\dl(x).
\end{align*}
Putting them together, we arrive at
\begin{align*}
\lddl \roll
&= - 4\pi^4 R_4\Bigr[ y^{-4} \log\frac{|y|}{l} \Bigr] 
\bigl( \dl(x)\,\dl(z) + \dl(x-y)\,\dl(y-z) \bigr)
\\
&\quad - 4\pi^4 R_4[y^{-4}] \,\dl(x)\,\dl(y - z)
- 2\pi^6\,\dl(x)\,\dl(y)\,\dl(z)
\\
&\quad - \pi^2 y^{-2} \,\del_x^\al  L_\al(x; y - x) \,\dl(y - z)
- \pi^2 y^{-2} \,\del_z^\bt  L_\bt(y - z; z) \,\dl(x)
\\
&= - 2\pi^2 \copadevino(y,x) \,\dl(y - z)
- 2\pi^2 \copadevino(y, y - z) \,\dl(x)  - 2\pi^6 \cross.
\end{align*}
Taking into account the weight factors, we then conclude that
$$
\lddl \sG^4_\canelon = 2\g\,\sG^4_\copadevino + 3\g^3\,\sG^4_\cruz.
$$

% \S 6
\section{The renormalization group $\ga$- and $\bt$-functions}
\label{sec:feel-the-beat}

Enter the Callan--Symanzik differential equations with zero mass:
\begin{align}
\biggl[ \pd{}{\log l} - \bt(\g) \pd{}{\g} + 2\ga(\g) \biggr]
\sG^2_p(x_1,x_2) &= 0
\label{eq:CS-eqns-one} % (6.1)
\\
\word{and}
\biggl[ \pd{}{\log l} - \bt(\g) \pd{}{\g} + 4\ga(\g) \biggr]
\,\sG^4(x_1,x_2,x_3,x_4) &= 0.
\label{eq:CS-eqns-two} % (6.2)
\end{align}
Here $\sG^2$ starts at order~$\g^0$ and $\sG^4$ at order~$\g$; the
subindex in $\sG^2_p$ recalls that the equations refer to proper
parts.  (The detailed calculations for the two-point function $\sG^2$
and its scale derivative are outlined in Appendix~\ref{app:colon}.)
The scale derivative in both cases starts at order~$\g^2$.  Therefore
we may assume that:
\begin{equation}
\bt(\g) = \g^2\bt_1 + \g^3\bt_2 + \g^4\bt_3 +\cdots ; \qquad
\ga(\g) = \g^2\ga_2 + \g^3\ga_3 + \g^4\ga_4 +\cdots ;
\label{eq:coefs-app} % (6.3)
\end{equation}
where $\ga$ does not contribute to~\eqref{eq:CS-eqns-two} at the first
significant order, and similarly for~$\bt$ and~\eqref{eq:CS-eqns-one};
and we try to compute then the $\ga_i$ and~$\bt_i$. Following
\cite{FreedmanJL92}, our labelling of the expansion coefficients
differs for the~$\bt$ and~$\ga$ functions.

Order by order, we find:
\begin{align}
\g^2 : \ & \lddl \sG^4_\pez
= \bt_1 \,\g^2 \frac{\del}{\del\g}\, \sG^4_\cruz ; \qquad
\lddl \sG^2_\ocaso = -2\ga_2 \,\g^2\, \sG^2_\punto ;
\notag \\
\g^3 : \ & \lddl \biggl[ \sG^4_\sosten + \sG^4_\copadevino \biggr]
= (\bt_2 - 4\ga_2) \,\g^3 \frac{\del}{\del\g}\, \sG^4_\cruz
+ \bt_1 \,\g^2 \frac{\del}{\del\g}\, \sG^4_\pez ;
\label{eq:gbar-cubed-fourpt} % (6.4)
\\
& \lddl \sG^2_\gafas = \bt_1\,\g^2 \frac{\del}{\del\g}\, \sG^2_\ocaso
- 2\ga_3 \,\g^2\, \sG^2_\punto ;
\label{eq:gbar-cubed-twopt} % (6.5)
\\
\g^4 : \ &\lddl \biggl[ \sG^4_\trenza + \sG^4_\orzuelo
+ \sG^4_\ojodegato + \sG^4_\casquete + \sG^4_\tiburon + \sG^4_\cometa 
+ \sG^4_\tetraedro + \sG^4_\canelon \biggr]
\notag \\[-2\jot]
&= (\bt_3 - 4\ga_3) \g^4 \frac{\del}{\del\g}\, \sG^4_\cruz
+ (\bt_2 - 2\ga_2) \g^3 \frac{\del}{\del\g}\, \sG^4_\pez
+ \bt_1 \g^2 \frac{\del}{\del\g}\, \sG^4_\sosten
+ \bt_1 \g^2 \frac{\del}{\del\g}\, \sG^4_\copadevino ;
\label{eq:gbar-fourth-terms} % (6.6)
\\
&\lddl \biggl[ \sG^2_\cebollas + \sG^2_\saturno + \sG^2_\cuca
+ \sG^2_\caracol\biggr]
%% no $\sG^2_\pesas$ here!
\notag \\
&= \bt_1\,\g^2 \frac{\del}{\del\g}\, \sG^2_\gafas
+ \bt_2\,\g^3 \frac{\del}{\del\g}\, \sG^2_\ocaso
- 2\ga_2\,\g^2 \,\sG^2_\ocaso - 2\ga_4\,\g^4 \,\sG^2_\punto \,.
\label{eq:gbar-fourth-terms-2-point} % (6.7)
\end{align}

The first equality above, in view of~\eqref{eq:fish-scale}, yields
$\bt_1 = 3$. This is the standard result. Off the same line, with the
help of~\eqref{eq:evening-star}, we read $\ga_2 = 1/12$. This is also
the standard result.

Then, on consulting \eqref{eq:bikini-scale}
and~\eqref{eq:winecup-scale}, equality~\eqref{eq:gbar-cubed-fourpt} is
seen to yield:
$$
6\g\,\sG^4_\pez - 6\g^2\,\sG^4_\cruz
= 6\g\,\sG^4_\pez + \bigl(\bt_2 - \tfrac{1}{3}\bigr)\g^2\,\sG^4_\cruz;
$$
that is, $\bt_2 = -17/3$. This is the standard result. Note the
automatic cancellation of the terms in~$\sG^4_\pez$.

We need the value of~$\ga_3$ in order to compute~$\bt_3$. Now, 
from~\eqref{eq:gbar-cubed-twopt} we observe that
$$
6\g\,\sG^2_\ocaso = 6\g\,\sG^2_\ocaso + 2\ga_3\,\g^2 \,\Dl\dl;
$$
so that $\ga_3 = 0$ obtains, at variance with both \cite{FreedmanJL92}
and~\cite{KleinertSF01} (differing between them); but in agreement
with~\cite{Schnetz97}.

\medskip

We turn to the computation of $\bt_3$, noting beforehand that
knowledge of $\ga_4$ is not necessary for it. First we check the
automatic cancellation of the terms in $\sG^4_\pez$
in~\eqref{eq:gbar-fourth-terms}:
$$
\Bigl( \frac{1}{3} - 2 - 2 - 8 \Bigr) \g^2\,\sG^4_\pez
= 2\Bigl( -\frac{17}{3} - \frac{1}{6} \Bigr) \g^2\,\sG^4_\pez ;
$$
as well as in $\sG^4_\sosten$ and $\sG^4_\copadevino$:
\begin{align*}
(3 + 2 + 4) \g\,\sG^4_\sosten &= 9 \g\,\sG^4_\sosten ;
\\[\jot]
(2 + 4 + 1 + 2) \g\,\sG^4_\copadevino &= 9 \g\,\sG^4_\copadevino ;
\end{align*}
which of course vouches for the soundness of our method. We should
also notice that, since the chain graphs do not yield $\sG^4_\cruz$
terms, they contribute nothing to the renormalization group functions.
The same is true of the shark graph.

Thus we read off $\bt_3$ from the terms in $\sG^4_\cruz$ on the
left hand-side of~\eqref{eq:gbar-fourth-terms}, with the result:
$$
\bt_3 = \frac{109}{4} + 12\,\zeta(3),
$$
numerically intermediate between the results in~\cite{FreedmanJL92}
and~\cite{KleinertSF01}. The discrepancy with the former is due to
different results for the $\sG^4_\cruz$ terms in the stye, cat's eye,
duncecap and roll graphs; for any others our scale derivatives
reproduce the results in the seminal paper on differential
renormalization. This number has, at any rate, no fundamental
significance; in fact $\bt_3$ and all subsequent coefficients of the
$\bt$-function can be made to vanish in an appropriate renormalization
scheme.

Before computing $\ga_4$ in our scheme, with an eye on the formulas
\eqref{eq:punto-final}, we check the cancellation of terms in
$\sG^2_\gafas$ and $\sG^2_\ocaso$ (coming from proper graphs only) in
\eqref{eq:gbar-fourth-terms-2-point}:
$$
(3 + 2 + 4) = 3\bt_1 = 9; \quad 
\Bigl( \frac{1}{2} - 6 - 6 \Bigr) 
= 2(\bt_2 - \ga_2) = - \frac{34}{3} - \frac{1}{6} \,. 
$$

Finally, for $\ga_4$ only the contribution from the
$\smash[b]{\saturn}$ graph survives in
\eqref{eq:gbar-fourth-terms-2-point}, and we obtain $\ga_4 = -5/96$,
again in agreement with~\cite{Schnetz97}.

% \S 7
\section{Conclusion}
\label{sec:sic-transit}

We reckon to have shown that, with a small amount of ingenuity, plus
eventual recourse to Gegenbauer polynomial techniques and
polylogarithms, renormalization of proper ``divergent'' graphs in
coordinate space by Epstein--Glaser methods is quite feasible.

As the perturbation order grows, this often demands integration over
the internal vertices. We have proved here the basic proposition
underpinning this latter technique, and provided quite a few examples.
We trust that, all along, the logical advantage of
distribution-theoretic methods for the recursive treatment of
divergences shines through.

The Popineau--Stora theorem was mentioned in
Section~\ref{ssc:stora-story}. A refinement of this
theorem~\cite{Stora08,RaymondCERNFeb13} provides an expression for
$\Gbar^{\,ll'}(\g)$ of the form
$$
\Gbar^{\,ll'}(\g) = \g + H\bigl( \Gbar^{\,ll'}(\g) \bigr),
$$
where the series $H$ is made up of successive contributions at each
perturbation step, always supported exclusively on the corresponding
main diagonal. We expect to take up in a future paper this interesting
combinatorial aspect of the procedure; at any rate, our methods
guarantee that no combinatorial \textit{difficulties} worth mentioning
appear.

The attentive reader will not have failed to notice the connection
between the apparently fortuite cancellations signaled in the previous
section and those in Section~\ref{ssc:stora-story}. Those
cancellations look to be just an infinitesimal aspect of the
Popineau--Stora theorem. Here we have done no more than to illustrate
the workings of that theorem and the Callan--Symanzik equations in the
Epstein--Glaser paradigm. A befitting derivation of those equations
within the distribution-theoretic approach is also left for the
future.

% \S 7.1
\subsection{The roads not taken}
\label{sec:gloria-mundi}

We have taken some pains to point out the shortcomings of differential
renormalization; however, in keeping with its spirit, our approach to
cograph parts is unabashedly ``low-tech'': some of the standard tools
of renormalization in $x$-space are not used.
\begin{itemize}
\item
Even for dealing with overlapping divergences, we have had no recourse
to partitions of unity.
\item
We do not use here \textit{meromorphic continuation}: we wanted to
illustrate the fact that real-variable methods \textit{\`a la} Epstein
and Glaser are enough to deal with the problems at hand. This is of
course a net loss in practice, since the analytic continuation tools
\cite{Hormander90,GelfandS64,HorvathCol74} borrowed in
\cite{NikolovST13,Keller10,DuetschFKR13} are quite powerful, although
not always available. The wisest course is to employ both real- and
complex-variable methods.
\item
The calculus of \textit{wave front sets} was not required.
\item
Steinmann's \textit{scaling degree} for distributions was never
invoked. There is no point in using it for massless
diagrams~\cite{RaymondCERNFeb13}, for which the log-homogeneous
classification is finer: all distributions of bidegree $(a,m)$ have
the same scaling degree irrespectively of the value of~$m$. Recent
work~\cite{Duetsch14} adapts the latter classification to the case of
massive particles, casting doubt on the future usefulness of the
scaling degree in quantum field theory. Even for general
distributions, Meyer's concept of weakly homogeneous distributions,
exploited in~\cite{Dang13}, appears more seductive.
\end{itemize}

\appendix

% \S A
\section{Formulas for extensions of distributions in $x$-space}
\label{app:exten-formulas}

Recall the definition of $R_d \bigl[r^{-d} \log^m(r/l) \bigr]$: 
from~\eqref{eq:genl-exten} we immediately obtain
\begin{equation}
R_d \Bigl[r^{-d} \log^m \frac{r}{l}\Bigr]
= \frac{1}{m+1}\, \del_\al \Bigl(
x^\al r^{-d} \log^{m+1}\frac{r}{l} \Bigr) 
= \sum_{k=0}^{m+1} c_{m+1,k}
\,\Dl\Bigl( r^{-d+2} \log^k \frac{r}{l} \Bigr)
\label{eq:loggy-exten} % (A.1)
\end{equation}
for suitable constants $c_{m+1,k}$. These are computed as follows.

\begin{lema} % 3
\label{lm:many-log-ren}
For any $m = 0,1,2,\dots$, 
\begin{align}
R_d \Bigl[ r^{-d} \log^m \frac{r}{l} \Bigr]
&= - \sum_{k=1}^{m+1} \frac{m!}{k!}\, (d - 2)^{-m+k-2}
\,\Dl \Bigl( r^{-d+2} \log^k \frac{r}{l} \Bigr)
+ \frac{m!}{(d - 2)^{m+1}}\,\Om_d\,\dl(r);
\nonumber \\
\intertext{and in particular,}
R_4 \Bigl[r^{-4} \log^m \frac{r}{l}\Bigr]
&= - \sum_{k=1}^{m+1} \frac{m!}{k!}\, \frac{1}{2^{m-k+2}}
\,\Dl \Bigl( r^{-2} \log^k \frac{r}{l} \Bigr)
+ \frac{m!}{2^m}\,\pi^2\,\dl(r).
\label{eq:many-log-ren} % (A.2)
\end{align}
\end{lema}

\begin{proof}
Taking derivatives, we get
$$
\del_\al \Bigl( r^{-d+2} \log^k \frac{r}{l} \Bigr)
= x^\al r^{-d} \Bigl( (2 - d) \log^k \frac{r}{l}
+ k \log^{k-1} \frac{r}{l} \Bigr).
$$
The coefficients $c_{m+1,k}$ are determined by the defining relation
$$
x^\al r^{-d} \log^{m+1} \frac{r}{l}
= (m + 1) \sum_{k=0}^{m+1} c_{m+1,k} 
\,\del_\al \Bigl( r^{-d+2} \log^k \frac{r}{l} \Bigr);
$$
so we get the recurrence, for $k \leq m$:
$$
(k + 1) c_{m+1,k+1} - (d - 2) c_{m+1,k} = 0.
$$
Clearly $c_{m+1,m+1} = -1/(d - 2)(m + 1)$. Thus
$c_{m+1,m} = -1/(d - 2)^2$. The
remaining $c_{m+1,k}$ terms follow at once. The last summand is
$c_{m+1,0}\,\Dl(r^{-d+2}) = m!(d - 2)^{-m-1} \Om_d\,\dl(r)$.
\end{proof}

We expand out the cases $m = 0,1,2$ of $d = 4$ for ready reference:
\begin{align}
R_4[r^{-4}] 
&= - \frac{1}{2}\,\Dl\biggl( r^{-2} \log\frac{r}{l} \biggr)
+ \pi^2\,\dl(r).
\label{eq:rminus4-no-log} % (A.3)
\\
R_4 \Bigl[r^{-4} \log \frac{r}{l}\Bigr]
&= - \frac{1}{4}\,\Dl \Bigl( r^{-2} \log^2 \frac{r}{l} \Bigr)
- \frac{1}{4} \Dl \Bigl( r^{-2} \log \frac{r}{l} \Bigr)
+ \frac{\pi^2}{2} \,\dl(r);
\label{eq:rminus4-one-log} % (A.4)
\\
R_4 \Bigl[r^{-4} \log^2 \frac{r}{l}\Bigr]
&= - \frac{1}{6}\,\Dl \Bigl( r^{-2} \log^3 \frac{r}{l} \Bigr)
- \frac{1}{4}\,\Dl \Bigl( r^{-2} \log^2 \frac{r}{l} \Bigr)
- \frac{1}{4} \Dl \Bigl( r^{-2} \log \frac{r}{l} \Bigr)
+ \frac{\pi^2}{2} \,\dl(r).
\label{eq:rminus4-two-logs} % (A.5)
\end{align}

\medskip

Explicit expressions for log-homogeneous distributions of higher 
bidegrees could in principle be computed from \eqref{eq:yo-yo} 
and~\eqref{eq:pizza-crust}. We have already met $R_4[r^{-6}]$ 
in~\eqref{eq:ren-hex}. We also need:
\begin{align}
R_4 \Bigl[r^{-6} \log \frac{r}{l}\Bigr]
&= \frac{1}{8}\, \Dl R_4\Bigl[r^{-4} \log\frac{r}{l} \Bigr]
+ \frac{3}{32}\, \Dl R_4[r^{-4}] + \frac{7\pi^2}{64}\,\Dl\dl(r),
\label{eq:ren-hex-log} % (A.6)
\\
R_4 \Bigl[r^{-6} \log^2 \frac{r}{l}\Bigr]
&= \frac{1}{8}\, \Dl R_4\Bigl[r^{-4} \log^2\frac{r}{l} \Bigr]
+ \frac{3}{16}\, \Dl R_4\Bigl[r^{-4} \log\frac{r}{l} \Bigr]
+ \frac{7}{64}\, \Dl R_4[r^{-4}] + \frac{15\pi^2}{128}\, \Dl\dl(r);
\notag \\
\lddl  R_4 \Bigl[r^{-6} \log^m \frac{r}{l}\Bigr]
&= -m\,R_4 \Bigl[r^{-6} \log^{m-1} \frac{r}{l}\Bigr]
\word{for} m \geq 1.
\label{eq:hex-log-scaling} % (A.7)
\end{align}
The last equation follows from~\eqref{eq:huevo-de-Colon} and the
algebra property. It is worth mentioning that an expression equivalent
to~\eqref{eq:ren-hex-log} was obtained by Jones on extending
$r^{-6} \log r$ by meromorphic continuation: see Eq.~(34) on page~255
of~\cite{Jones82}. The one on the following line appears to be new.

The reader might wish to run the algebra checks:
$$
r^2 R_4\Bigl[ r^{-6} \log^m \frac{r}{l} \Bigr]
= R_4\Bigl[ r^{-4} \log^m \frac{r}{l} \Bigr].
$$

% \S B
\section{On the two-point function}
\label{app:colon}

The free Green function $G_\free(x_1,x_2)$ is the same as the 
Euclidean ``propagator'' but for the sign:
$$
\frac{1}{4\pi^2(x_1 - x_2)^2} \,,
$$
Perturbatively, the corrected or dressed propagator is:
\begin{align*}
G(x_1 - x_2) &= G_\free(x_1 - x_2)
+ \iint G_\free(x_1 - x_a)\, \Sg(x_a - x_b)\, G_\free(x_b - x_2) 
\\
&\ \ + \iiiint G_\free(x_1 - x_2)\, \Sg(x_1 - x_b)\,
G_\free(x_b - x_c)\, \Sg(x_c - x_d)\, G_\free(x_d - x_2) + \cdots
\end{align*}
Here $\Sg$ is the proper (one-particle irreducible) self-energy. The
solution of the above convolution equation is given by the 
convolution inverse
$$
(G_\free - \Sg)^{*-1}(x) =: \sG^2(x);
$$
this is what we call the \textit{two-point function}. Thus in
first approximation $\sG^2(x)$ is given by
$$
\Bigl( \frac{1}{4\pi^2x^2} \Bigr)^{*-1} = -\Dl\dl(x).
$$

We gather:
\begin{itemize}

\item
At order $\g^0$ the two-point function is simply given by
$-\Dl\dl(x) \equiv \sG^2_\punto$.

\item
The sunset graph comprises the first correction, at order~$\g^2$,
already computed in~\eqref{eq:ren-hex}.

\item
The unique third-order graph for the two-point function \goggles\ 
is partially renormalized as
$$
\Rbar_8 \bigl[ x^{-2} (x^{-4})^{*2} \bigr] = \pi^2 x^{-2} \biggl(
4R_4\Bigl[ x^{-4} \log\frac{|x|}{l} \Bigr] - \pi^2\,\dl(x) \biggr).
$$
The homogeneous distribution $x^{-2}\,\dl(x)$ is renormalized by the
standard formula~\eqref{eq:pizza-crust},\break
yielding $\frac{1}{8} \,\Dl\dl(x)$; note that the algebra rule is
fulfilled. Thus from~\eqref{eq:wagnerian-opera},
\begin{align}
\goggles &= R_4\bigl[ x^{-2} (x^{-4})^{*2} \bigr]
= 4\pi^2 R_4\Bigl[ x^{-6} \log\frac{|x|}{l} \Bigr]
- \frac{\pi^4}{8} \,\Dl\dl(x),
\notag \\
& \word{with scale derivative} - 4\pi^2 \sunset \,.
\label{eq:clear-sight} % (B.1)
\end{align}
Here and for subsequent graphs, the reader can consult
\eqref{eq:ren-hex-log} for totally explicit forms in terms of
Laplacians.

\item
There are four proper fourth-order graphs for the two-point function.
A very simple one, of the chain type, is \ \onions\ \ whose
renormalized form we write from~\eqref{eq:golden-braid} at once,
\begin{align}
R_4\bigl[ x^{-2} (x^{-4})^{*3} \bigr]
&= 12\pi^4 R_4\Bigl[ x^{-6} \log^2\frac{|x|}{l} \Bigr]
- 3\pi^4 R_4[x^{-6}] + \frac{2\zeta(3) - 1}{4}\,\pi^6\,\Dl\dl(x),
\notag \\
&\word{with scale derivative} -6\pi^2 \goggles \,.
\label{eq:ripe-onions} % (B.2)
\end{align}

\item
Next we consider the \textit{saturn} graph. Analogously, the ground
work was already done for the stye diagram, and we may write at once:
\begin{align}
R_4 \biggl[ \saturn \biggr]
&= -\pi^4 R_4\Bigl[ x^{-6} \log\frac{|x|}{l} \Bigr] 
- \frac{5\pi^4}{4}\, R_4[x^{-6}] \,,
\notag \\
&\word{with scale derivative}
\pi^4 \sunset + \frac{5\pi^6}{16}\, \Dl\dl \,.
\label{eq:baby-eater} % (B.3)
\end{align}

\item 
The \textit{roach} graph \ \roach \ has the bare form
$x^{-2}\,\catseye$, renormalized at once~by
\begin{gather}
4\pi^4 R_4\Bigl[ x^{-6} \log^2\frac{|x|}{l} \Bigr] 
+ 4\pi^4 R_4\Bigl[ x^{-6} \log\frac{|x|}{l} \Bigr]
- 2\pi^4 R_4[x^{-6}] - \frac{\pi^6}{8}\, \Dl\dl(x),
\notag \\
\word{with} \lddl \roach = - 2\pi^2 \goggles - 4\pi^4 \sunset
+ \frac{\pi^6}{4} \,\Dl\dl\,.
\label{eq:gimlet} % (B.4)
\end{gather}
The simplicity of our treatment for this graph stands in stark
contrast with the ``combinatorial monstrosity of the forest
formula''~\cite{Tkachov99}, patent in
\cite[Example~3.2]{DuetschFKR13}.

\item 
Finally, there is the \textit{snail} graph \ \snail \ whose amplitude
is of the bare form
$$
f(x) = \iint \frac{du\,dv}{u^4(v - u)^2(x - u)^2v^2(x - v)^4} \,.
$$

As indicated in \cite{Schnetz97}, this can be related to previously
computed amplitudes:
\begin{align*}
& R_4 \biggl[ \snail \biggr]
+ R_4 \biggl[ \roach \biggr] - R_4 \biggl[ \onions \biggr]
\\
&\qquad = 2\pi^2 R_4 \biggl[ \goggles \biggr] 
+ \int du\, \del^u_\mu Q^\mu(u;v,x),
\end{align*}
where
$$
Q^\mu(u;v,x) 
:= \bigl( (u^\mu - x^\mu) u^{-2}(u - v)^{-2} (u - x)^{-2}
 - u^\mu u^{-4}(u - v)^{-2} \bigr) v^{-2} (v - x)^{-4} x^{-2}.
$$
This yields:
\begin{align}
& R_4 \biggl[ \snail \biggr]
\notag \\
&\quad = 8\pi^4 R_4\Bigl[ x^{-6} \log^2\frac{|x|}{l} \Bigr] 
+ 4\pi^4 R_4\Bigl[ x^{-6} \log\frac{|x|}{l} \Bigr] 
- \pi^4 R_4[x^{-6}] + \frac{4\zeta(3) - 3}8\pi^6 \Dl\dl(x)\,,
\notag \\
&\quad\word{with scale derivative}
- 4\pi^2 \goggles - 4\pi^4 \sunset - \frac{\pi^6}{4}\, \Dl\dl\,.
\label{eq:escargot} % (B.5)
\end{align}

\item
At order four there is an improper contribution to the propagator, by
the double sunset:
\begin{gather*}
\sF R_4\biggl[ \barbells \biggr]
= -p^2 \biggl( \frac{\pi^2}{4} \log\frac{|p|}{\La}
- \frac{5\pi^2}{16} \biggr)^2
= - p^2\,\frac{\pi^4}{16} \biggl( \log^2\frac{|p|}{\La}
- \frac{5}{2} \log\frac{|p|}{\La} + \frac{25}{16} \biggr),
\\[\jot]
\word{implying} \barbells 
= \frac{\pi^2}{2}\, R_4\Bigl[ x^{-6} \log\frac{|x|}{l} \Bigr]
- \frac{9\pi^4}{256}\, \Dl\dl(x).
\end{gather*}
This is just $R_4[r^{-6}] * r^{-2} * R_4[r^{-6}]$, yielding:
\begin{align}
R_4\biggl[ \barbells \biggr] 
= -2\pi^4\, R_4\biggl[ x^{-6} \log\frac{|x|}{l} \biggr]
+ \frac{9\pi^6}{64}\, \Dl\dl(x);
\notag \\
\word{with scale derivative:} 2\pi^4 \sunset\,.
\label{eq:le-bicicle} % (B.6)
\end{align}

\end{itemize}

For the set:
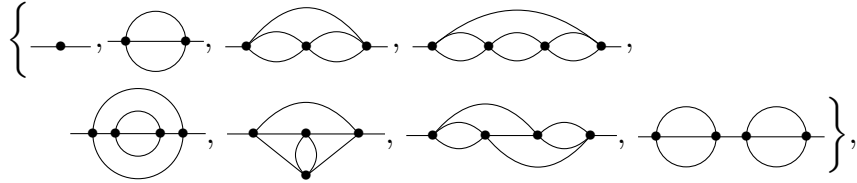
\begin{align*}
& \biggl\{ \justapoint\,, \sunset, \ \goggles\,, \ \onions\,,
\\
&\qquad \saturn, \ \roach, \ \snail, \ \barbells \biggr\},
\end{align*}
the weights are respectively given by $\{1,\frac{1}{6},\frac{1}{4},
\frac{1}{8},\frac{1}{12},\frac{1}{4},\frac{1}{4},\,\frac{1}{36}\}$.
Therefore, up to four loops, with $r = |x_1 - x_2|$, the two-point
function $\sG^2(x_1,x_2)$ is of the form:
\begin{align}
& - \Dl\dl(r) - \frac{(16\pi^2\g)^2}{(4\pi^2)^3}\frac{1}{6}\,\sunset
+ \frac{(16\pi^2\g)^3}{(4\pi^2)^5} \frac{1}{4}\, \goggles 
- \frac{(16\pi^2\g)^4}{(4\pi^2)^7} \Biggl( \frac{1}{8} \onions
\notag \\
&\hspace*{6em} + \frac{1}{12} \saturn + \frac{1}{4} \roach 
+ \frac{1}{4} \snail +  \frac{1}{36} \barbells \Biggr) 
\notag \\
&= - \Dl\dl(r) - \g^2 \frac{2}{3\pi^2}\, R_4[r^{-6}]
+ \g^3 \biggl( \frac{4}{\pi^2}\,R_4\Bigl[ r^{-6}\log\frac{r}{l} \Bigr]
- \frac{1}{8}\, \Dl\dl(r) \biggr) 
\notag \\
&\quad - \g^4 \Biggl[ \biggl( \frac{6}{\pi^2}\,
R_4\Bigl[ r^{-6} \log^2\frac{r}{l} \Bigr] 
- \frac{3}{2\pi^2}\, R_4[r^{-6}] + \frac{2\zeta(3) - 1}{8}\,\Dl\dl(r)
\biggr)
\notag \\
&\hspace*{4em} - \biggl( 
\frac{1}{3\pi^2}\, R_4\Bigl[ r^{-6} \log\frac{r}{l}\Bigr] 
+ \frac{5}{12\pi^2}\, R_4[r^{-6}] \biggr)
\notag \\
&\hspace*{4em} + \biggl( 
\frac{4}{\pi^2}\, R_4\Bigl[ r^{-6} \log^2\frac{r}{l} \Bigr]
+ \frac{4}{\pi^2}\, R_4\Bigl[ r^{-6} \log\frac{r}{l} \Bigr]
- \frac{2}{\pi^2}\, R_4[r^{-6}] - \frac{1}{8}\, \Dl(r) \biggr)
\notag \\
&\hspace*{4em} + \biggl( 
\frac{8}{\pi^2}\, R_4\Bigl[ r^{-6} \log^2\frac{r}{l} \Bigr]
+ \frac{4}{\pi^2}\, R_4\Bigl[ r^{-6} \log\frac{r}{l} \Bigr] 
- \frac{1}{\pi^2}\, R_4[r^{-6}] 
+ \frac{4\zeta(3) - 3}{8}\, \Dl\dl(r) \biggr)
\notag \\
&\hspace*{4em} + \biggl( 
- \frac{2}{9\pi^2}\, R_4\Bigl[ r^{-6} \log\frac{r}{l} \Bigr]
+ \frac{1}{64}\, \Dl(x) \biggr) \Biggl] + \, O(\bar g^5).
\label{eq:dos-puntos} % (B.7)
\end{align}
The right hand side of~\eqref{eq:dos-puntos} may be abbreviated as
$$
\sG^2_\punto + \sG^2_\ocaso + \sG^2_\gafas + \sG^2_\cebollas
+ \sG^2_\saturno + \sG^2_\cuca + \sG^2_\caracol + \sG^2_\pesas
+ O(\bar g^5).
$$
We see that the practical rule to go from the calculated scale
derivatives of the graphs to their contributions to the two-point
function is as before: multiply the coefficient of the scale
derivative by $-\g/\pi^2$ raised to a power equal to the difference in
the number of vertices, and also by the relative weight, for each
diagram in question. An exception is the case of the sunset graph,
which, on account of~\eqref{eq:dos-puntos}
and~\eqref{eq:ren-hex-scale}, fulfils:
\begin{equation}
\lddl \,\sG^2_\ocaso  = -\frac{\g^2}{6} \,\sG^2_\punto \,.
\label{eq:evening-star} % (B.7)
\end{equation}

The tableau of scale derivatives for the two-point function, on
account of \eqref{eq:clear-sight}--%
% \eqref{eq:ripe-onions}, \eqref{eq:baby-eater}, \eqref{eq:gimlet},
% \eqref{eq:escargot},
\eqref{eq:le-bicicle} and \eqref{eq:evening-star} is as follows.
\begin{align}
\lddl\, \sG^2_\gafas &= 6\g \,\sG^2_\ocaso .
\notag \\
\lddl\, \sG^2_\cebollas &= 3\g \,\sG^2_\gafas .
\notag \\
\lddl\, \sG^2_\saturno
&= \frac{1}{2}\,\g^2 \,\sG^2_\ocaso 
+ \frac{5}{48}\,\g^4 \,\sG^2_\punto .
\notag \\
\lddl\, \sG^2_\cuca 
&= 2\g \,\sG^2_\gafas - 6\g^2 \,\sG^2_\ocaso 
+ \frac{1}{4}\,\g^4 \,\sG^2_\punto .
\notag \\
\lddl\, \sG^2_\caracol 
&= 4\g \,\sG^2_\gafas - 6\g^2 \,\sG^2_\ocaso 
- \frac{1}{4}\,\g^4 \,\sG^2_\punto .
\notag \\
\lddl\, \sG^2_\pesas &= \frac{1}{3}\,\g^2\sG_\ocaso .
\label{eq:punto-final} % (B.9)
\end{align}

% \S C
\section{Radial extensions in $p$-space and momentum amplitudes} 
\label{app:pspace-extens}

% \S C.1
\subsection{Fourier transforms} 
\label{sap:Fourier-trans}

Let now $\sF$ denote the Fourier transformation on $\sS'(\bR^d)$, with a
standard convention
$$
\sF\vf(p) = \int e^{-ip\cdot x} \vf(x) \,dx,  \word{so that, e.g.,}
\sF[e^{-r^2/2}] = (2\pi)^{d/2} e^{-p^2/2}.
$$
A standard calculation \cite{GelfandS64} gives
$$
\sF[r^\la] = 2^{\la+d} \pi^{d/2}\, 
\frac{\Ga(\half(\la + d))}{\Ga(-\half\la)}\, |p|^{-d-\la},
$$
valid for $-d < \Re\la < 0$, where both sides are locally integrable
functions. In particular, $\sF[r^{-2}] = 4\pi^2\,p^{-2}$ when $d = 4$.

The Fourier transforms of the radial distributions 
$R_d\bigl[ r^{-d} \log^m(r/l) \bigr]$ are found as follows. Since
$\sF(\dl) = 1$ and $\sF(\Dl h) = -p^2\,\sF h$, it is enough to compute
$\sF\bigl[ r^{-d+2} \log^k(r/l) \bigr]$ for $k \leq m$, on account of 
Lemma~\ref{lm:many-log-ren}. For simplicity, we do it here only for
$d = 4$.

\begin{lema} % 4
\label{lm:many-log-p}
Write $\La := 2/le^\ga$, where $\ga$ is Euler's constant. Then, for
$d = 4$, we get
\begin{align}
\sF\Bigl[ r^{-2} \log^k \frac{r}{l} \Bigr]
&= 4\pi^2 \sum_{j=0}^k (-)^k C_{jk}\, p^{-2} \log^j \frac{|p|}{\La}\,;
\notag \\
\sF\Bigl[ \Dl\Bigl( r^{-2} \log^k \frac{r}{l} \Bigr) \Bigr]
&= 4\pi^2 \sum_{j=0}^k (-)^{k+1} C_{jk}\, \log^j \frac{|p|}{\La}\,,
\label{eq:many-log-p} % (C.1)
\end{align}
for suitable nonnegative constants $C_{jk}$.
\end{lema}

\begin{proof}
If $|t| < 1$, we obtain
$$
\sF[r^{-2} (r/l)^{2t}] = \frac{4\pi^2}{p^2} 
\,\frac{\Ga(1 + t)}{\Ga(1 - t)} \biggl( \frac{l|p|}{2} \biggr)^{-2t}.
$$
Then the known Taylor series expansion \cite[Thm.~10.6.1]{BorosM04},
valid for $|t| < 1$:
\begin{equation}
\log \Ga(1 + t) 
= -\ga t + \sum_{k=2}^\infty (-)^k \frac{\zeta(k)}{k}\, t^k
\label{eq:log-gamma} % (C.2)
\end{equation}
suggests rewriting the previous equality as
\begin{align}
\sF[r^{-2} (r/l)^{2t}]
&= \frac{4\pi^2}{p^2} \biggl( \frac{|p|}{\La} \biggr)^{-2t}
\,e^{2\ga t} \,\frac{\Ga(1 + t)}{\Ga(1 - t)} 
\notag \\
&= \frac{4\pi^2}{p^2} \exp\biggl\{ -2t \log \frac{|p|}{\La}
- 2 \sum_{m=1}^\infty \frac{\zeta(2m+1)}{2m + 1}\, t^{2m+1} \biggr\}.
\label{eq:hard-to-resist} % (C.3)
\end{align}

Differentiation at $t = 0$ then yields $\sF[r^{-2}] = 4\pi^2\,p^{-2}$,
already noted, as well as
\begin{align}
\sF\Bigl[ r^{-2} \log \frac{r}{l} \Bigr]
&= - 4\pi^2 p^{-2} \log \frac{|p|}{\La} \,,
\nonumber \\
\sF\Bigl[ r^{-2} \log^2 \frac{r}{l} \Bigr]
&= 4\pi^2 p^{-2} \log^2 \frac{|p|}{\La} \,,
\nonumber \\
\sF\Bigl[ r^{-2} \log^3 \frac{r}{l} \Bigr]
&= - 4\pi^2 p^{-2} \biggl( \log^3 \frac{|p|}{\La}
+ \frac{1}{2}\, \zeta(3) \biggr),
\label{eq:log-roll} % (C.4)
\\
\sF\Bigl[ r^{-2} \log^4 \frac{r}{l} \Bigr]
&= 4\pi^2 p^{-2} \biggl( \log^4 \frac{|p|}{\La} 
+ 2\zeta(3) \log \frac{|p|}{\La} \biggr),
\nonumber \\
\sF\Bigl[ r^{-2} \log^5 \frac{r}{l} \Bigr]
&= - 4\pi^2 p^{-2} \biggl( \log^5 \frac{|p|}{\La}
+ 5\,\zeta(3) \log^2 \frac{|p|}{\La} + \frac{3}{2}\, \zeta(5) \biggr),
\nonumber
\end{align}
and so on, using the Fa\`a di Bruno formula. The corresponding formula
to~\eqref{eq:log-roll} in \cite{FreedmanJL92} is in error.
In~\cite{Schnetz97} the correct value does appear.
\end{proof}

Conversely, the first few renormalized log-homogeneous distributions
of Section~\ref{ssc:log-homog} have the following Fourier transforms
in $\sS'(\bR^4)$:
\begin{align*}
\sF R_4 \bigl[ r^{-4} \bigr]
&= - 2\pi^2 \log \frac{|p|}{\La} + \pi^2,
\\
\sF R_4 \Bigl[ r^{-4} \log \frac{r}{l} \Bigr]
&= \pi^2 \log^2 \frac{|p|}{\La} - \pi^2 \log \frac{|p|}{\La} 
+ \frac{\pi^2}{2} \,,
\\
\sF R_4 \Bigl[ r^{-4} \log^2 \frac{r}{l} \Bigr]
&= - \frac{2\pi^2}{3} \log^3 \frac{|p|}{\La}
+ \pi^2 \log^2 \frac{|p|}{\La} - \pi^2 \log \frac{|p|}{\La}
+ \pi^2 \biggl( \frac{1}{2} - \frac{\zeta(3)}{3} \biggr),
\\
\sF R_4 \Bigl[ r^{-4} \log^3 \frac{r}{l} \Bigr]
&= \frac{\pi^2}{2} \log^4 \frac{|p|}{\La}
- \pi^2 \log^3\frac{|p|}{\La} + \frac{3\pi^2}{2} \log^2\frac{|p|}{\La}
\\
&\qquad - \pi^2 \biggl( \frac{3}{2} - \zeta(3) \biggr)
\log \frac{|p|}{\La}
+ \pi^2 \biggl( \frac{3}{4} - \frac{\zeta(3)}{2} \biggr).
\end{align*}

The quadratically divergent graphs for the two-point function required
here possess the transforms:
\begin{align}
\sF R_4[r^{-6}]
&= \frac{\pi^2}{4}\, p^2 \log\frac{|p|}{\La} - \frac{5\pi^2}{16}\,p^2,
\notag \\
\sF R_4 \Bigl[ r^{-6} \log \frac{r}{l} \Bigr]
&= - \frac{\pi^2}{8}\, p^2 \log^2 \frac{|p|}{\La}
+ \frac{5\pi^2}{16}\, p^2 \log \frac{|p|}{\La}
- \frac{17\pi^2}{64}\, p^2;
\label{eq:rminus6-logs} % (C.5)
\\
\sF R_4 \Bigl[ r^{-6} \log^2 \frac{r}{l} \Bigr]
&= \frac{\pi^2}{12}\, p^2 \log^3 \frac{|p|}{\La}
- \frac{5\pi^2}{16}\, p^2 \log^2 \frac{|p|}{\La}
+ \frac{17\pi^2}{32}\, p^2 \log \frac{|p|}{\La}
- \pi^2 \Bigl( \frac{49}{128} - \frac{\zeta(3)}{24} \Bigr) p^2,
\notag
\end{align}
in view of \eqref{eq:ren-hex} and~\eqref{eq:ren-hex-log}. The first
identity here gives the sunset graph in momentum space. The second one
gives essentially the ``goggles'' graph \eqref{eq:clear-sight} of
Appendix~\ref{app:colon}.

It follows by induction from~\eqref{eq:extra-laps} that
$\sF R_4[r^{-4-2m}]$ is a linear combination of $p^{2m}$ and\break
$p^{2m} \log\bigl( |p|/\La \bigr)$. It can be written as
$$
\sF R_4[r^{-4-2m}] = \frac{1}{4^m m!(m+1)!}
\Bigl( a_m\,p^{2m} \log\frac{|p|}{\La} + b_m\,p^{2m} \Bigr).
$$
On Fourier-transforming the case $d = 4$, $n = 1$ 
of~\eqref{eq:extra-laps}, one finds the recurrence relation
$$
- p^2 \Bigl( a_m\,p^{2m} \log\frac{|p|}{\La} + b_m\,p^{2m} \Bigr)
= a_{m+1}\,p^{2m+2} \log\frac{|p|}{\La} + \biggl( b_{m+1}
+ (-)^m \frac{(2m + 3)\pi^2}{(m + 1)(m + 2)} \biggr) p^{2m+2}.
$$
Thus $a_m = (-)^m a_0 = (-)^{m-1} 2\pi^2$ and 
$b_{m+1} + b_m = (-)^m \pi^2\bigl( 1/(m+1) + 1/(m+2) \bigr)$. Since
$b_0 = \pi^2$, that yields $b_m = (-)^m \pi^2 (H_m + H_{m+1})$.
We recover a result already found in~\cite{Carme}:
$$
\sF R_4[r^{-4-2m}] = \frac{(-)^{m-1}\pi^2}{4^m m!(m+1)!} \biggl(
2\,p^{2m} \log\frac{|p|}{\La} - (H_m + H_{m+1}) p^{2m} \biggr).
$$
Analogous formulas for $\sF R_d[r^{-d-2m}]$ can be derived
from~\eqref{eq:extra-laps} in the same way.

% \S C.2
\subsection{On the amplitudes in $p$-space}
\label{sap:pspace}

It is now straightforward to perform the conversion to momentum space
graph by graph; but the details are hardly worthwhile for us, since
our renormalization scheme and consequent treatments of the RG and the
$\bt$-function for the model do not require that conversion.

% \S C.2.1
\subsubsection{Two-point amplitudes in $p$-space}
\label{ssp:pspace-twopt}

The free Green function has Fourier transform
$$
(2\pi)^4 \,\frac{\dl(p_1 + p_2)}{p^2} 
=: (2\pi)^4\,\dl(p_1 + p_2) G_\free(|p|),
$$
where $|p| := |p_1| = |p_2|$ corresponds to the difference variable on
$x$-space. There is a series of corrections to the free propagator:
$$
G(|p|) = \frac{1}{p^2} + \frac{1}{p^2}\,\Sg(|p|)\,\frac{1}{p^2}
+ \frac{1}{p^2}\,\Sg(|p|)\,\frac{1}{p^2}\,\Sg(|p|)\,\frac{1}{p^2}
+ \cdots
$$
Just as on $x$-space, the above equation is solved by
$$
G(|p|) = \bigl( p^2 - \Sg(|p|) \bigr)^{-1}.
$$
What we call the momentum-space two-point function $\sG^2(|p|)$ is
$G(|p|)^{-1} = p^2 - \Sg(|p|)$; this is the Fourier transform
of~$\sG^2(r)$, by the definition of the latter.

\medskip

One can now obtain from the list~\eqref{eq:rminus6-logs} of Fourier
transforms of quadratically divergent graphs, together with
Eqs.~\eqref{eq:ren-hex} and \eqref{eq:clear-sight}--\eqref{eq:gimlet},
the corrections to the propagator associated to these graphs in
momentum space. We omit this trivial conversion; taking account of
different conventions, this coincides in a few cases with results
in~\cite{Schnetz97}. It is remarkable that the Ap\'ery constant
$\zeta(3)$ disappears in the computation of the chain
graph~$\cebollas\,$.

% \S C.2.2
\subsubsection{Four-point amplitudes in $p$-space}
\label{ssp:pspace-fourpt}

The graphs in momentum space relevant for the four-point function are
more complicated. We just exemplify for the fish graph. Still omitting
the permutations of the vertices, the amplitude in $x$-space is of the
form
\begin{align*}
\sG_\pez(x_1,x_2,x_3,x_4) &= \frac{g^2}{(4\pi^2)^2} \biggl[ 
\dl(x_1 - x_2)\,\dl(x_2 - x_3)\,\dl(x_3 - x_4)
\\
&\hspace*{3em} - \frac{1}{2\pi^2} \,\dl(x_1 - x_2)
\,\Dl \Bigl( (x_2 - x_3)^{-2} \log \frac{|x_2 - x_3|}{l} \Bigr)
\,\dl(x_3 - x_4) \biggr]
\\
&= \frac{g^2}{(4\pi^2)^2} \biggl[ \dl(\xi_1)\,\dl(\xi_2)\,\dl(\xi_3)
- \frac{1}{2\pi^2} \,\dl(\xi_1)
\, \Dl \Bigl( \xi_2^{-2} \log \frac{|\xi_2|}{l} \Bigr)
\,\dl(\xi_3) \biggr]
\\
&=: \wh\sG_\pez(\xi_1,\xi_2,\xi_3);
\end{align*}
where we have introduced the difference variables
$\xi_1 = x_1 - x_2$, $\xi_2 = x_2 - x_3$, $\xi_3 = x_3 - x_4$. The
\textit{reduced} Fourier transform $\wh\sG_\pez(p_1,p_2,p_3)$, defined
by
\begin{align*}
\wh\sG_\pez(p_1,p_2,p_3)
&:= \int \wh\sG_\pez(\xi_1,\xi_2,\xi_3)
\\
&\qquad \x 
\exp\bigl[ -i(p_1\xi_1 + (p_1+p_2)\xi_2 + (p_1+p_2+p_3)\xi_3) \bigr]
\,d\xi_1 \,d\xi_2 \,d\xi_3
\end{align*}
yields
$\sG_\pez(p_1,p_2,p_3,p_4) = (2\pi)^4\,\dl(p_1 + p_2 + p_3 + p_4)
\,\wh\sG_\pez(p_1,p_2,p_3)$, where $\sG_\pez$ is the \textit{ordinary}
Fourier transform, defined by:
$$ 
\sG_\pez(p_1,p_2,p_3,p_4) := \int \sG_\pez(x_1,x_2,x_3,x_4) 
\exp[-i(p_1x_1 +\cdots+ p_4x_4)] \,dx_1\,dx_2\,dx_3\,dx_4.
$$

This is general for functions of the difference variables. In our
present case, we obtain
\begin{align*}
\sG_\pez(p_1,p_2,p_3,p_4) &= g^2 \,\dl(p_1 +\cdots+ p_4)
\biggl[ 1 - 2 \log\frac{|p_1 + p_2|}{\La} \biggr]
\\
&= g^2 \,\dl(p_1 +\cdots+ p_4) 
\biggl[ 1 - 2 \log\frac{|p_3 + p_4|}{\La} \biggr].
\end{align*}

\subsection*{Acknowledgments}
At the beginning of this work we enjoyed the warm hospitality of the
ZiF and the Universit\"at Bielefeld. We are most grateful to Philippe
Blanchard, Michael D\"utsch, Ricardo Estrada and Raymond Stora for
many discussions on Epstein--Glaser renormalization and distribution
theory. We thank Stefan Hollands and Ivan Todorov for helpful
comments. JMG-B is thankful to Kurusch Ebrahimi-Fard and Fr\'ed\'eric
Patras for much sharing, over the years, of reflections on
renormalization theory. His work was supported by the Spanish
Ministerio de Educaci\'on y Ciencia through grant FPA2012--35453 and
by the Universidad de Costa Rica through its von~Humboldt chair. HGG
and JCV acknowledge support from the Vicerrector\'ia de
Investigaci\'on of the Universidad de Costa Rica.

\end{document}